\documentclass[aps,pra,twocolumn,nofootinbib, superscriptaddress,10pt]{revtex4-2}
\usepackage{graphicx}
\usepackage{epstopdf}
\usepackage{amsmath}
\usepackage{amssymb}
\usepackage{mathrsfs}
\usepackage{amsthm}
\usepackage{bm}
\usepackage{url}
\usepackage[T1]{fontenc}
\usepackage{csquotes}
\MakeOuterQuote{"}
\usepackage{thmtools, thm-restate}

\newtheoremstyle{note}
{\topsep/2}               % ABOVE SPACE
{\topsep/2}               % BELOW SPACE
{}                      % BODY FONT
{\parindent}            % INDENT (empty value is the same as 0pt)
{\itshape}              % HEAD FONT
{.}                     % HEAD PUNCTUATION
{5pt plus 1pt minus 1pt}% HEAD SPACE
{}

\theoremstyle{note}
\newtheorem{thm}{Theorem}
\newtheorem{lem}{Lemma}
\newtheorem{conjecture}{Conjecture}
\newtheorem{corollary}{Corollary}
\newtheorem{proposition}{Proposition}

\theoremstyle{definition}

\theoremstyle{remark}

\newtheorem{thm*}{Theorem}% This is a fake theorem environment
\makeatletter
\expandafter\newcommand\csname thethm*default\endcsname{\thethm*}
\newcommand{\thmstarnum}[1]{\expandafter\gdef\csname thethm*\endcsname{#1*}}
\thmstarnum{\thethm}
\expandafter\g@addto@macro\csname endthm*\endcsname{\thmstarnum{\thethm}}
\makeatother

\newtheorem{lem*}{Lemma}

\makeatletter
\expandafter\newcommand\csname thelem*default\endcsname{\thelem*}
\newcommand{\lemstarnum}[1]{\expandafter\gdef\csname thelem*\endcsname{#1*}}
\lemstarnum{\thelem}
\expandafter\g@addto@macro\csname endlem*\endcsname{\lemstarnum{\thelem}}
\makeatother

\def\vec#1{\bm{#1}} %% overriding the original command

\newcommand{\tr}{\operatorname{tr}}

\newcommand{\spa}{\operatorname{span}}

\newcommand{\imply}{\mathrel{\Rightarrow}}

\newcommand{\rmd}{\mathrm{d}}
\newcommand{\rme}{\mathrm{e}}
\newcommand{\rmi}{\mathrm{i}}

\newcommand{\rmP}{\mathrm{P}}

\newcommand{\rmT}{\mathrm{T}}

\newcommand{\iid}{\mathrm{iid}}
\newcommand{\sep}{\mathrm{sep}}

\newcommand{\CP}{\mathrm{CP}}

\newcommand{\caH}{\mathcal{H}}

\newcommand{\caQ}{\mathcal{Q}}

\newcommand{\caS}{\mathcal{S}}
\newcommand{\caT}{\mathcal{T}}
\newcommand{\caV}{\mathcal{V}}

\newcommand{\bbE}{\mathbb{E}}

\newcommand{\scrA}{\mathscr{A}}
\newcommand{\scrB}{\mathscr{B}}
\newcommand{\scrC}{\mathscr{C}}

\newcommand{\scrI}{\mathscr{I}}

\newcommand{\sic}{\mathrm{SIC}}
\newcommand{\mub}{\mathrm{MUB}}
\newcommand{\lsp}{\hspace{0.1em}}

\newcommand{\be}{\begin{equation}}
\newcommand{\ee}{\end{equation}}
\newcommand{\ba}{\begin{align}}
\newcommand{\ea}{\end{align}}

\def\<{\langle}  %% overriding the original command \<
\def\>{\rangle}  %% overriding the original command \>

\def\eqref#1{\textup{(\ref{#1})}}  %% overriding the original command \eqref
\newcommand{\eref}[1]{Eq.~\textup{(\ref{#1})}}
\newcommand{\Eref}[1]{Equation~\textup{(\ref{#1})}}

\newcommand{\esref}[2]{Eqs.~\textup{(\ref{#1})} and \textup{(\ref{#2})}}

\newcommand{\Esref}[2]{Equations~\textup{(\ref{#1})} and \textup{(\ref{#2})}}

\newcommand{\ecref}[2]{Eqs.~\textup{(\ref{#1})}-\textup{(\ref{#2})}}

\newcommand{\fref}[1]{Fig.~\ref{#1}}

\newcommand{\fsref}[1]{Figs.~\ref{#1}}

\newcommand{\tref}[1]{Table~\ref{#1}}

\newcommand{\sref}[1]{Sec.~\ref{#1}}
\newcommand{\Sref}[1]{Section~\ref{#1}}

\newcommand{\thref}[1]{Theorem~\ref{#1}}
\newcommand{\Thref}[1]{Theorem~\ref{#1}}
\newcommand{\thsref}[1]{Theorems~\ref{#1}}

\newcommand{\lref}[1]{Lemma~\ref{#1}}
\newcommand{\Lref}[1]{Lemma~\ref{#1}}
\newcommand{\lsref}[1]{Lemmas~\ref{#1}}

\newcommand{\pref}[1]{Proposition~\ref{#1}}
\newcommand{\Pref}[1]{Proposition~\ref{#1}}

\newcommand{\crref}[1]{Corollary~\ref{#1}}
\newcommand{\Crref}[1]{Corollary~\ref{#1}}
\newcommand{\crsref}[1]{Corollaries~\ref{#1}}
\newcommand{\Crsref}[1]{Corollaries~\ref{#1}}

\newcommand{\cref}[1]{Conjecture~\ref{#1}}
\newcommand{\Cref}[1]{Conjecture~\ref{#1}}

\newcommand{\aref}[1]{Appendix~\ref{#1}}

\newcommand{\rcite}[1]{Ref.~\cite{#1}}
\newcommand{\rscite}[1]{Refs.~\cite{#1}}

\begin{document}
\title{Quantum Measurements in the Light of  Quantum State Estimation}

\author{Huangjun Zhu}

\email{zhuhuangjun@fudan.edu.cn}

\affiliation{State Key Laboratory of Surface Physics and Department of Physics, Fudan University, Shanghai 200433, China}

\affiliation{Institute for Nanoelectronic Devices and Quantum Computing, Fudan University, Shanghai 200433, China}

\affiliation{Center for Field Theory and Particle Physics, Fudan University, Shanghai 200433, China}

\begin{abstract}
Starting from a simple estimation problem, here we propose a general approach for decoding quantum measurements from the perspective of information extraction. By virtue of the estimation fidelity only, we provide surprisingly simple characterizations of  rank-1 projective measurements, mutually unbiased measurements, and symmetric informationally complete measurements.  
Notably, our conclusions do not rely on
 any assumption on the rank, purity, or the number of measurement outcomes, and we do not need bases to start with. Our work demonstrates that all these  elementary quantum measurements are uniquely determined by their information-extraction capabilities, which are not even anticipated before. In addition, we offer a new perspective for understanding noncommutativity and incompatibility from tomographic performances, which also leads to a universal criterion for detecting quantum incompatibility.
Furthermore, we show that the estimation fidelity can be used to distinguish inequivalent mutually unbiased bases and symmetric informationally complete measurements.
 In the course of  study, we introduce the concept of (weighted complex projective) $1/2$-designs and show that all $1/2$-designs are tied to symmetric informationally complete measurements, and vice versa. 
\end{abstract}

\date{\today}
\maketitle

%\tableofcontents

\section{Introduction} Quantum measurements are a basic tool for extracting information from quantum systems and a bridge for connecting the quantum world with the classical world \cite{Neum55,NielC10book,BuscLPY16}. They also play  indispensable  roles in almost all quantum information processing tasks, such as quantum computation, quantum communication, quantum metrology, quantum sensing,  quantum simulation, and quantum characterization, verification, and validation (QCVV). Although there are numerous works on quantum measurements, the mysteries about quantum measurements have never been fully explored, even for the simplest quantum measurements.

 Prominent examples of quantum measurements include \emph{rank-1 projective measurements}, \emph{mutually unbiased measurements} (MUMs) based on  \emph{mutually unbiased bases} (MUB) \cite{Schw60,Ivan81,WootF89,DurtEBZ10,BengZ17book}, and \emph{symmetric informationally complete measurements} (SICs for short) \cite{Zaun11,ReneBSC04,ScotG10,FuchHS17,BengZ17book}. These quantum  measurements stand out because of their crucial roles in foundational studies and practical quantum information processing. Notably, rank-1 projective measurements are the canonical quantum measurements discussed in most elementary textbooks on quantum mechanics. MUMs are tied to the complementarity principle \cite{Bohr28}, uncertainty relations \cite{Heis27,Robe29,BuscLW14,WehnW10,ColeBTW17}, and are useful in quantum state estimation \cite{Ivan81,WootF89,DurtEBZ10,RoyS07,Zhu14IOC,AdamS10} and quantum cryptography \cite{BennB84,DurtEBZ10,ColeBTW17}. SICs play a crucial role in connecting the Born rule with the law of total probability and in the Bayesian interpretation of quantum theory \cite{FuchS13,ApplFSZ17}; SICs are also useful in constructing quasiprobability representations with minimal negativity \cite{Zhu16Q} and  in quantum state estimation \cite{Scot06,ZhuE11,Zhu12the,ZhuH18U}. In addition, the rich mathematical structures underlying MUB and SICs are a source of inspiration and have  attracted the attention of numerous researchers; see \rscite{DurtEBZ10,FuchHS17,BengZ17book,HoroRZ22} for reviews.

All the quantum measurements mentioned above have very simple  algebraic descriptions in the language of positive operator-valued measures (POVMs) \cite{NielC10book,BuscLPY16}. However, such algebraic descriptions lack clear operational meanings beyond the Born rule. 
Notably, the information theoretical significance of these measurements is far from being clear despite the efforts of many researchers. This awkward situation is in sharp contrast with the rapid development of quantum information science. Now, it is natural to ask if these measurements can be characterized by  simple tasks in quantum information processing.
What is so special about rank-1 projective measurements from the perspective of information extraction? How about other elementary quantum measurements, such as MUMs and SICs?

In this work, we propose a general approach for decoding quantum measurements from a simple and well-studied estimation problem: estimation of Haar random pure states \cite{MassP95,DerkBE98,LatoPT98,Haya98,BrusM99,GisiP99,Mass00,AcinLP00,Bana01,BagaBM02,HayaHH05}. Here Haar random pure states can also be replaced by certain discrete sets, which are amenable to experiments. 
By virtue of tomographic performances as quantified by the estimation fidelity, we provide surprisingly simple operational characterizations of various typical and important quantum measurements, including  rank-1 projective measurements, MUMs, and SICs. Remarkably, our characterizations do not need any assumption on the rank, purity, or the number of measurement outcomes, and we do not need bases to start with. Our work demonstrates that all these elementary quantum measurements are uniquely determined by their information-extraction capabilities and therefore can be defined in purely information theoretic terms, in sharp contrast with traditional algebraic definitions, which lack clear operational meanings.

In addition, we offer a new perspective for understanding noncommuting and incompatible measurements \cite{Busc86,HeinoMZ16,BuscLPY16,GuhnHKP21} from tomographic performances. Notably, we show that incompatibility is a resource rather than a limitation to enhance the estimation fidelity. Moreover, we  prove a tight upper bound for the two-copy estimation fidelity based on compatible measurements, which reveals an intriguing connection between quantum incompatibility and SICs and also provides a universal criterion for detecting quantum incompatibility.
The connection with entropic uncertainty relations \cite{GhirMR03,WehnW10,ColeBTW17} is also discussed briefly. 
Furthermore, our work leads to a simple operational approach for distinguishing inequivalent MUB and SICs, which cannot be distinguished by inspecting pairwise overlaps alone.  The approach we introduce is also very useful to studying other discrete symmetric structures tied to the quantum state space. Moreover, all these results are amenable to experimental demonstration with current technologies.

In the course of study, we derive a number of results on quantum measurements and (weighted complex projective) $t$-designs \cite{Hogg82,Zaun11,ReneBSC04,Scot06}, which are of  interest beyond the main focus of this work. Notably, we introduce the concept of $1/2$-designs and show that SICs are essentially the only  $1/2$-designs. This result may shed some light on the search for general "fractional designs", although this is not the  focus of this work. In addition,
 we introduce the concept of cross frame potential, which is surprisingly useful to studying typical quantum measurements and discrete symmetric structures tied to the quantum state space. 
 Furthermore, we establish a simple connection between the estimation fidelity and the $t$th frame potential with $t=1/2$ and thereby clarifying  the operational significance of this frame potential. Our work may have implications for a number of active research areas, including quantum measurements, quantum estimation theory, geometry of quantum states, $t$-designs, and foundational studies on quantum incompatibility and steering.

 The rest of this paper is organized as follows. In \sref{sec:QM} we first introduce basic concepts on quantum measurements and an order relation based on data processing;  then we derive several results on rank-1 projective measurements and MUMs. In \sref{sec:POVMtdesign} we discuss the connections between $t$-designs and quantum measurements and explore the applications of a special frame potential. In \sref{sec:decoding} we propose a general approach for decoding quantum measurements based on a simple estimation problem. In \sref{sec:TypicalQM} by virtue of the estimation fidelity we offer surprisingly simple characterizations of various typical quantum measurements. In \sref{sec:incompatibility} we explore the connections between the estimation fidelity and quantum incompatibility. In \sref{sec:IneqMUBSIC} we provide an operational approach for distinguishing inequivalent MUB and SICs. \Sref{sec:summary} summarizes this paper. To streamline the presentation of the main results,  technical proofs are relegated to the appendices.

\section{\label{sec:QM}Quantum measurements}   
\subsection{Basic concepts}
Let $\caH$ be a  $d$-dimensional Hilbert space  associated with the quantum system under consideration. Quantum states on $\caH$ are usually represented by density operators, which are positive (semidefinite) operators of trace~1. Quantum measurements on $\caH$ are basic tools for extracting information from the quantum system as encoded in the quantum state. In this work we are interested in the information-extraction capabilities of quantum measurements, but not  the post-measurement quantum states. In this context,  a quantum measurement on $\caH$  can be described  by a POVM, which is composed of a set (or collection) of positive operators on $\caH$, usually called  POVM elements, that sum up to the identity operator \cite{NielC10book,BuscLPY16}. Here we use the same notation for the identity operator as the number~1 to simplify the notation; in addition, numbers in operator equations are implicitly multiplied by the identity operator. 
 
Let $\rho$ be a quantum state on $\caH$ and  $\scrA=\{A_j\}_j$ a POVM  on $\caH$. If we perform the POVM $\scrA$ on  $\rho$, then the probability $p_j$ of obtaining outcome $j$ reads $p_j=\tr(\rho A_j)$ according to the Born rule. The POVM $\scrA$ is informationally complete (IC) if its POVM elements span the whole operator space on $\caH$ \cite{Prug77,Scot06,ZhuE11}. This condition guarantees that  any quantum state on $\caH$ can be reconstructed accurately  from  frequencies of measurement outcomes as long as the POVM can be performed sufficiently many times. For comparison, the  POVM  $\scrA$ is trivial if all POVM elements are  proportional to the identity operator, in which case no information can be extracted by performing the POVM.  The POVM $\scrA$ is rank 1 if each POVM element is proportional to a  rank-1 projector. The POVM $\scrA$ is unbiased if all POVM elements have the same trace, in which case the completely mixed state will yield a flat probability distribution when the POVM is performed. Suppose $\scrA$ is an unbiased rank-1 POVM; then $\scrA$ is equiangular  if all the pairwise overlaps $\tr(A_j A_k)$ for $j\neq k$ are equal.

\subsection{\label{sec:Order}An order relation and simple POVMs} 
The idea of data  (information) processing leads to a natural order relation on POVMs. 
  Let $\scrA=\{A_j\}_j$ and $\{B_k\}_k$ be two POVMs on $\caH$. The POVM $\scrA$ is a \emph{coarse graining} of $\scrB$, denoted by $\scrA\preceq \scrB$ or $\scrB\succeq \scrA$,  if $\scrA$ can be constructed  from $\scrB$ by  data processing \cite{MartM90,Zhu15IC,ZhuHC16}. More specifically, $\scrA\preceq \scrB$ if the POVM elements of $\scrA$  can be expressed as 
  \begin{align}\label{eq:CoarseGraining}
A_j=\sum_k\Lambda_{jk}B_k\quad \forall j,
  \end{align}
where  $\Lambda$ is  a stochastic matrix, which satisfies the normalization condition $\sum_{j}\Lambda_{jk}=1$. In this case the measurement statistics of $\scrA$ can be simulated by performing $\scrB$ and then applying suitable data processing. Alternatively, we also say $\scrB$ \emph{refines} $\scrA$ or $\scrB$ is a \emph{refinement} of $\scrA$. Intuitively, coarse graining can never lead to information gain, while refinement can never lead to information loss.

Two POVMs are \emph{equivalent} if they are coarse graining of each other (note the distinction from unitary equivalence); such POVMs are essentially the same from the perspective of information extraction.
A coarse graining or refinement of a POVM $\scrA$ is trivial  (nontrivial) if it is (not) equivalent to $\scrA$. To clarify when a coarse graining is nontrivial, we need to introduce a special function on POVMs. The  \emph{purity} of a POVM $\scrA=\{A_j\}_j$ \cite{Zhu14T} is defined as 
\begin{align}
\wp(\scrA)=\sum_j\frac{1}{d}\frac{\tr(A_j^2)}{\tr A_j}=\sum_j\frac{\tr A_j}{d}\frac{\tr(A_j^2)}{(\tr A_j)^2},
\end{align}
where $d$ is the dimension of the underlying Hilbert space, and the summation runs over nonzero POVM elements in $\scrA$. From this definition it is easy to verify that 
\begin{align}
\frac{1}{d}\leq \wp(\scrA)\leq 1;
\end{align}
the lower bound is saturated iff all POVM elements are proportional to the identity operator, so that the POVM is trivial, while the upper bound is saturated iff all nonzero POVM elements are rank 1, so that the POVM is rank 1.

\begin{lem}\label{lem:EquivalentPOVM}
Suppose $\scrA=\{A_j\}_j$ is a coarse graining of $\scrB=\{B_k\}_k$ as defined in \eref{eq:CoarseGraining}. Then $\wp(\scrA)\leq  \wp(\scrB)$, and the following three statements are equivalent:
\begin{enumerate}
\item $\scrA$ is equivalent to $\scrB$;
\item  $\wp(\scrA)= \wp(\scrB)$;

\item $\Lambda_{jk}\Lambda_{jl}=0$ whenever $B_k,B_l$ are  linearly independent.
\end{enumerate}
\end{lem}
\Lref{lem:EquivalentPOVM} is proved in \aref{asec:POVMorder}. It shows that a coarse graining is trivial iff it only mixes POVM elements that are proportional to each other in addition to  the zero POVM element. \Lref{lem:EquivalentPOVM} also shows that the purity is a strict order-monotonic function \cite{ZhuHC16}. Such functions are useful not only  to studying quantum incompatibility, but also to studying quantum steering \cite{ZhuHC16,Zhu15IC,HeinJN22}.

A POVM  is \emph{simple} if no POVM element is proportional to another POVM element, that is, all POVM elements are pairwise linearly independent. By definition a simple POVM has no POVM element that is equal to the zero operator. The following result was originally proved in \rcite{MartM90} (see also \rcite{Kura15}); it is also a simple corollary of \lref{lem:EquivalentPOVM} as shown in \aref{asec:POVMorder}.

\begin{lem}\label{lem:SimplePOVM}
	Two simple POVMs are equivalent iff they are identical up to relabeling. 	
	Every POVM is equivalent to a unique simple POVM up to relabeling. 
\end{lem}

Restriction to simple POVMs is quite helpful to avoiding unnecessary complications, but usually does not cause any loss of generality. For example, all results on simple POVMs derived in this work can easily be extended to  general POVMs with minor modifications. Nevertheless, nonsimple POVMs are occasionally useful in technical analysis, so we do not assume that all POVMs are simple.  In the rest of this paper instead we take the weaker assumption that no POVM element is equal to the zero operator unless  stated otherwise.

A POVM is \emph{maximal} if every refinement is equivalent to itself. Suppose $\scrA$ and $\scrB$ are equivalent POVMs; then $\scrA$ is maximal iff $\scrB$ is maximal. The following proposition is a variant of a result proved in \rcite{MartM90}, which characterizes the set of rank-1 POVMs via the order relation based on data processing. It is also a simple corollary of \lref{lem:EquivalentPOVM}.
\begin{proposition}\label{pro:MaximalPOVM}
A POVM is maximal iff it is rank-1.
\end{proposition}
 \Lref{lem:EquivalentPOVM} and \pref{pro:MaximalPOVM} show that 	every refinement of a rank-1 POVM is equivalent to the POVM; in other words, any rank-1 POVM 	has no nontrivial refinement. In addition, a rank-1 POVM cannot be equivalent to any POVM that is not rank 1. These observations lead to  the following proposition.
\begin{proposition}\label{pro:CoarseGraining}
Suppose $\scrA$ is a coarse graining of a rank-1 POVM $\scrB$. Then  $\scrA$ is equivalent to (is a trivial coarse graining of) $\scrB$ iff $\scrA$ is rank 1.
\end{proposition}

\subsection{Quantum incompatibility}
Let $\scrA$ and $\scrB$  be two arbitrary POVMs on $\caH$. Then $\scrA$ and $\scrB$
commute if all POVM elements in $\scrA$ commute with all POVM elements in $\scrB$.  This definition also applies to two sets of positive operators.
By contrast,  $\scrA$ and $\scrB$ are \emph{compatible} or jointly measurable if they admit a common refinement  \cite{HeinoMZ16,Busc86,GuhnHKP21,QuinVB14,UolaBGP15,Zhu15IC}. In that case, the measurement statistics of both $\scrA$ and $\scrB$ can be simulated by performing the common refinement.
Otherwise, $\scrA$ and $\scrB$ are \emph{incompatible}.
Generalizations to three or more POVMs are immediate.  Note that commuting POVMs are automatically compatible, but not vice versa in general.
By definition the compatibility relation is closely tied to the order relation discussed in \sref{sec:Order}. This connection is very useful to detecting quantum incompatibility \cite{HeinoMZ16,Zhu15IC,ZhuHC16,HeinJN22} (cf. \sref{sec:incompatibility}).

\begin{proposition}\label{pro:CompatibleRank1}
Suppose $\scrA$ and $\scrB$ are two POVMs on $\caH$ with $\scrB$ being  rank 1. Then $\scrA$ and $\scrB$ are compatible iff $\scrA$ is a coarse graining of $\scrB$.
\end{proposition}

\begin{proposition}\label{pro:CompatibleRank11}
Two rank-1 POVMs are compatible iff they are equivalent. Two simple rank-1 POVMs are compatible iff they are identical up to relabeling. 
\end{proposition}
\Pref{pro:CompatibleRank1} is a simple corollary of \pref{pro:MaximalPOVM}. \Pref{pro:CompatibleRank11} is a simple corollary of \pref{pro:CompatibleRank1} and \lref{lem:SimplePOVM}.

\subsection{Projective measurements and mutually unbiased measurements}
A POVM is \emph{reducible} if its POVM elements can be divided into two  groups such that each group contains at least one nonzero POVM element and
all POVM elements in one group are orthogonal to all POVM elements in the other group (cf. \rcite{Zhu21Z}). In this case, the POVM is a direct sum of two POVMs.
Notably, any POVM containing a projector that is not equal to the identity or the zero operator is reducible; note that such a projector is necessarily orthogonal to all other POVM elements.  A POVM  is \emph{irreducible} if it is not reducible; such a POVM cannot be expressed as a direct sum of two POVMs.

A projective measurement  (also known as a von Neumann measurement) is a special POVM in which all the  POVM elements are  mutually orthogonal projectors and is thus reducible except for the trivial projective measurement. It is usually  characterized by a Hermitian operator  via spectral decomposition. Rank-1 projective measurements are special projective measurements in which all POVM elements are mutually orthogonal rank-1 projectors. They are associated with nondegenerate
Hermitian operators and are the canonical example of quantum measurements as discussed in most textbooks \cite{Neum55}. In addition, they are in one-to-one correspondence with orthonormal bases if we identify bases that differ only by overall phase factors. In view of the crucial roles played by rank-1 projective measurements, here we summarize
their main characteristics  that are useful in the current study. The detailed proofs are relegated to \aref{asec:ProjectiveMU}. 
 \begin{lem}\label{lem:rank1Projective}
Any simple rank-1 POVM $\scrA$ on $\caH$ has at least $d$ POVM elements and satisfies the inequality $\dim(\spa(\scrA))\geq d$. 
	Each  bound is saturated iff $\scrA$ is a rank-1 projective measurement.
\end{lem}

\begin{lem}\label{lem:commutePOVMrank1Proj}
	Suppose  $\scrA$ is a simple rank-1 POVM on $\caH$ and $\scrB$ is a set of distinct rank-1 projectors on $\caH$. Then $\scrA$ and $\scrB$ commute iff $\scrB\subseteq\scrA$ and the projectors in $\scrB$ are mutually orthogonal. 	
\end{lem}
Note that every rank-1 positive operator on $\caH$ is proportional to a rank-1 projector. As an implication of  \lref{lem:commutePOVMrank1Proj} and this observation, if a simple rank-1 POVM $\scrA$ commutes with a nonempty set of pairwise linearly independent rank-1 positive operators (say some POVM elements in another simple rank-1 POVM), then these positive operators must be mutually orthogonal, and $\scrA$ contains a set of rank-1 projectors that are proportional to these rank-1 positive operators, respectively. In this case, the POVM $\scrA$  is a direct sum of a rank-1 projective measurement and another POVM and is thus  reducible, assuming that the underlying Hilbert space $\caH$  has dimension at least 2.

\begin{lem}\label{lem:commutePOVMrank1}
	Two simple rank-1 POVMs commute iff they are identical rank-1 projective measurements up to relabeling. 
\end{lem}
\Lref{lem:commutePOVMrank1}	is a simple corollary of 
\lref{lem:commutePOVMrank1Proj}; a direct proof is presented in  \aref{asec:ProjectiveMU}. As an implication of \lref{lem:commutePOVMrank1}, 
any simple rank-1 POVM that commutes with itself is a  rank-1 projective measurement.

Two orthonormal bases $\{|\psi_j\>\}_{j=1}^d$ and $\{|\varphi_k\>\}_{k=1}^d$ for $\caH$ 
are \emph{mutually unbiased} (MU) or complementary if all  the transition probabilities  $|\<\psi_j|\varphi_k\>|^2$ are equal to $1/d$. In this case,
the corresponding measurements are also  referred to as MU and are often regarded as maximally incompatible \cite{Schw60,Ivan81,WootF89,DurtEBZ10,BengZ17book,DesiSFB19}. 
Such measurements are quite useful in many tasks in quantum information processing, including quantum state estimation \cite{Ivan81,WootF89,DurtEBZ10,RoyS07,Zhu14IOC,AdamS10} and quantum cryptography \cite{BennB84,DurtEBZ10,ColeBTW17} in particular. 
It is known that the number of bases in any MUB cannot surpass $d+1$; when the upper bound is saturated, the MUB is called a complete set of MUB (CMUB), and the corresponding set of measurements is called a complete set of MUMs (CMUMs).

As a generalization, two positive operators $A$ and $B$ on $\caH$ are MU if $\tr(A B)=\tr(A)\tr(B)/d$.
Two POVMs $\{A_j\}_j$ and $\{B_k\}_k$ on $\caH$ are MU if each POVM element in $\scrA$ and each POVM element in $\scrB$ are MU, that is,  $\tr(A_j B_k)=\tr(A_j)\tr(B_k)/d$. The following theorem sets an upper bound for the number of rank-1 POVMs that are MU, which is reminiscent of the upper bound for MUB \cite{Ivan81,WootF89,DurtEBZ10}.
\begin{thm}\label{thm:MUPOVM}
Any set of MU simple rank-1  POVMs on $\caH$ contains at most $d+1$ POVMs. If the upper bound is saturated, then  all the POVMs in the set are rank-1 projective measurements, which form a CMUMs. 
\end{thm}

At this point, it is worth pointing out that the assumption of simplicity of POVMs in \lsref{lem:rank1Projective}-\ref{lem:commutePOVMrank1} and \thref{thm:MUPOVM} is convenient, but not essential, as pointed out in \sref{sec:Order}. Without this assumption,  these results still  hold after minor modifications, as presented below, given that every POVM is equivalent to a simple POVM according to \lref{lem:SimplePOVM}. Similar remarks apply to other results presented in this manuscript.

\lemstarnum{\ref{lem:rank1Projective}}

\begin{lem*}
	Any rank-1 POVM $\scrA$ on $\caH$ has at least $d$ POVM elements, and the lower bound is saturated iff $\scrA$ is a rank-1 projective measurement. Meanwhile, $\dim(\spa(\scrA))\geq d$, and the  bound is
	saturated iff $\scrA$ is equivalent to a rank-1 projective measurement.
\end{lem*}

\lemstarnum{\ref{lem:commutePOVMrank1Proj}}

\begin{lem*}
Suppose  $\scrA$ is a  rank-1 POVM on $\caH$ and $\scrB$ is a set of rank-1 positive operators on $\caH$. Then $\scrA$ and $\scrB$ commute iff 
\begin{align}
\sum_{A\in \scrA\,|\, A\propto B}=\frac{B}{\tr B} \quad \forall B\in \scrB,
\end{align}	
and every two operators in $\scrB$ are either mutually orthogonal or proportional to each other.	
\end{lem*}

\lemstarnum{\ref{lem:commutePOVMrank1}}

\begin{lem*}
	Two rank-1 POVMs commute iff they are equivalent to a same rank-1 projective measurement.
\end{lem*}

\thmstarnum{\ref{thm:MUPOVM}}

\begin{thm*}
	Any set of MU rank-1 POVMs on $\caH$ contains at most $d+1$ POVMs. If the upper bound is saturated, then  all the POVMs in the set are equivalent to rank-1 projective measurements, and the corresponding bases form a CMUB. 
\end{thm*}

\section{\label{sec:POVMtdesign}Quantum measurements and $t$-designs} 
\subsection{$t$-designs}  
Let $\caS=\{|\psi_j\>,w_j\}_{j=1}^m$ be a weighted set (or collection)  of states in $\caH$, where $w_j> 0$ and $\sum_j w_j=d$ (to avoid unnecessary complications, in this paper we assume that all weights are strictly positive unless stated otherwise;  \lref{lem:EAL} in \sref{sec:FPhalf} is an exception).  As in the discussion of orthonormal bases, here we identify weighted sets that differ  only by overall phase factors. Then a weighted set is also  regarded as a distribution on the set of all pure states, which forms the complex projective space $\CP^{d-1}$.
When the weights are not mentioned explicitly, we take the convention that all states have the same weight.

Given a positive integer $t$, the set $\caS$   is a (weighted complex projective) \emph{$t$-design} if $\sum_j w_j(|\psi_j\>\<\psi_j|)^{\otimes t}$ is proportional to the projector $P_t $ onto the symmetric subspace in $\caH^{\otimes t}$ \cite{Hogg82,Zaun11,ReneBSC04,Scot06} (see \rcite{CzarGGZ20} for mixed-state designs). In view of the normalization condition $\sum_j w_j=d$, the set $\caS$ is a $t$-design iff
\begin{align}
\sum_j w_j(|\psi_j\>\<\psi_j|)^{\otimes t}=\frac{dP_t   }{D_t   },
\end{align}
where $D_t=\tr(P_t)$ is the dimension of the $t$-partite symmetric subspace and its explicit expression reads
 \begin{align}\label{eq:Dt}
 D_t   =\binom{d+t-1}{t}. 
 \end{align}
By definition it is easy to verify that a $t$-design is also a $t'$-design for any positive integer $t'$ that is smaller than or equal to $t$, that is, $t'\leq t$. 

Given any pair of positive integers $d$ and $t$, one can construct a $t$-design in dimension $d$ with a finite number of elements \cite{SeymZ84}. 
To achieve this goal, nevertheless, the number of elements is at least \cite{Hogg82,Scot06}
\begin{align}\label{eq:designEleNumLB}
\binom{d+\lceil t/2\rceil-1}{\lceil t/2\rceil}\binom{d+\lfloor t/2\rfloor-1}{\lfloor t/2\rfloor};
\end{align}
the lower bound  is equal to $d,d^2, d^2(d+1)/2, d^2(d+1)^2/4$ for $t = 1,2,3,4$, respectively. An orthonormal basis (with uniform weights) is the simplest  1-design. Prominent examples of 2-designs include CMUB and  SICs. In particular, a SIC stands out as a minimal 2-design, which saturates the lower bound in \eref{eq:designEleNumLB} with $t=2$. 
Recall that a SIC in dimension $d$ is composed of $d^2$  quantum states $|\psi_1\>,|\psi_2\>, \ldots, |\psi_{d^2}\>$  with an equal pairwise fidelity of $1/(d+1)$ \cite{Zaun11,ReneBSC04,ScotG10,FuchHS17,BengZ17book}, that is, 
\begin{align}\label{eq:SIC}
|\<\psi_j|\psi_k\>|^2 =\frac{d\delta_{jk}+1}{d+1},\quad j,k=1,2,\ldots, d^2.
\end{align}
Here each state has weight $1/d$ according to
the current normalization convention, but we shall not mention this weight explicitly for simplicity when there is no danger of confusion. In addition, the set characterized by \eref{eq:SIC} and the corresponding POVM are both referred to as a SIC (cf. \sref{sec:POVMdesignCon}). When the dimension $d$ is a power of 2, any orbit of the Clifford group is a 3-design; in particular, 
the set of stabilizer states forms a 3-design \cite{KuenG13,Zhu17MC,Webb16}. In addition, special orbits of the Clifford group can form 4-designs  \cite{ZhuKGG16,GrosNW21}.

The $t$th \emph{frame potential} is an important tool for studying $t$-designs;  given the weighted set   $\caS=\{|\psi_j\>,w_j\}_{j=1}^m$, it is defined as \cite{Zaun11,ReneBSC04,Scot06}
\begin{align}\label{eq:FP}
\Phi_t(\caS):=\sum_{j,k}w_j w_k |\<\psi_j|\psi_k\>|^{2t}.
\end{align}  
It is well known that this frame potential  satisfies the following inequality
\begin{align}\label{eq:FPLB}
\Phi_t(\caS)\geq \frac{d^2}{D_t},
\end{align}
which  is saturated iff
$\caS$ is a  $t$-design. This inequality provides a simple criterion for determining whether a weighted set is a $t$-design.

To generalize the concept of frame potential mentioned above, let $\caS=\{|\psi_j\>,w_j\}_{j=1}^m$ and $\caT=\{|\varphi_k\>,w_k'\}_{k=1}^{n}$ be two weighted  sets   of states in $\caH$, which satisfy $w_j,w_k'> 0$ and $\sum_j w_j=\sum_k w_k'=d$. The $t$th \emph{cross frame potential} between $\caS$ and $\caT$ is defined as
\begin{align}\label{eq:crossFP}
\Phi_t(\caS,\caT)=\Phi_t(\caT,\caS):=\sum_{j,k}w_j w_k' |\<\psi_j|\varphi_k\>|^{2t}. 
\end{align}
Note that the definitions in \esref{eq:FP}{eq:crossFP} are applicable even if some weights $w_j,w_k'$ are equal to zero. In addition, $\Phi_t(\caS,\caS)=\Phi_t(\caS)$, so the frame potential $\Phi_t(\caS)$ can be regarded as the cross frame potential between $\caS$ and itself. 
The significance of the cross frame potential is highlighted in \sref{sec:FPhalf} and \ref{sec:FidCal}.

 Although the frame potential $\Phi_t$ was originally introduced when $t$ is a positive integer, the definition in \eref{eq:FP} applies to any positive real number $t$. Similar generalization applies to the cross frame potential defined in \eref{eq:crossFP}. However, it is not so easy to generalize the concept of $t$-designs in this way (the special case $t=1/2$ will be discussed in \sref{sec:FPhalf}). For example, 
the $t$th frame potentials of an orthonormal basis, SIC, and CMUB are respectively given by 
\begin{align}
\Phi_t(\mbox{basis})&=d,\label{eq:FPbasis}\\
\Phi_t(\mbox{SIC})&=1+\frac{d^2-1}{(d+1)^t},\label{eq:FPSIC}\\
\Phi_t(\mbox{CMUB})&=\frac{d+d^{3-t}}{d+1}.
\end{align}
In addition, the $t$th frame potential of Haar random pure states can be computed as follows,
\begin{align}
\Phi_t(\mbox{Haar})&=d^2\int_{\CP^{d-1}} |\<0|\psi\>|^{2t}\rmd \mu(\psi)\nonumber\\
&=\frac{d^2\int_{\theta=0}^{\pi/2} (\cos\theta)^{2t+1}(\sin\theta)^{2d-3} \rmd \theta}{\int_{\theta=0}^{\pi/2} \cos\theta(\sin\theta)^{2d-3}\rmd\theta}\nonumber\\
&=\frac{d^2\Gamma(d)\Gamma(t+1)}{\Gamma(d+t)}, \label{eq:HaarFP}
\end{align} 
where $\rmd \mu(\psi)$ denotes the normalized  measure on the complex projective space $\CP^{d-1}$ that is induced by the Haar measure on the unitary group. If $t$ is an integer, then the above equation yields
\begin{align}
\Phi_t(\mbox{Haar})=\frac{d^2t!}{d(d+1)\cdots (d+t-1)}=\frac{d^2}{D_t},
\end{align}
which saturates the lower bound in \eref{eq:FPLB}. So the ensemble of  Haar random pure states forms a $t$-design for any positive integer $t$ as expected. In view of this fact, the ensemble of Haar random pure states is regarded as an $\infty$-design.

\subsection{\label{sec:POVMdesignCon}Connection between $t$-designs and quantum measurements}

Given any $t$-design $\{|\psi_j\>,w_j\}_{j=1}^m$ with $t\geq 1$, we can construct a rank-1 POVM of the form $\{w_j|\psi_j\>\<\psi_j|\}_{j=1}^m$. Conversely, any rank-1 POVM determines a $t$-design up to irrelevant overall phase factors. Quantum measurements based on $t$-designs have  numerous applications in quantum information processing.
 Notably,  collective measurements based on $t$-designs are optimal for pure-state estimation \cite{HayaHH05,Zhu12the,ZhuH18U}.
Measurements constructed from 2-designs are  optimal for linear quantum state tomography \cite{Scot06,RoyS07,ZhuE11,Zhu12the} and quantum state verification \cite{ZhuH19O,LiHZ19,LiHZ20}.  Measurements constructed from 3-designs are useful in shadow estimation and entanglement detection \cite{HuanKP20,ElbeKHB20,ZhouZL20}. In addition, the quantum measurement constructed from Haar random pure states, referred to as the \emph{isotropic measurement} henceforth, is of special interest in quantum state estimation \cite{Zhu14IOC,Zhu12the} and discrimination \cite{MattWW09}.

Thanks to the connection mentioned above, some concepts defined for $t$-designs can be generalized to rank-1 POVMs, and vice versa. Notably,
the definitions of the $t$th frame  potential and cross frame  potential can be extended to rank-1 POVMs. To be concrete, let $\scrA=\{A_j\}_j$ and $\scrB=\{B_k\}_k$ be two rank-1 POVMs. Then the $t$th frame potential of $\scrA=\{A_j\}$  reads
\begin{align}\label{eq:POVMFP}
\Phi_t(\scrA)=\sum_{j,k}\frac{[\tr(A_jA_k)]^t}{[\tr(A_j)\tr(A_k)]^{t-1}},
\end{align}
which is applicable for any positive number $t$.
Similarly, the $t$th cross frame potential between $\scrA$  and $\scrB$ reads
\begin{align}\label{eq:POVMcrossFP}
\Phi_t(\scrA,\scrB)=\sum_{j,k}\frac{[\tr(A_jB_k)]^t}{[\tr(A_j)\tr(B_k)]^{t-1}}. 
\end{align}
The definitions of frame potential and cross frame potential in \esref{eq:POVMFP}{eq:POVMcrossFP} can be further generalized to POVMs that are not necessarily rank-1, although they are most useful when the POVMs are rank~1. In addition, these definitions are applicable even if some POVM elements are equal to the zero operator as long as the summations are restricted to POVM elements that are not equal to the zero operator. Note that equivalent POVMs have the same frame potential; similarly, equivalent pairs of POVMs have the same cross frame potential.

On the other hand, a $t$-design is called simple if the corresponding POVM is simple. 
Two $t$-designs are MU if the corresponding POVMs are MU. This definition reduces to the usual definition of MUB when each weighted set is an orthonormal basis with uniform weights.

\subsection{\label{sec:FPhalf}Applications of the frame potential $\Phi_{1/2}$}

\begin{figure}[b]
	\includegraphics[width=7.5cm]{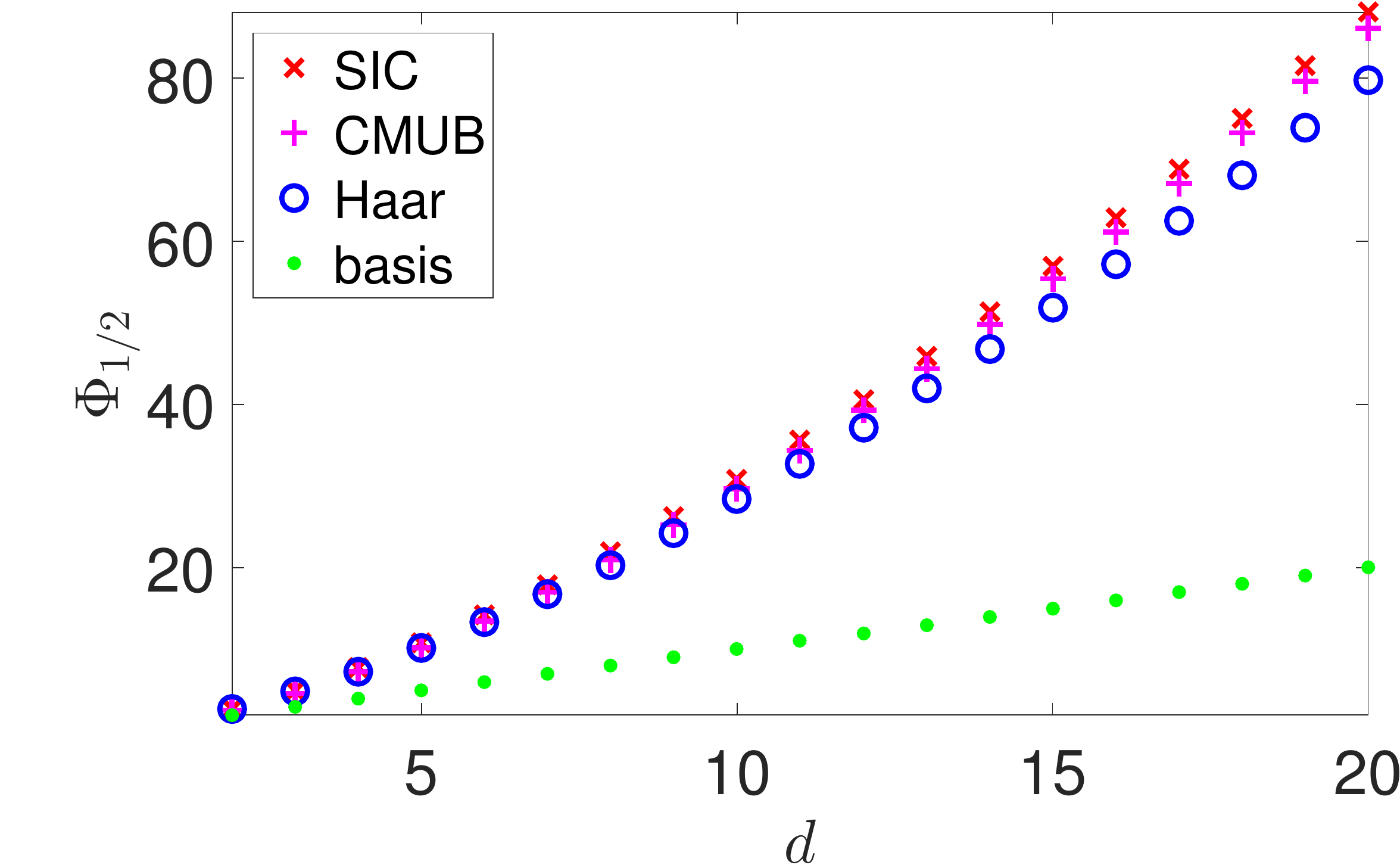}
	\caption{\label{fig:FP} Frame potentials $\Phi_{1/2}$ of an orthonormal basis, SIC, CMUB, and Haar random pure states, respectively.  Here $d$ is the dimension of the underlying Hilbert space $\caH$. 
	}
\end{figure}

Here we are particularly interested in  the (cross) frame potential $\Phi_t$ with  $t=1/2$, whose significance will become clear shortly. A   weighted set   of states in $\caH$ is a  \emph{$1/2$-design} if it is a  1-design and can attain the maximum frame potential $\Phi_{1/2}$ among all 1-designs.  Here we choose the maximum frame potential because the square-root function is concave rather than convex.
By definition a 1/2-design is automatically a 1-design; this requirement is motivated by our study on quantum measurements: any rank-1 measurement corresponds to a 1-design, and vice versa. By contrast, a 1-design is not necessarily a $1/2$-design, although a $(t+1)$-design is automatically a $t$-design when $t$  is a positive integer. The intuition about $t$-designs for an integer $t$ cannot be taken for granted in the current study. Incidentally, the search for "fractional designs" has been a long-standing open problem. So far it is still not clear how to define $t$-designs with  arbitrary real parameter $t$. This problem deserves further study, but is not crucial to the current work.

According to \ecref{eq:FPbasis}{eq:HaarFP}, the frame potentials $\Phi_{1/2}$
for an orthonormal basis, SIC, CMUB, and Haar random pure states are respectively given by 
\begin{align}
\Phi_{1/2}(\mbox{basis})&=d, \\
\Phi_{1/2}(\mbox{SIC})&=1+(d-1)\sqrt{d+1},\label{eq:FPhalfSIC}\\
\Phi_{1/2}(\mbox{CMUB})&=\frac{d+d^{5/2}}{d+1},\\
\Phi_{1/2}(\mbox{Haar})&=\frac{\sqrt{\pi}d^2\Gamma(d)}{2\Gamma\bigl(d+\frac{1}{2}\bigr)}. 
\end{align}
Quite unexpectedly, the frame potential attains its minimum at an orthonormal basis and its maximum at a SIC, as illustrated in \fref{fig:FP}.  What is more surprising is that this conclusion holds even if we consider
 all possible 1-designs, as shown in \lref{lem:FPhalfLBUB} below. 
\begin{lem}\label{lem:FPhalfLBUB}
	Any 1-design  $\caS$ in $\caH$ satisfies 
	\begin{align}\label{eq:FPhalfLBUB}
	d\leq \Phi_{1/2}(\caS)\leq 
	1+(d-1)\sqrt{d+1}.
	\end{align}	
	If $\caS$ is simple, then the lower bound  is saturated iff $\caS$ is an orthonormal basis, while
	the upper bound is saturated iff  $\caS$ is a SIC.
\end{lem}

\Lref{lem:FPhalfLBUB} reveals  intriguing connections between the frame potential $\Phi_{1/2}(\caS)$ and orthonormal bases as well as SICs. Notably, it shows that, among all 1-designs (including $t$-designs for any positive integer $t$), SICs are essentially the only $1/2$-designs; in other words, all $1/2$-designs are tied to SICs. Remarkably, SICs are uniquely characterized by the maximum frame potential, without any assumption even on the number of  elements.
These observations will have profound implications for understanding quantum measurements in the perspective of quantum state estimation, as we shall see later. The detailed proof of \lref{lem:FPhalfLBUB} is presented in \aref{asec:FP}. Here it is worth pointing out that the proof is tied to a surprising  result on the  $1/2$-moment in statistics, as formulated and proved in \aref{asec:moment}, which is of independent interest.

When the number of states is limited, the upper bound in \lref{lem:FPhalfLBUB} can be improved, and the maximum frame potential is tied to tight equiangular lines. Recall that  a set $\{|\psi_j\>\}_{j=1}^m$ composed of $m$ states is equiangular if all pairwise fidelities are equal \cite{LemmS73,Zaun11}.  The equiangular set is tight if  $\sum_j |\psi_j\>\<\psi_j|$ is proportional to the identity (automatically guaranteed for 1-designs), in which case the pairwise fidelities saturate the Welch bound \cite{Welc74}:
\begin{align}\label{eq:EAL}
|\<\psi_j|\psi_k\>|^2=\begin{cases}
1 & j=k,\\
\frac{m-d}{d(m-1)} & j\neq k. 
\end{cases}
\end{align}
Any set of equiangular states in dimension $d$ can contain at most $d^2$ states, and the upper bound is saturated iff the set is a SIC \cite{LemmS73,Zaun11,ApplFZ15G}. In \lref{lem:EAL} below we assume that all weights in $\caS$ are nonnegative, but not necessarily strictly positive.
\begin{lem}\label{lem:EAL}
	Any 1-design $\caS$ composed of  $m$ states satisfies
	\begin{align}\label{eq:EALfp}
	&\Phi_{1/2}(\caS)\leq  \frac{d^2}{m}+\frac{d}{m}\sqrt{d(m-1)(m-d)}.
	\end{align}
	When $m>d^2$, the upper bound cannot be saturated; when $d\leq m\leq d^2$, the  bound is saturated iff $\caS$ is composed of   $m$ equiangular states (with uniform weights).
\end{lem}
Note that the upper bound in \eref{eq:EALfp} is strictly monotonically increasing in $m$. 

Next, we turn to the cross frame potential $\Phi_{1/2}(\caS,\caT)$, which will play  important roles in studying MUMs and SICs. 
\begin{lem}\label{lem:crossFP}
	Any pair of 1-designs  $\caS$ and $\caT$ in $\caH$ satisfies 
	\begin{align}\label{eq:crossFPLBUB}
	d\leq \Phi_{1/2}(\caS,\caT)\leq d^{3/2}.
	\end{align}	
The upper bound is saturated iff $\caS$ and $\caT$ are MU. 	If $\caS$ and $\caT$ are simple, then the lower bound  is saturated iff $\caS$ and $\caT$ are identical orthonormal bases  up to relabeling.
\end{lem}

 \begin{lem}\label{lem:crossFP2design}
Suppose $\caS$ and $\caT$ are 1-designs in $\caH$, and one of them is a 2-design; then 
	\begin{align}\label{eq:crossFP2design}
	\Phi_{1/2}(\caS,\caT)\leq 	1+(d-1)\sqrt{d+1}.
	\end{align}	
	If $\caS$ and $\caT$ are simple, then the upper bound  is saturated iff $\caS$ and $\caT$ are identical SICs up to relabeling. 
\end{lem}

Here 
we identify weighted sets that differ only by overall phase factors as mentioned before.
Thanks to the connections between   1-designs and  POVMs, \lsref{lem:FPhalfLBUB}-\ref{lem:crossFP2design} above can be generalized to rank-1 POVMs immediately as summarized in \lsref{lem:POVMFP}-\ref{lem:POVMcrossFP2design} below.
 \begin{lem}\label{lem:POVMFP}
Any rank-1 POVM $\scrA$ satisfies
	\begin{align}
	d\leq \Phi_{1/2}(\scrA)\leq 
	1+(d-1)\sqrt{d+1}.
	\end{align}	
If $\scrA$ is simple, then 	the lower bound  is saturated iff $\scrA$ is a rank-1 projective measurement, while the upper bound is saturated iff $\scrA$ is a SIC.
\end{lem}

\begin{lem}\label{lem:POVMEAL} 
	Any  rank-1 POVM $\scrA$ composed of  $m$ POVM elements satisfies
	\begin{align}\label{eq:POVMEALfp}
	&\Phi_{1/2}(\scrA)\leq  \frac{d^2}{m}+\frac{d}{m}\sqrt{d(m-1)(m-d)}.
	\end{align}
	When $m>d^2$, the upper bound cannot be saturated; when $d\leq m\leq d^2$, the upper bound is saturated iff $\scrA$ is unbiased and equiangular. 
\end{lem}

 \begin{lem}\label{lem:POVMcrossFP}
	Any pair of rank-1 POVMs $\scrA$ and $\scrB$ on $\caH$ satisfies 
	\begin{align}
	d\leq \Phi_{1/2}(\scrA, \scrB)\leq d^{3/2}.
	\end{align}	
	The upper bound is saturated iff $\scrA$ and $\scrB$ are MU. 	If $\scrA$ and $\scrB$ are simple, then the lower bound  is saturated iff $\scrA$ and $\scrB$ are identical rank-1 projective measurements up to relabeling.
\end{lem}

 \begin{lem}\label{lem:POVMcrossFP2design}
Suppose $\scrA$ and $\scrB$ are rank-1 POVMs on $\caH$, and one of them is  	
constructed from a 2-design; then
	\begin{align}
	\Phi_{1/2}(\scrA,\scrB)\leq 	1+(d-1)\sqrt{d+1}.
	\end{align}	
	If $\scrA$ and $\scrB$ are simple, then the upper bound  is saturated iff $\scrA$ and $\scrB$ are identical SICs up to relabeling. 
\end{lem}
The above results demonstrate the significance of the (cross) frame potential $\Phi_{1/2}$ in characterizing typical quantum measurements. These results are the stepping stones for understanding quantum measurements in the perspective of quantum state estimation, as we shall see shortly.

\section{\label{sec:decoding}Decoding quantum measurements from a simple estimation problem}    
\subsection{\label{sec:problem}Reexamination of a  simple estimation problem} 
Suppose a quantum device can prepare a random  pure quantum state $\rho$ on $\caH$ according to the normalized Haar measure and we are asked to estimate the identity of $\rho$ based on suitable quantum measurements. A natural figure of merit is the fidelity averaged over the measurement outcome and the random pure state. 
Given $N$ copies of the pure state $\rho$, then what average fidelity can we achieve? This problem is  well known in the quantum information community and  has been studied by many eminent researchers \cite{MassP95,DerkBE98,LatoPT98,Haya98,BrusM99,GisiP99,Mass00,AcinLP00,Bana01,BagaBM02,HayaHH05}, whose works have greatly improved our understanding about information extraction from quantum systems. 
However, the significance of this problem in decoding quantum measurements has not been fully appreciated.

If we perform the POVM $\scrA=\{A_j\}_j$ on $\rho^{\otimes N}$, then the probability of obtaining outcome $A_j$ is $p_j=\tr(\rho^{\otimes N}A_j)$. Let $\hat{\rho}_j$ be the estimator corresponding to outcome $j$. Then the average  fidelity reads
\begin{align}\label{eq:FidelityUB}
\bar{F}&=\sum_j \int_{\CP^{d-1}} \rmd \mu(\psi) p_j \tr(\rho\hat{\rho}_j)\nonumber\\
& =\sum_j \int_{\CP^{d-1}} \rmd \mu(\psi) \tr\bigl[\rho^{\otimes N+1}(A_j \otimes \hat{\rho}_j)\bigr]\nonumber \\
&=\frac{1}{D_{N+1}}\sum_j  \tr[P_{N+1}(A_j \otimes \hat{\rho}_j)]
\nonumber\\
&=\frac{1}{D_{N+1}}\sum_j  \tr[\tilde{\caQ}(A_j) \hat{\rho}_j],
\end{align}
 where $\rmd \mu(\psi)$ denotes the normalized  measure on the complex projective space $\CP^{d-1}$ that is induced by the Haar measure on the unitary group, $\rho=|\psi\>\<\psi|$,  and
\begin{align} \tilde{\caQ}(A_j):=\tr_{1,\ldots,N}[P_{N+1}(A_j\otimes 1)].
\end{align}
Note that 
\begin{align}
\tr[\tilde{\caQ}(A_j) \hat{\rho}_j]\leq \|\tilde{\caQ}(A_j)\|=\max_{\rho} \tr[P_{N+1}(A_j\otimes \rho)],
\end{align}
where $\|\tilde{\caQ}(A_j)\|$ denotes the operator norm of $\tilde{\caQ}(A_j)$ and
the maximization is taken over all normalized pure states on $\caH$. In addition, the upper bound is saturated if the estimator  $\hat{\rho}_j$ is supported in the eigenspace of $\tilde{\caQ}(A_j)$ corresponding to  the largest eigenvalue and only then.

The \emph{estimation fidelity} of $\scrA$ is defined as the maximum fidelity that can be achieved by the POVM $\scrA$ (with a judicial choice of each estimator $\hat{\rho}_j$)  and is given by 
\begin{align}\label{eq:FidPOVM}
F(\scrA)=\sum_j \frac{\|\tilde{\caQ}(A_j)\|}{D_{N+1}}. 
\end{align}
Define
\begin{equation}
\caQ(O):=(N+1)!\tilde{\caQ}(O) 
\end{equation}
for  any linear operator $O$ acting on $\caH^{\otimes N}$. Then the estimation fidelity in \eref{eq:FidPOVM} can be expressed as 
\begin{align}\label{eq:FidPOVM2}
F(\scrA)=\sum_j \frac{\|\caQ(A_j)\|}{(N+1)! D_{N+1}}=\sum_j \frac{\|\caQ(A_j)\|}{d(d+1)\cdots (d+N)}. 
\end{align}
As we shall see shortly, $F(\scrA)$ encodes valuable information about the POVM $\scrA$. To facilitate the following discussions, here we summarize a number of simple but instructive facts. 
Let $\scrI$ be the trivial POVM that is composed of the identity on $\caH$ as the only POVM element. 
\begin{lem}\label{lem:FidPOVMbasic}
	Suppose $\scrA,\scrB$ are POVMs on $\caH^{\otimes N}$, $\scrC$ is a POVM on $\caH^{\otimes k}$, and $U$ is a unitary operator on $\caH$. Then 
	\begin{gather}
	F\Bigl(U^{\otimes N}\scrA {U^\dag}^{\otimes N}\Bigr)= F(\scrA),\label{eq:fidUnitaryInvariance}\\
	F(\scrA\otimes \scrI^{\otimes k})= F(\scrA), \label{eq:fidTensorTrivialPOVM}\\
	F(\scrA)\leq F(\scrB) \quad \mbox{if } \scrA\preceq\scrB, \label{eq:fidCoarseGraining}\\
F(\scrC\otimes \scrA)=F(\scrA\otimes \scrC)\geq  \max\{F(\scrA),F(\scrC)\}. \label{eq:fidProductComponent}
	\end{gather}
\end{lem}
Here the notation $\scrA\preceq\scrB$ means $\scrA$ is a coarse graining of $\scrB$ as defined in \sref{sec:Order}. \Lref{lem:FidPOVMbasic} in particular implies that equivalent POVMs can achieve the same estimation fidelity as expected.

At this point, it is worth pointing out that the above results still apply if the ensemble of  Haar random pure states involved in the estimation problem is replaced by any ensemble of pure states that forms a $t$-design with $t=N+1$. This observation is quite helpful in devising experiments for  demonstrating these results.

\subsection{\label{sec:FidCal}Calculation of the estimation fidelity}
Next, we derive  a number of basic results that are useful to computing the estimation fidelity in \eref{eq:FidPOVM2}, especially for product measurements. Let $A, B, C $ be positive semidefinite operators on $\caH$. Straightforward calculation shows that
\begin{widetext}
\begin{align}
\caQ(A)&=\tr(A)+A,  \label{eq:QA}\\
\caQ(A\otimes B)&=\tr(A)\tr(B)+\tr(AB)+\tr(B)A+\tr(A)B
+AB+BA \label{eq:QAB},\\
\caQ(A\otimes B\otimes C)&=\tr(A)\tr(B)\tr(C)+\tr(AB)\tr(C)+\tr(BC)\tr(A)+\tr(CA)\tr(B)+\tr(ABC)+\tr(ACB)\nonumber\\
&+\tr(B)\tr(C)A+\tr(C)\tr(A)B+\tr(A)\tr(B)C+\tr(BC)A+\tr(CA)B+\tr(AB)C+\tr(C)(AB+BA)\nonumber\\
&+\tr(B)(AC+CA)+\tr(A)(BC+CB)+ABC+ACB+BCA+BAC+CAB+CBA. \label{eq:QABC}
\end{align}
Here numbers in operator equations, such as $\tr(A)$, are implicitly multiplied by the identity operator; a similar convention applies to other equations in this paper. 
When $\tr(A)=\tr(B)=\tr(C)=1$, Eqs.~\eqref{eq:QA}-\eqref{eq:QABC} simplify to 
\begin{align}
\caQ(A)=&1+A, \label{eq:QAtrace1}\\
\caQ(A\otimes B)=&1+f+A+B+AB+BA,\label{eq:QABtrace1}\\
\caQ(A\otimes B\otimes C)=&1+f_{12}+f_{23}+f_{31}+f_{123}+f_{132}+(1+f_{23})A+(1+f_{31}) B+(1+f_{12})C+AB+BA\nonumber\\
&+AC+CA+BC+CB
+ABC+ACB+BCA+BAC+CAB+CBA,  \label{eq:QABCtrace1}
\end{align}	
where $f=f_{12}=\tr(AB)$, $f_{23}=\tr(BC)$, $f_{31}=\tr(CA)$, $f_{123}=\tr(ABC)$, and $f_{132}=\tr(ACB)$. 
These equations indicate that the three-copy estimation fidelity may depend on the triple products of POVM elements in addition to pairwise overlaps (cf. \sref{sec:IneqMUBSIC}).
\end{widetext}

\Eref{eq:QA} implies that  
\begin{align}\label{eq:QAnorm}
\|\caQ(A)\|=\tr(A)+\|A\|,
\end{align}
from which we can derive the following lemma. 
\begin{lem}\label{lem:OneCopyFidLBUB}
	The estimation fidelity $F(\scrA)$ of any POVM  $\scrA=\{A_j\}_j$ on $\caH$ satisfies
	\begin{equation}
	\frac{1}{d}\leq F(\scrA)=\frac{1}{d+1}+\frac{1}{d(d+1)}\sum_j \|A_j\|\leq \frac{2}{d+1}, \label{eq:OneCopyFidUB}
	\end{equation}
	and the lower bound is saturated iff $\scrA$ is trivial, while the upper bound is saturated iff $\scrA$ is  rank-1. 
\end{lem}
The equality in \eref{eq:OneCopyFidUB}
was originally derived in \rcite{Bana01} and play an important role in studying information-disturbance relations. The upper bound in  \eref{eq:OneCopyFidUB} was known even earlier in the context of optimal quantum state estimation \cite{BrusM99,AcinLP00}. Here we are interested in \lref{lem:OneCopyFidLBUB} because it clarifies the estimation fidelities of single-copy measurements and highlights the significance of rank-1 measurements (cf. \pref{pro:MaximalPOVM}). 
If $\scrA$ is a POVM with $m$ POVM elements, then \lref{lem:OneCopyFidLBUB} yields
\begin{equation} F(\scrA)\leq \frac{1}{d+1}+\frac{m}{d(d+1)},
\end{equation}
given that $\|A_j\|\leq 1$ for any POVM element $A_j$. 
Here the inequality is saturated when $\scrA$ is a projective measurement, in which case
 the estimation fidelity is completely determined by the number of measurement outcomes.

The following lemma is a stepping stone for studying two-copy estimation fidelities. 
\begin{lem}\label{lem:Sym2ProjNorm}
	Suppose $A$ and $B$ are nonzero positive operators on $\caH$ and $f=\tr(AB)/[\tr(A)\tr(B)]$. Then 
	\begin{align}\label{eq:SymProjNorm}
	\|\caQ(A\otimes B)\|\leq 2\tr(A)\tr(B)\bigl(1+f+\sqrt{f}\lsp\bigr). 
	\end{align}
	The upper bound is saturated iff one of the two conditions holds:
\begin{enumerate}
	\item  both $A$ and $B$ are rank 1;
	\item  $A$ and $B$ have orthogonal supports and one of them is rank 1. 
\end{enumerate}	 
\end{lem}

When $A=|\psi\>\<\psi|$ and $B=|\varphi\>\<\varphi|$ are rank-1 projectors, \lref{lem:Sym2ProjNorm} yields
\begin{align}
\|\caQ(|\psi\>\<\psi|\otimes |\varphi\>\<\varphi|)\|=2(1+|\<\psi|\varphi\>|^2+|\<\psi|\varphi\>|),
\end{align}
which in turn implies that
\begin{align}
\bigl\|\caQ(|j\>\<j|\otimes |k\>\<k|)\bigr\|=&\begin{cases}
6 & j=k,\\
2 &j\neq k. 
\end{cases}\label{eq:Qjk}
\end{align}
Note that \eref{eq:Qjk} also follows from  \eref{eq:QABtrace1}. By contrast, \eref{eq:QABCtrace1} implies that
\begin{align}\label{eq:Qjkl}
\bigl\|\caQ(|j\>\<j|\otimes |k\>\<k|\otimes |l\>\<l|)\bigr\|=&\begin{cases}
24 & j=k=l,\\
6 &j=k\neq l,\\
2 &j\neq k, k\neq l, j\neq l.
\end{cases}
\end{align}
Here $|j\>,|k\>,|l\>$ denote basis states in the computational basis.

By virtue of \eref{eq:FidPOVM2} and \lref{lem:Sym2ProjNorm}, we can derive a general upper bound for the two-copy estimation fidelity of  any product measurement. 
\begin{lem}\label{lem:FidprodPOVM}
Let $\scrA$ and $\scrB$ be two POVMs on $\caH$. Then  the two-copy estimation fidelity $F(\scrA\otimes \scrB)$ satisfies
	\begin{align}\label{eq:FidProdUBgen}
	&F(\scrA\otimes \scrB)\leq\frac{2d(d+1)+2\Phi_{1/2}(\scrA,\scrB)}{d(d+1)(d+2)},
	\end{align}		
	and the upper bound is saturated iff   $\scrA$ and $\scrB$ are  rank~1. 
\end{lem}

The following lemma is an immediate corollary of  \lref{lem:FidprodPOVM}  with $\scrB=\scrA$. 
\begin{lem}\label{lem:Fid2copyFP}
Let $\scrA$ be a POVM on  $\caH$. 
Then the two-copy estimation fidelity $F(\scrA^{\otimes 2})$ satisfies 
\begin{equation}\label{eq:Fid2copyFP}
F(\scrA^{\otimes 2})\leq \frac{2d(d+1)+2\Phi_{1/2}(\scrA)}{d(d+1)(d+2)},
\end{equation}	
and the upper bound is saturated iff  $\scrA$ is rank~1. 	
\end{lem}
 \Lref{lem:Fid2copyFP}  establishes a precise connection between the two-copy estimation fidelity $F(\scrA^{\otimes 2})$ and the frame potential $\Phi_{1/2}(\scrA)$ [instead of the frame potential $\Phi_2(\scrA)$ as one may naively expect] and thereby  endowing the frame potential with a concrete operational meaning in the context of quantum state estimation. Similarly, \lref{lem:FidprodPOVM}
endows the cross frame potential with a concrete operational meaning.

\subsection{Impact of coarse graining}
Suppose $\scrA$ is a rank-1 POVM on $\caH$; then any nontrivial coarse graining of $\scrA$ can only achieve a smaller estimation fidelity according to \lref{lem:OneCopyFidLBUB}. The situation is a bit more complicated for the two-copy estimation fidelity.
Here we try to shed some light on this problem, which will be useful to studying the connection between the  estimation fidelity and quantum incompatibility, as we shall see in \sref{sec:incompatibility}. 

\begin{lem}\label{lem:SymProjNormSum}
	Suppose $A, B_1, B_2$ are rank-1 positive operators on $\caH$ and $B=B_1 +B_2$. Then
	\begin{align}\label{eq:SymProjNormSum}
	\|\caQ(A\otimes B)\|\leq  \|\caQ(A\otimes B_1 )\|+\|\caQ(A\otimes B_2)\|,
	\end{align}
	and the inequality is saturated iff one of the following three conditions holds:
	\begin{enumerate}
		\item  $B_2$ is proportional to $B_1$;
		\item $A$ is orthogonal to both $B_1$ and $B_2$;
		\item  $A, B_1, B_2$ are all supported in a common two-dimensional subspace of $\caH$ and $A$ is orthogonal to $B_1$ or $B_2$.
	\end{enumerate}
\end{lem}

\begin{lem}\label{lem:SymProjNormSumPOVM}
	Suppose  $\scrA=\{A_j\}_{j=1}^m$ is a simple rank-1 POVM on $\caH$; let  $B_1, B_2,\ldots, B_n$ be $n\geq 2$ rank-1 positive operators on $\caH$ that are pairwise linearly independent,  and let $B=\sum_{k=1}^n B_k$. Then
	\begin{align}\label{eq:SymProjNormSumPOVM}
	\sum_{j=1}^m\|\caQ(A_j \otimes B)\|\leq  \sum_{j=1}^m \sum_{k=1}^n\|\caQ(A_j\otimes B_k )\|,
	\end{align}
	and the inequality is saturated iff
	$B_1, B_2,\ldots, B_n$ are mutually orthogonal
	and they commute with all POVM elements in $\scrA$. In that case, $\scrA$ contains $n$ rank-1 projectors that are proportional to $B_1, B_2, \ldots, B_n$, respectively, and is thus reducible. 	
\end{lem}

Note that the inequality in \eref{eq:SymProjNormSumPOVM} cannot be saturated if $\scrA$ is irreducible or $B_1, B_2,\ldots, B_n$ are not mutually orthogonal. This observation leads to the following lemma.

\begin{lem}\label{lem:FidCoarsegraining}
	Suppose $\scrA$ and $\scrC$ are two  rank-1  POVMs on $\caH$ and  $\scrB$ is a coarse graining of $\scrC$. Suppose $\scrA$ is irreducible or $\scrC$ contains no two POVM elements that are mutually orthogonal. Then 
	\begin{align}\label{eq:ABAC}
	F(\scrA\otimes \scrB)\leq F(\scrA\otimes \scrC),
	\end{align}
	and the inequality is saturated iff $\scrB$ is equivalent to $\scrC$. 
\end{lem}
\Lref{lem:FidCoarsegraining} shows that the two-copy estimation fidelity $F(\scrA\otimes \scrC)$ can only decrease when $\scrC$ is replaced by a nontrivial coarse graining, assuming that $\scrA$ is irreducible or $\scrC$ contains no two POVM elements that are mutually orthogonal. Nevertheless, this conclusion no longer holds if the underlying assumption is dropped; cf. \eref{eq:FidProjective} in the next section.

\section{\label{sec:TypicalQM}Typical quantum measurements in the light of quantum state estimation}   
\Lref{lem:OneCopyFidLBUB} in \sref{sec:FidCal} offers a succinct characterization of rank-1 measurements as optimal single-copy measurements. Here we further demonstrate that the estimation fidelity  in \esref{eq:FidPOVM}{eq:FidPOVM2} is a powerful tool for characterizing  typical quantum measurements, including rank-1 projective measurements, MUMs, SICs, and measurements based on tight equiangular lines. Notably, all these elementary quantum measurements are uniquely characterized by extremal one-copy and two-copy  estimation fidelities as summarized in \tref{tab:POVMtypical}. In other words, all these   measurements are uniquely determined by their information-extraction capabilities. To achieve our goal, we shall completely characterize all quantum measurements that can attain certain extremal estimation fidelities in a number of natural settings. Note that it is not enough to determine one optimal measurement.

\subsection{Optimal collective measurements and $t$-designs}

As a generalization of \lref{lem:OneCopyFidLBUB}, the following theorem determines tight lower bound and upper bound for the estimation fidelity $F(\scrA)$ for $N$-copy measurements, assuming that all collective measurements are accessible. 
In addition, optimal $N$-copy measurements are clarified. 
\begin{thm}\label{thm:nCopyFidUB}
Let $\scrA=\{A_j\}_j$ be any POVM 	
 on $\caH^{\otimes N}$. Then the $N$-copy estimation fidelity  $F(\scrA)$ satisfies 
	\begin{equation}
\frac{1}{d}\leq	F(\scrA)\leq \frac{N+1}{N+d}, \label{eq:nCopyFidLBUB}
	\end{equation}	
and the lower bound is saturated iff $\tilde{\caQ}(A_j)$ for each $j$ is proportional to the identity, while	
	 the upper bound is saturated iff $P_N A_j P_N$ for each $j$ is  proportional to the $N$th tensor power of a pure state. 
\end{thm}
The lower bound in \eref{eq:nCopyFidLBUB} corresponds to the performance of a random guess; the upper bound is well known in the context of optimal quantum state estimation \cite{MassP95,BrusM99,AcinLP00,HayaHH05}.   A self-contained proof of \thref{thm:nCopyFidUB} is presented in \aref{asec:OptCollProof}. 
When restricted to the $N$-partite symmetric subspace, a POVM is  optimal iff it has the form
\begin{align}\label{eq:OptimalCollPOVM}
\Bigl\{\frac{D_N }{d}w_j (|\psi_j\>\<\psi_j|)^{\otimes N}\Bigr\}_j, 
\end{align}
where $D_N$ is determined by  \eref{eq:Dt} and $\{|\psi_j\>,w_j\}_j$
 forms a  $t$-design with $t=N$ \cite{HayaHH05}. This observation establishes a simple connection between optimal collective measurements and quantum measurements based on $t$-designs.

 Incidentally,  prominent examples of 2-designs include CMUB \cite{Schw60,Ivan81,WootF89,DurtEBZ10,BengZ17book} and  SICs \cite{Zaun11,ReneBSC04,ScotG10,FuchHS17,BengZ17book}.  When the dimension $d$ is a power of 2, any orbit of the Clifford group is a 3-design. In the case of a qubit, the vertices of the regular tetrahedron, octahedron, cube, icosahedron, and dodecahedron inscribed  on the Bloch sphere form $t$-designs with $t=2,3,3,5,5$, respectively.  The vertices of the octahedron also correspond to a CMUB.
 These  platonic solids can be used to construct
 optimal collective measurements according to \eref{eq:OptimalCollPOVM} \cite{LatoPT98}. 
 Note that $D_t=t+1$ when $d=2$ by \eref{eq:Dt}.
 Suppose a platonic solid forms a $t$-design and let $\{\vec{r}_j\}_{j=1}^m$  be the set of unit vectors that specify its vertices on the Bloch sphere. Then the corresponding optimal collective measurements on the symmetric subspace of $\caH^{\otimes t}$ can be constructed as follows,
 \begin{align}
 \biggl\{\frac{t+1 }{m} \Bigl(\frac{1+\vec{r}_j\cdot\vec{\sigma}}{2}\Bigr)^{\otimes t}\biggr\}_{j=1}^m, 
 \end{align}
where $\vec{\sigma}=(\sigma_x,\sigma_y,\sigma_z)$ is the vector composed of the three Pauli operators.  The optimal two-copy collective measurement based on the regular tetrahedron has already been realized in photonic experiments \cite{HouTSZ18}. Optimal collective measurements 
based on other platonic solids might also be realized in the near future.

\subsection{Rank-1 projective measurements and SICs}

From now on we focus on the estimation fidelities of individual measurements, which are more instructive to understanding quantum  measurements on $\caH$. Given a POVM $\scrA$  on $\caH$, recall that $F(\scrA^{\otimes k})$ denotes the $k$-copy estimation fidelity achieved by the product POVM $\scrA^{\otimes k}$ built from  $\scrA$. Suppose $\scrA$ is a rank-1 projective measurement on $\caH$; by virtue of \lref{lem:OneCopyFidLBUB} and Eqs.~\eqref{eq:Qjk},  \eqref{eq:Qjkl}, it is straightforward to verify that 
\begin{align}\label{eq:FidProjective}
F(\scrA^{\otimes 2})=F(\scrA)=\frac{2}{d+1},\quad F(\scrA^{\otimes 3})
=\frac{2(d+5)}{(d+2)(d+3)}.
\end{align}
Interestingly,  identical projective measurements on two copies can only achieve the same estimation fidelity as a single-copy projective measurement, but  identical projective measurements on three copies can enhance the estimation fidelity.

\begin{figure}
	\includegraphics[width=7cm]{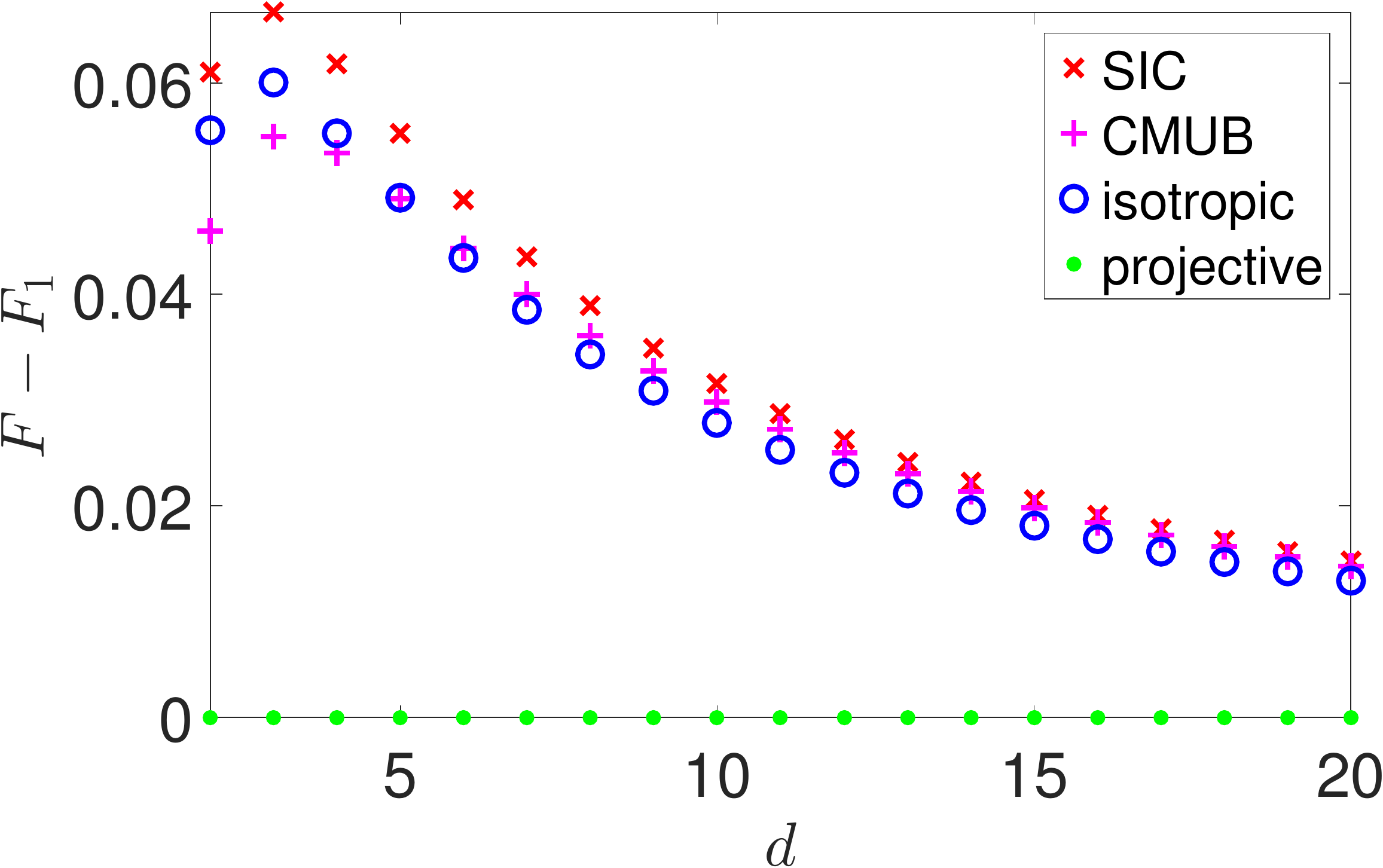}
	\caption{\label{fig:fidiid}
	Shifted two-copy estimation fidelities $F=F(\scrA^{\otimes 2})$, where $\scrA$ is a measurement constructed from an orthonormal basis (projective), SIC, CMUB, or Haar random pure states (isotropic).  Here $F_1=2/(d+1)$ is the one-copy estimation fidelity achieved by a rank-1 projective measurement (or any rank-1 measurement). 
	}
\end{figure}

\begin{table*}
	\caption{\label{tab:POVMtypical}
Operational characterizations of typical quantum measurements in terms of extremal one-copy and two-copy estimation fidelities. Each measurement (pair of measurements) in the left column is completely determined by the estimation fidelities in the bold font in the corresponding row (at most three estimation fidelities are required in each case). Here $\scrA$ and $\scrB$ are POVMs on $\caH$; $F(\scrA)$ and $F(\scrB)$ are one-copy estimation fidelities, while $F(\scrA^{\otimes 2})$, $F(\scrB^{\otimes 2})$, and $F(\scrA\otimes \scrB)$ are two-copy estimation fidelities. The estimation fidelities $F_2^{\iid}$ and $F_2^{\sep}$ are defined in \esref{eq:F2iid}{eq:FidMUB}, respectively.
	}	
	\begin{math}
	\begin{array}{c|ccccc}
	\hline\hline
	\mbox{Quantum measurements}	&F(\scrA) &F(\scrA^{\otimes 2}) & F(\scrB) & F(\scrB^{\otimes 2}) & F(\scrA\otimes \scrB)\\[0.5ex]
	\hline
	\mbox{rank-1}& \bm{\frac{2}{d+1}} &- &- & -& -  \\[0.5ex]	
	\mbox{rank-1 projective}& \bm{\frac{2}{d+1}} &\bm{\frac{2}{d+1}} &- & -& -  \\[0.5ex]
		
	\mbox{identical rank-1 projective}& \bm{\frac{2}{d+1}} &\frac{2}{d+1} & \bm{\frac{2}{d+1}}& \frac{2}{d+1} & \bm{\frac{2}{d+1}}  \\[0.5ex]	
	
	\mbox{SIC}& \frac{2}{d+1} &\bm{F_2^{\iid}} & -&- & -  \\[0.5ex]
	\mbox{identical SICs}& \frac{2}{d+1} &\bm{F_2^{\iid}} &\frac{2}{d+1} & F_2^{\iid} & \bm{F_2^{\iid}}  \\[0.5ex]	
	
\mbox{MU rank-1 projective}& \frac{2}{d+1} &\bm{\frac{2}{d+1}} & \frac{2}{d+1}& \bm{\frac{2}{d+1}}&  \bm{F_2^{\sep}} \\[0.5ex]
	\hline\hline
	\end{array}	
	\end{math}	
\end{table*}

When $\scrA$ is a SIC, by virtue of \lref{lem:OneCopyFidLBUB}, \eref{eq:FPhalfSIC}, and \lref{lem:Fid2copyFP}, we can deduce that
\begin{align}\label{eq:FidSIC}
F(\scrA)=\frac{2}{d+1},\quad F(\scrA^{\otimes 2})=F_2^{\iid},
\end{align}
where
\begin{gather}\label{eq:F2iid}
F_2^{\iid}:=\frac{2[d^2+d+1 +(d-1)\sqrt{d+1}\lsp]}{d(d+1)(d+2)}. 
\end{gather}
On the other hand, $F(\scrA^{\otimes 3})$ depends on the specific SIC under consideration, as we shall see in \sref{sec:IneqSIC}.
Although rank-1 projective measurements and SICs share the same single-copy estimation fidelity, their two-copy estimation fidelities are quite different, as illustrated in \fref{fig:fidiid}. What is remarkable is that both rank-1 projective measurements and SICs are completely characterized by  one-copy and two-copy estimation fidelities, as shown in \thref{thm:Fid2copyProjSIC} and its corollaries below (cf. \tref{tab:POVMtypical}).
\begin{thm}\label{thm:Fid2copyProjSIC}
Let $\scrA$ be any rank-1 POVM 	
on $\caH$.	Then the two-copy estimation fidelity $F(\scrA^{\otimes 2})$ satisfies 
\begin{equation}\label{eq:Fid2copyLBUB}
\frac{2}{d+1}\leq F(\scrA^{\otimes 2})\leq F_2^{\iid}.
\end{equation}
If $\scrA$ is a simple rank-1 POVM, then the lower bound is saturated iff $\scrA$ is a rank-1 projective measurement, while 
the upper bound is saturated iff $\scrA$ is a SIC.			
\end{thm}
\Thref{thm:Fid2copyProjSIC} follows from \lsref{lem:POVMFP} and  \ref{lem:Fid2copyFP}. It highlights  special and intriguing roles played by rank-1 projective measurements and SICs in quantum state estimation. Note that rank-1 projective measurements are completely characterized by the lower bound for the two-copy estimation fidelity $F(\scrA^{\otimes 2})$ as presented in \eref{eq:Fid2copyLBUB}, while SICs are completely characterized by the upper bound. Furthermore, the "rank-1" assumption in \thref{thm:Fid2copyProjSIC} can be dropped thanks to \crsref{cor:12copyProj} and \ref{cor:Fid2copyUBmix} below. 
\begin{corollary}\label{cor:12copyProj}
	A simple POVM $\scrA$   on $\caH$ is a rank-1 projective measurement iff it satisfies the following condition $F(\scrA^{\otimes 2})=F(\scrA)=2/(d+1)$. 
\end{corollary}

\begin{corollary}\label{cor:rank1Projidentical}
	Two simple POVMs  $\scrA$ and $\scrB$ on $\caH$ are identical rank-1 projective measurements up to relabeling iff	they satisfy
	$F(\scrA\otimes \scrB)=F(\scrB)=F(\scrA)= 2/(d+1)$. 
\end{corollary}

\begin{corollary}\label{cor:Fid2copyUBmix}
	Any POVM $\scrA$ on $\caH$ satisfies the inequality $F(\scrA^{\otimes 2})\leq F_2^{\iid}$.	
	If $\scrA$ is simple, then the upper bound 
	 is saturated iff $\scrA$ is a SIC. 
\end{corollary}

\begin{corollary}\label{cor:fidgSIC}
	Suppose $\scrA_1,\scrA_2,\ldots, \scrA_g$ are $g$  POVMs on $\caH$. Then 
	\begin{align}\label{eq:fidgSIC}
	\sum_{r,s} F(\scrA_r\otimes \scrA_s)\leq g^2 F_2^{\iid}.
	\end{align}
	If in addition these POVMs are simple, then 
 the upper bound is saturated iff $\scrA_1,\scrA_2,\ldots, \scrA_g$ are  identical SICs up to relabeling. 
\end{corollary}

\begin{corollary}\label{cor:SICidentical}
Two simple POVMs  $\scrA$ and $\scrB$ on $\caH$ are identical SICs up to relabeling iff	they satisfy the condition
	$F(\scrA\otimes \scrB)=F(\scrA^{\otimes 2})= F_2^{\iid}$. 
\end{corollary}

\Crref{cor:12copyProj} follows from \lref{lem:OneCopyFidLBUB} and  \thref{thm:Fid2copyProjSIC}, while \crref{cor:rank1Projidentical} follows from 
\lsref{lem:POVMcrossFP}, \ref{lem:OneCopyFidLBUB}, and  \ref{lem:FidprodPOVM} (cf. \thref{thm:FidProdPOVMLB} below). The two corollaries  offer  succinct operational characterizations of  rank-1 projective measurements and identical  rank-1 projective measurements via one-copy and two-copy estimation fidelities.
\Crsref{cor:Fid2copyUBmix}-\ref{cor:SICidentical}
also follow from  \thref{thm:Fid2copyProjSIC}   as shown in \aref{asec:Fid2copyProof}; they   offer  even more succinct operational characterizations of  SICs and identical SICs  via  two-copy estimation fidelities. Surprisingly, here we do not need any assumption on the rank, purity, or the number of POVM elements. 
Note that the  isotropic measurement and measurements based on CMUB are suboptimal as illustrated in \fref{fig:fidiid}, although they have (many) more outcomes, in sharp contrast with the conclusion presented in \thref{thm:nCopyFidUB} and the results on traditional quantum state tomography \cite{Scot06,ZhuE11,Zhu14IOC}, in which measurements based on higher $t$-designs cannot perform worse. In addition, here the characterization of SICs is much simpler than most known   alternatives, including conventional characterizations based on maximal equiangular lines or minimal 2-designs \cite{Zaun11,ReneBSC04,Scot06,ApplFZ15G}.

The above results demonstrate that both rank-1 projective measurements and SICs  are uniquely determined by their information-extraction capabilities. In other words, these elementary quantum measurements can be defined in purely information theoretic terms, in sharp contrast with traditional algebraic definitions, which lack clear operational meanings. As far as we know, similar results have never been derived or even anticipated before.

\subsection{\label{sec:EAL}Measurements based on tight equiangular lines}
Here we show that measurements based on tight equiangular lines stand out as optimal measurements when the number of outcomes is limited, which generalizes the optimality result on SICs as stated in \thref{thm:Fid2copyProjSIC}.

If $\scrA=\{A_j\}_{j=1}^{m}$ is a POVM  on $\caH$ that is constructed from a set of $m$ tight equiangular lines, then the POVM elements $A_j$ of $\scrA$ have the form $A_j=d|\psi_j\>\<\psi_j|/m$, where $|\psi_j\>$ satisfy \eref{eq:EAL}. 
The estimation fidelity of $\scrA^{\otimes 2}$ can be derived by  virtue of  \lsref{lem:POVMEAL} and  \ref{lem:Fid2copyFP} [cf. \eref{eq:FidPOVM2} and \lref{lem:Sym2ProjNorm}], with the result
\begin{align}
F(\scrA^{\otimes 2})
&=\frac{2}{d+2}+\frac{2d+2\sqrt{d(m-1)(m-d)}
}{m(d+1)(d+2)}.
\end{align}
Moreover, such a POVM is optimal among all rank-1 POVMs with $m$ POVM elements, as shown in the following corollary, which is an immediate consequence of  \lsref{lem:POVMEAL} and  \ref{lem:Fid2copyFP}. 
\begin{corollary}\label{cor:MaxFidEAL}
	Suppose  $\scrA$ is a rank-1 POVM on $\caH$ that is composed of $m$ POVM elements; then the two-copy estimation fidelity $F(\scrA^{\otimes 2})$ satisfies 
	\begin{align}\label{eq:MaxFidEAL}
	F(\scrA^{\otimes 2})&\leq \frac{2}{d+2}+\frac{2d+2\sqrt{d(m-1)(m-d)}
	}{m(d+1)(d+2)}.
	\end{align}
	The upper bound is saturated iff $d\leq m\leq d^2$ and  the POVM  $\scrA$ is unbiased and equiangular. 
\end{corollary}
Note that the upper bound in \eref{eq:MaxFidEAL} is strictly monotonically increasing in $m$ as expected. 
As an implication of \lref{lem:OneCopyFidLBUB} and \crref{cor:MaxFidEAL}, an $m$-outcome POVM $\scrA$ is an unbiased rank-1 equiangular POVM iff $F(\scrA)=2/(d+1)$ and the inequality in \eref{eq:MaxFidEAL} is saturated.

\subsection{Mutually unbiased measurements}
When $\scrA$ and $\scrB$ are MU rank-1 projective measurements, the two-copy estimation fidelity $F(\scrA\otimes\scrB)$ can be computed using  \lsref{lem:POVMcrossFP} and \ref{lem:FidprodPOVM}, with the result
\begin{align}\label{eq:FidMUB}
F(\scrA\otimes\scrB)=F_2^{\sep}:=\frac{2(d+1+\sqrt{d}\lsp)}{(d+1)(d+2)}.
\end{align}
As illustrated in \fref{fig:fid4settings}, this estimation fidelity is very close to the estimation fidelity achieved by independent and identical measurements based on a SIC. Nevertheless, it turns out this is the maximum estimation fidelity achievable by separable measurements,  including all measurements realized by local operations and classical communication (LOCC), as manifested in the notation $F_2^{\sep}$. Recall that a POVM is separable if each POVM element is proportional to a separable density operator. 
\begin{lem}\label{lem:FidSepUB}
Let $\scrA=\{A_j\}_j$ be any separable POVM  on $\caH^{\otimes2}$. Then the two-copy estimation fidelity $F(\scrA)$  satisfies $F(\scrA)\leq F_2^{\sep}$.
If $\scrA$ is  rank-1, then the upper bound is saturated iff each POVM element $A_j$ satisfies
the condition 
$d\tr(WA_j)=\tr(A_j)$, where $W=2P_2-1$ is the swap operator. 
\end{lem}
When the POVM element $A_j$ is a tensor product, the equality $d\tr(WA_j)=\tr(A_j)$ means  $A_j$ is a tensor product of two MU positive operators, which is  reminiscent of the definition of MUMs.   If we only consider product measurements, then  only MU  rank-1 POVMs can attain the maximum estimation fidelity $F_2^{\sep}$.

\begin{thm}\label{thm:FidProdUB}
Suppose	$\scrA$ and $\scrB$ are two POVMs  on $\caH$. Then  the two-copy estimation fidelity $F(\scrA\otimes\scrB)$ satisfies $F(\scrA\otimes \scrB)\leq F_2^{\sep} $,
	and the inequality is saturated iff  $\scrA$ and $\scrB$ are MU rank-1 POVMs. 
\end{thm}

Although optimal product POVMs on $\caH^{\otimes 2}$ are necessarily rank 1, it should be noted that optimal separable POVMs are not subjected to this constraint.  In the case of a qubit for example, an  optimal separable POVM can be constructed from the following four rank-2 operators,
\begin{equation}
\begin{aligned}
\frac{|0+\rangle\<{0+}|+|{+0}\rangle\<+0|}{2}, \quad \frac{|0-\rangle\<{0-}|+|{-0}\rangle\<-0|}{2},\\
\frac{|1+\rangle\<{1+}|+|{+1}\rangle\<+1|}{2}, \quad \frac{|1-\rangle\<{1-}|+|{-1}\rangle\<-1|}{2},
\end{aligned}
\end{equation}
where $|\pm\>=(|0\>\pm |1\>)/\sqrt{2}$.

\begin{figure}
	\includegraphics[width=7cm]{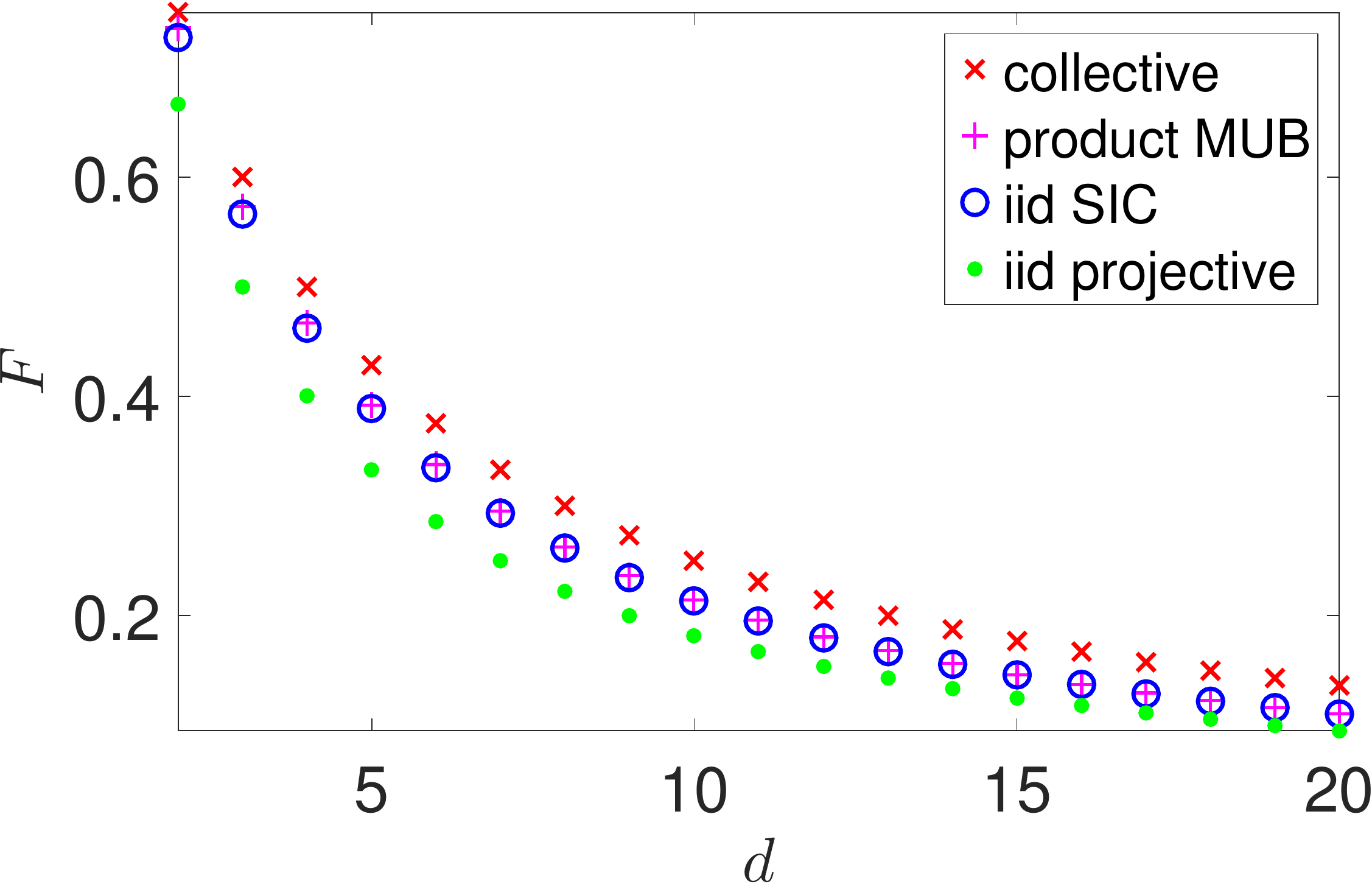}
	\caption{\label{fig:fid4settings}
		Two-copy estimation fidelities achieved by optimal collective measurements, optimal product measurements (based on MUB), and optimal independent and identical  measurements (based on a SIC), respectively. 	
		As a benchmark, the figure also shows the estimation fidelity achieved by  independent and identical rank-1 projective measurements. 
	}
\end{figure}

\begin{corollary}\label{cor:MaxFidMinProdPOVM}
	Suppose $\scrA$ and $\scrB$ are two POVMs on $\caH$ that satisfy the condition $F(\scrA\otimes \scrB)= F_2^{\sep}$. Then $\scrA\otimes \scrB$ has at least $d^2$ POVM elements, and the lower bound is saturated iff $\scrA$ and $\scrB$ are MU rank-1 projective measurements. 
\end{corollary}

\begin{corollary}\label{cor:FidToMUB}
Suppose	$\scrA$ and $\scrB$ are two simple POVMs  on $\caH$. Then $\scrA$ and $\scrB$
are MU rank-1 projective measurements iff   they satisfy the conditions $F(\scrA\otimes \scrB)= F_2^{\sep}$ and $F(\scrA^{\otimes 2})=F(\scrB^{\otimes 2})=2/(d+1)$. 
\end{corollary}

\begin{corollary}\label{cor:MaxFMUB}
	Suppose $\scrA_1,\scrA_2,\ldots, \scrA_g$ are $g$  POVMs on $\caH$. Then 
	\begin{align}\label{eq:FidCMUB}
	\sum_{r\neq s} F(\scrA_r\otimes \scrA_s)\leq g(g-1) F_2^{\sep},
	\end{align}
	and the  upper bound is saturated iff $\scrA_1,\scrA_2,\ldots, \scrA_{g}$ are MU rank-1 POVMs. If in addition  these POVMs are simple and
	 $g=d+1$, then the upper bound is saturated iff
 $\scrA_1,\scrA_2,\ldots, \scrA_{d+1}$ 
 are rank-1 projective measurements, which form a CMUMs.  
\end{corollary}

Thanks to \thref{thm:FidProdUB} and its corollaries, MU rank-1 projective measurements are completely characterized by  two-copy estimation fidelities as summarized in \tref{tab:POVMtypical}. Remarkably,
 we do not even need bases to start with; orthonormal bases appear naturally once the estimation fidelities reach certain extremal values. 
 In addition, \crsref{cor:FidToMUB} and \ref{cor:MaxFMUB} do not require any assumption on the rank, purity, or the  number of POVM elements. Such simple operational characterizations are not even anticipated in the literature as far as we know. These results are of intrinsic interest to studying quantum estimation theory, quantum measurements, and the complementarity principle.

At this point, it is instructive to compare \crref{cor:fidgSIC} with \crref{cor:MaxFMUB}. The former characterizes identical SICs via the maximum of $\sum_{r, s} F(\scrA_r\otimes \scrA_s)$, which is a sum of two-copy estimation fidelities, while the latter characterizes (complete sets of) MUMs via the maximum of $\sum_{r\neq s} F(\scrA_r\otimes \scrA_s)$. The only difference in the latter summation is that the diagonal terms are absent. Quite surprisingly, this minor change in the summation leads to a jump from SICs to MUMs.

\section{\label{sec:incompatibility}Decoding quantum incompatibility}    
In addition to characterizing typical quantum measurements, the estimation fidelity encodes valuable information about noncommutativity and incompatibility of quantum measurements. Traditionally, incompatibility is usually understood as  a limitation, as embodied in the complementarity principle \cite{Bohr28} and uncertainty relations
\cite{Heis27,Robe29,BuscLW14,WehnW10,ColeBTW17}. With the development of quantum information science, incompatibility is also recognized as a resource \cite{HeinoMZ16,GuhnHKP21}. 
As a byproduct, here we show that incompatibility is a useful resource to enhance the estimation fidelity, thereby offering additional insight on this topic.

\subsection{Quantum incompatibility and the estimation fidelity}

\begin{thm}\label{thm:FidCommPOVM}
Suppose $\scrA$ and $\scrB$ are two  commuting POVMs on $\caH$. Then the two-copy estimation fidelity $F(\scrA\otimes\scrB)$	satisfies $F(\scrA\otimes \scrB)\leq 2/(d+1)$. 
\end{thm}

\Thref{thm:FidCommPOVM} shows that  two-copy measurements based on two commuting POVMs cannot provide any advantage over one-copy measurements. In other words,  noncommutativity is necessary to go beyond the single-copy limit on the estimation fidelity. \Thref{thm:FidProdPOVMLB} below further shows that noncommutativity is also  sufficient to achieve this goal when one POVM is rank 1, as illustrated in \fref{fig:fidQubitPOVM}.

\begin{thm}\label{thm:FidJMPOVM}
Suppose  $\scrA$ and $\scrB$ are two compatible POVMs on $\caH$.	Then the two-copy estimation fidelity $F(\scrA\otimes \scrB)$ satisfies $F(\scrA\otimes\scrB )\leq F_2^{\iid}$. If in addition both $\scrA$ and $\scrB$ are simple, then the upper bound is saturated iff $\scrA$ and $\scrB$ are identical SICs up to relabeling. 			
\end{thm}
Here $F_2^{\iid}$ is  defined in \eref{eq:F2iid}. \Thref{thm:FidJMPOVM} shows that  two-copy measurements based on two compatible POVMs cannot go beyond the estimation fidelity achieved by product measurements based on identical  SICs. It offers another operational characterization of SICs that does not rely on any assumption on the rank, purity, or the number of measurement outcomes; meanwhile, it provides a universal criterion for detecting incompatibility of two POVMs. Such universal criteria are quite rare in the literature \cite{Zhu15IC,ZhuHC16,HeinJN22}.
These  results highlight the intriguing connection between SICs and quantum incompatibility, which is of intrinsic interest to foundational studies. 

Suppose $\scrA_1, \scrA_2, \ldots, A_g$ are $g$ compatible POVMs. Then \thref{thm:FidJMPOVM} implies that 
\begin{align}
\sum_{r\neq s}^g F(\scrA_r\otimes \scrA_s)\leq g(g-1)F_2^{\iid}.
\end{align}
In addition, since $\scrA_1, \scrA_2, \ldots, A_g$ admit a common refinement, say $\scrB$, we can also deduce that 
\begin{align}
F(\scrA_1\otimes \scrA_2 \otimes \cdots \otimes \scrA_r)\leq F(\scrB^{\otimes g})\leq F_g^{\iid},
\end{align}
where $F_g^{\iid}$ denotes the maximum estimation fidelity achieved by identical and independent measurements on $\caH^{\otimes g}$. This result provides a universal criterion for detecting   incompatibility of $g$ arbitrary POVMs.  Unfortunately, it is not easy to determine $F_g^{\iid}$ for $g\geq 3$; this problem deserves further study.

\begin{thm}\label{thm:FidProdPOVMLB}
	Suppose  $\scrA$ and $\scrB$ are  POVMs on $\caH$ with $\scrA$ being rank 1. Then the two-copy estimation fidelity $F(\scrA\otimes\scrB)$	satisfies $F(\scrA\otimes \scrB)\geq F(\scrA)=2/(d+1)$, and the inequality is saturated iff $\scrA$ commutes with $\scrB$.  If $\scrA$ and $\scrB$ are simple rank-1 POVMs,  then the inequality is saturated iff $\scrA$ and $\scrB$ are identical rank-1 projective measurements up to relabeling. 
\end{thm}
The last statement in \thref{thm:FidProdPOVMLB} is tied to the characterization of identical rank-1 projective measurements presented in \crref{cor:rank1Projidentical}. 
As an implication of \thsref{thm:FidProdUB} and \ref{thm:FidProdPOVMLB}, any pair of rank-1 POVMs $\scrA$  and $\scrB$ on $\caH$ satisfies 
\begin{align}
\frac{2}{d+1}\leq F(\scrA\otimes \scrB)\leq F_2^{\sep}. 
\end{align}
The lower bound is saturated iff $\scrA$ and $\scrB$ commute, which means they are equivalent to a same rank-1 projective measurement, while the upper bound is saturated iff $\scrA$ and $\scrB$ are MU. Note that MU measurements are often regarded as maximally incompatible measurements \cite{Schw60,Ivan81,WootF89,DurtEBZ10,DesiSFB19}. 
The above results show that incompatibility is a resource to enhance the estimation fidelity. Although the significance of incompatibility as a resource has been recognized before, results like \thref{thm:FidProdPOVMLB} are still quite rare because it is not easy to establish conditions that are both necessary and sufficient.

\subsection{Concrete examples}

As an illustration, let us consider two  binary POVMs $\scrA=\{A_+, A_-\}$ and $\scrB=\{B_+, B_-\}$ acting on  a qubit ($d=2$), where $A_-=1-A_+$ and $B_-=1-B_+$. 
Note that the two POVMs are completely determined by the two  effect operators $A_+$ and $B_+$, respectively,  which satisfy the condition $0\leq A_+, B_+\leq 1$. Without loss of generality, we can assume that $1\leq \tr(A_+),\tr(B_+)<2$. Then 
 $A_\pm$ and $B_\pm$ can be expressed as 
\begin{align}
A_\pm=\frac{1\pm\alpha\pm\vec{a}\cdot \vec{\sigma}}{2},\quad B_\pm=\frac{1\pm\beta\pm\vec{b}\cdot \vec{\sigma}}{2},
\end{align}
where $\vec{\sigma}=(\sigma_x,\sigma_y,\sigma_z)$ is the vector composed of the three Pauli operators and the parameters $\alpha, \beta, \vec{a},\vec{b}$ satisfy the conditions
\begin{align}
0\leq \alpha,\beta< 1,\quad  |\vec{a}|\leq 1-\alpha, \quad |\vec{b}|\leq 1-\beta. 
\end{align}
The parameter $\alpha$ ($\beta$) characterizes the bias of the POVM $\scrA$ ($\scrB$), while the parameter $|\vec{a}|$ ($|\vec{b}|$) characterizes the sharpness of  $\scrA$ ($\scrB$). 
Notably,  $\scrA$ is unbiased iff $\alpha=0$, while
$\scrA$ is rank 1 iff $\alpha=0$ and $|\vec{a}|=1$. Similarly, $\scrB$ is unbiased iff $\beta=0$, 
while $\scrB$ is rank 1 iff $\beta=0$ and $|\vec{b}|=1$.
In addition, $\scrA$ and $\scrB$ are MU iff $\vec{a}\cdot \vec{b}=0$.

The single-copy estimation fidelity $F(\scrA)$ can be computed by virtue of \lref{lem:OneCopyFidLBUB}, with the result
\begin{align}
F(\scrA)=\frac{3+|\vec{a}|}{6}. 
\end{align}
The two-copy estimation fidelity $F(\scrA^{\otimes 2})$  can be computed  by virtue of \esref{eq:FidPOVM2}{eq:QAB}, with the result
\begin{align}
F(\scrA^{\otimes 2})=\frac{3+|\vec{a}|+\alpha |\vec{a}|}{6}.
\end{align}
Interestingly, the single-copy estimation fidelity of a  binary POVM on a qubit is completely determined by its sharpness, while the two-copy estimation fidelity  depends on both    sharpness and bias. In addition, the parameters $\alpha$ and $|\vec{a}|$ are completely determined by $F(\scrA)$ and $F(\scrA^{\otimes 2})$.
Notably, the inequality $F(\scrA^{\otimes 2})\geq F(\scrA)$ is saturated iff $\scrA$ is unbiased ($\alpha=0$) or trivial ($|\vec{a}|=0$).

To determine the estimation fidelity of the product POVM $\scrA\otimes \scrB$, we need to compute $\caQ(A_\pm\otimes B_\pm)$.  By virtue of \eref{eq:QAB} we can derive the following result,
\begin{align}
\caQ(A_+\otimes B_+)&= 3(1+\alpha)(1+\beta)+\vec{a}\cdot \vec{b}+(1+\beta)\vec{a}\cdot \vec{\sigma} \nonumber\\
&\quad +(1+\alpha) \vec{b}\cdot \vec{\sigma},
\end{align}
which implies that
\begin{align}
\|\caQ(A_+\otimes B_+)\|&=3(1+\alpha)(1+\beta)+\vec{a}\cdot \vec{b}\nonumber\\
&\quad+|(1+\beta)\vec{a} +(1+\alpha) \vec{b}|. 
\end{align}
The norm $\|\caQ(A_+\otimes B_-)\|$ can be derived by replacing $\beta$ and $\vec{b}$ with $-\beta$ and $-\vec{b}$, respectively; a similar recipe applies to $\|\caQ(A_-\otimes B_\pm)\|$.  Now the estimation fidelity $F(\scrA\otimes \scrB)$ can be calculated using  \eref{eq:FidPOVM2}, with the result
\begin{align}
F(\scrA\otimes \scrB)&=\frac{1}{2}+\frac{1}{24}\bigl[\lsp|(1+\beta)\vec{a} +(1+\alpha) \vec{b}|\nonumber\\
&\quad+|(1-\beta)\vec{a} -(1+\alpha) \vec{b}|\nonumber\\
&\quad +|(1+\beta)\vec{a} -(1-\alpha) \vec{b}|\nonumber\\
&\quad +|(1-\beta)\vec{a} +(1-\alpha) \vec{b}|\lsp\bigr]. \label{eq:FidPOVMbinary}
\end{align}

\begin{figure}
	\includegraphics[width=7.5cm]{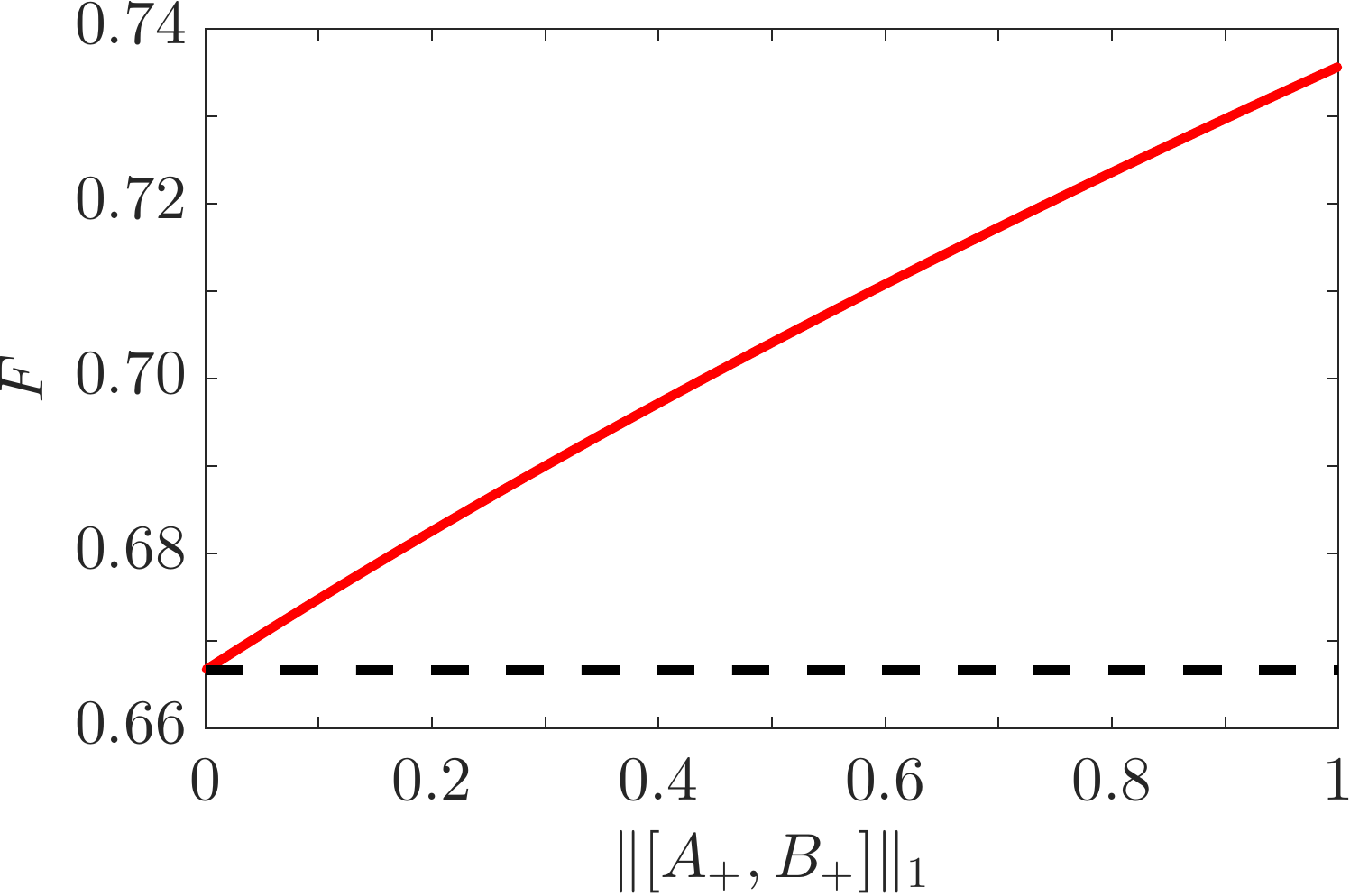}
	\caption{\label{fig:fidQubitPOVM}
		Relation between the  two-copy estimation fidelity 	and the commutator of POVM elements. Here $F=F(\scrA\otimes\scrB)$ is the estimation fidelity achieved  by the tensor product of two rank-1 projective measurements as shown in \eref{eq:Fab}, while $\|[A_+,B_+]\|_1$ is the  1-norm of the commutator $[A_+,B_+]$. 
		As a benchmark, the black dashed line denotes the one-copy estimation fidelity achieved by a  rank-1 projective measurement.
	}
\end{figure}

If $\scrA$ and $\scrB$ are unbiased, which means $\alpha=\beta=0$, then \eref{eq:FidPOVMbinary} reduces to 
\begin{align}
&F(\scrA\otimes \scrB)=\frac{1}{2}+\frac{|\vec{a}+\vec{b}| +|\vec{a}-\vec{b}| }{12}.
\end{align}
Note that $F(\scrA\otimes \scrB)\leq 2/3$ iff $|\vec{a}+\vec{b}| +|\vec{a}-\vec{b}|\leq 2$.  It is well known that the latter condition  holds iff $\scrA$  and $\scrB$ are compatible (jointly measurable) \cite{Busc86,StanRH08,BuscS10,YuLLO10}. Therefore,  incompatibility is both necessary and sufficient to enhance the estimation fidelity beyond the single-copy limit in this special case.  In addition, we have $F(\scrA\otimes \scrB)=\max\{F(\scrA),F(\scrB)\}$
when $\vec{a}$ and $\vec{b}$ are parallel or antiparallel, which is consistent with \thsref{thm:FidCommPOVM} and \ref{thm:FidProdPOVMLB}.

If $\scrA$ and $\scrB$ are rank-1 projective measurements, which means $\alpha=\beta=0$ and $|\vec{a}|=|\vec{b}|=1$, then   $F(\scrA\otimes \scrB)$ can  be expressed as 
\begin{align}
F(\scrA\otimes \scrB)
&=\frac{3+\sqrt{1+|\vec{a}\times \vec{b}|}}{6}=\frac{3+\sqrt{1+|\sin\phi|}}{6}
\nonumber\\
&=\frac{3+\sqrt{1+\|[A_+,B_+]\|_1}}{6}\label{eq:Fab}
, 
\end{align}
where $\phi$ is the angle between $\vec{a}$ and $\vec{b}$, and $\|[A_+,B_+]\|_1$ is the Schatten 1-norm (or trace norm) of the commutator  $[A_+,B_+]:=A_+ B_+ - B_+ A_+$. Note that 
\begin{align}
[A_+,B_+]=[A_-,B_-]=-[A_+,B_-]=-[A_-,B_+]. 
\end{align}
Therefore, $F(\scrA\otimes \scrB)\geq 2/3$, and the lower bound is saturated iff $\vec{a}$ and $\vec{b}$ are parallel or antiparallel, in which case  $\scrA$ and $\scrB$ commute and are identical rank-1 projective measurements up to relabeling, as shown in \thref{thm:FidProdPOVMLB} and illustrated in \fref{fig:fidQubitPOVM}.

\subsection{Connection with entropic uncertainty relations}

Entropic uncertainty relations are another important manifestation of quantum incompatibility \cite{WehnW10,ColeBTW17}. Given two rank-1 projective measurements $\scrA$ and $\scrB$ on $\caH$, is there any connection between  the two-copy estimation fidelity $F(\scrA\otimes \scrB)$ and entropic uncertainty relations between $\scrA$ and $\scrB$? Here we shall reveal a precise connection in the case of a qubit. 
Denote by $H_\rho(\scrA)$ [$H_\rho(\scrB)$] the entropy of measurement outcomes when the projective measurement $\scrA$ ($\scrB$) is performed on a given state $\rho$. Then the entropy sum $H_\rho(\scrA)+H_\rho(\scrB)$ satisfies a state-independent entropic uncertainty relation \cite{GhirMR03,WehnW10,ColeBTW17}, 
\begin{gather}
H_\rho(\scrA)+H_\rho(\scrB)\geq H_{\mathrm{mes}}(\scrA, \scrB)= H_{\mathrm{mes}}(\phi),  \label{eq:EntropicUR}
\end{gather}
where $H_{\mathrm{mes}}(\scrA\otimes \scrB)$ denotes the minimum entropy sum associated with the two projective measurements $\scrA, \scrB$, and $\phi$ is the angle between $\vec{a}$ and $\vec{b}$ as in \eref{eq:Fab}. Note that $H_{\mathrm{mes}}(\scrA,\scrB)$ is completely determined by $\phi$, so we can  write $H_{\mathrm{mes}}(\phi)$ in its place.

\begin{figure}[b]
	\includegraphics[width=7.5cm]{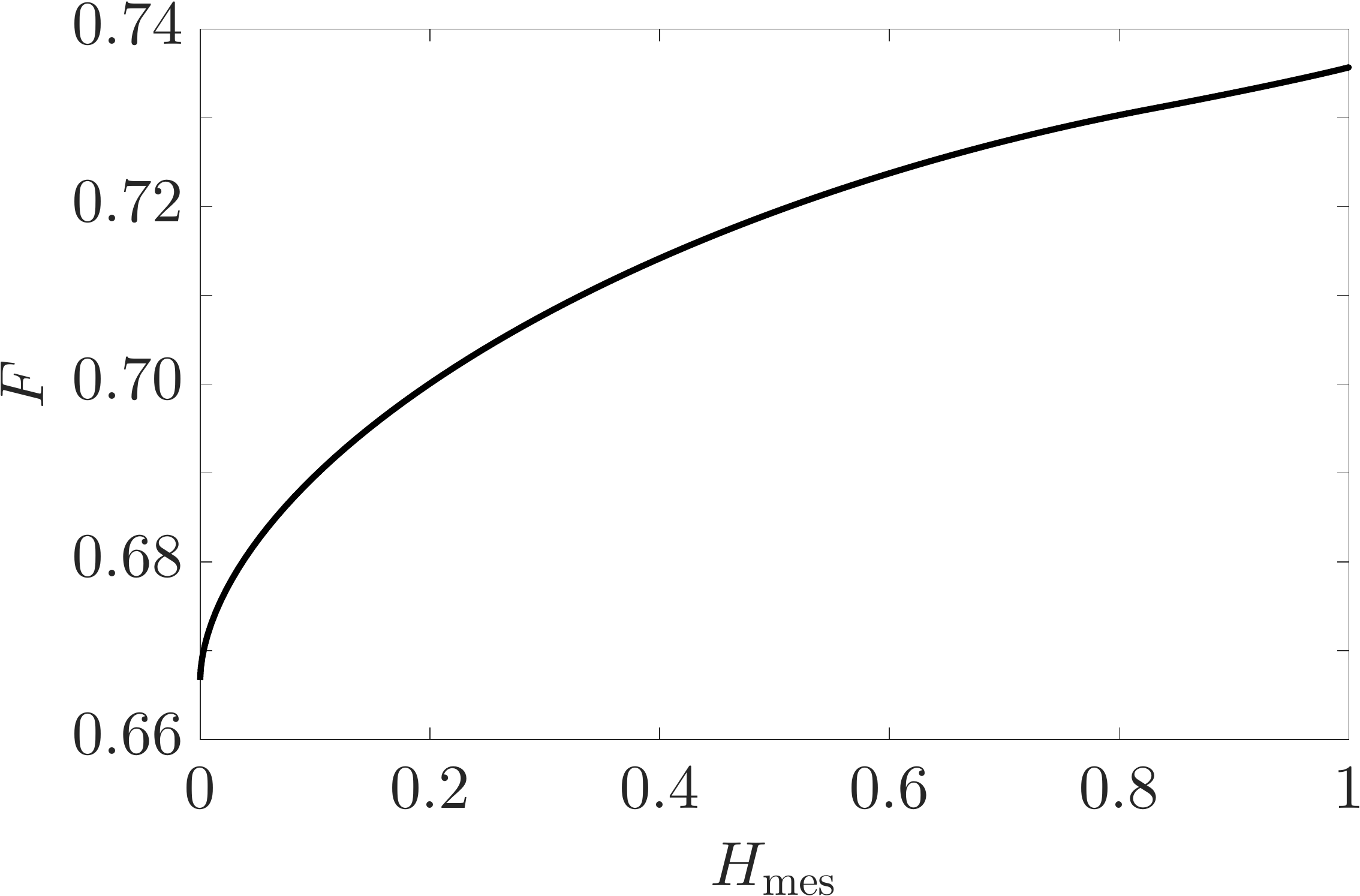}
	\caption{\label{fig:MesFid}
		Relation between the  two-copy estimation fidelity  and the minimum entropy sum. Here $F$ is the two-copy estimation fidelity associated with two rank-1  projective measurements on a qubit as shown in \eref{eq:Fab}, and $H_{\mathrm{mes}}$ is the  minimum entropy sum	presented in \eref{eq:MES}, which characterizes  the entropic uncertainty relation in \eref{eq:EntropicUR}.
	}
\end{figure}

Let 
\begin{gather}
p=\frac{1+\cos\theta }{2}, \quad q=\frac{1+\cos(\theta-\phi) }{2}.
\end{gather}
Then $H_{\mathrm{mes}}(\phi)$ can be expressed as  \cite{GhirMR03,ColeBTW17}
\begin{gather}
H_{\mathrm{mes}}(\phi)=\min_{0\leq \theta<2\pi}[h_{\mathrm{bin}}(p)+h_{\mathrm{bin}}(q)], \label{eq:MES}
\end{gather}
where $h_{\mathrm{bin}}(p)$ is the binary Shannon entropy defined as 
\begin{align}
h_{\mathrm{bin}}(p):=-p\log_2 p-(1-p)\log_2(1-p).
\end{align}
When $0\leq \phi\leq \phi_1$ with $\phi_1\approx 1.17056$, the minimum  in  \eref{eq:MES} is attained at $\theta=\phi/2$ \cite{GhirMR03}; in general, the minimum can be determined by numerical calculation. In addition, it is easy to verify that 
\begin{align}
H_{\mathrm{mes}}(-\phi)=H_{\mathrm{mes}}(\pi+\phi)=H_{\mathrm{mes}}(\phi).
\end{align}
So the value of $H_{\mathrm{mes}}(\phi)$ is determined by $|\sin\phi|$. Moreover, it is not difficult to show that $H_{\mathrm{mes}}(\phi)$ is monotonically increasing in  $|\sin\phi|$, just like  the two-copy estimation fidelity $F(\scrA\otimes \scrB)$ in  \eref{eq:Fab}.  Therefore, the minimum entropy sum  $H_{\mathrm{mes}}(\scrA,\scrB)$ associated with two rank-1 projective measurements on a qubit 
is determined by the two-copy estimation fidelity $F(\scrA\otimes \scrB)$, and vice versa, as illustrated in \fref{fig:MesFid}.  This observation reveals a surprising connection between quantum state estimation and entropic uncertainty relations.

\section{\label{sec:IneqMUBSIC}Distinguishing inequivalent MUB and SICs}  
In this section we show that the three-copy estimation fidelity can be used to distinguish inequivalent discrete symmetric structures tied to the quantum state space, including MUB and SICs in particular. 
This capability is rooted in the fact that  the three-copy estimation fidelity encodes valuable information about the triple products of POVM elements, which play important roles in studying SICs \cite{Zhu10,ApplFF11,ApplDF14} and discrete Wigner functions \cite{Woot87}. 
Note that such information cannot be retrieved by considering pairwise overlaps alone.

\subsection{\label{sec:IneqMUB}Operational distinction between  inequivalent  MUB} 

\begin{figure}[b]
	\includegraphics[width=7.5cm]{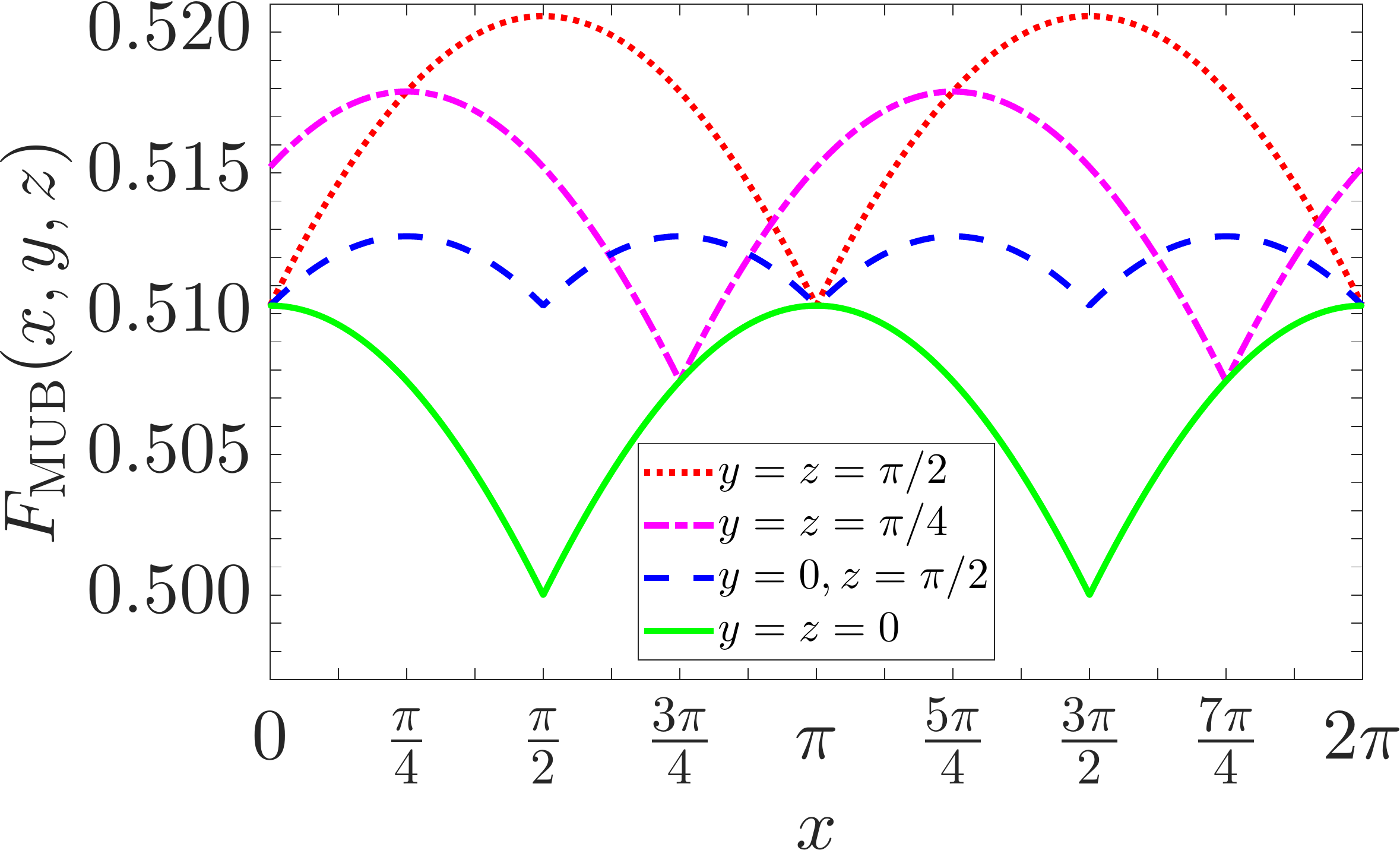}
	\caption{\label{fig:fidMUB1}
		Three-copy estimation fidelity  $F_{\mub}(x,y,z)$   achieved by the product projective measurement based on a triple of MUB as determined in  \eref{eq:Fid3MUB}. 
	}
\end{figure}

\begin{figure*}
	\includegraphics[width=14cm]{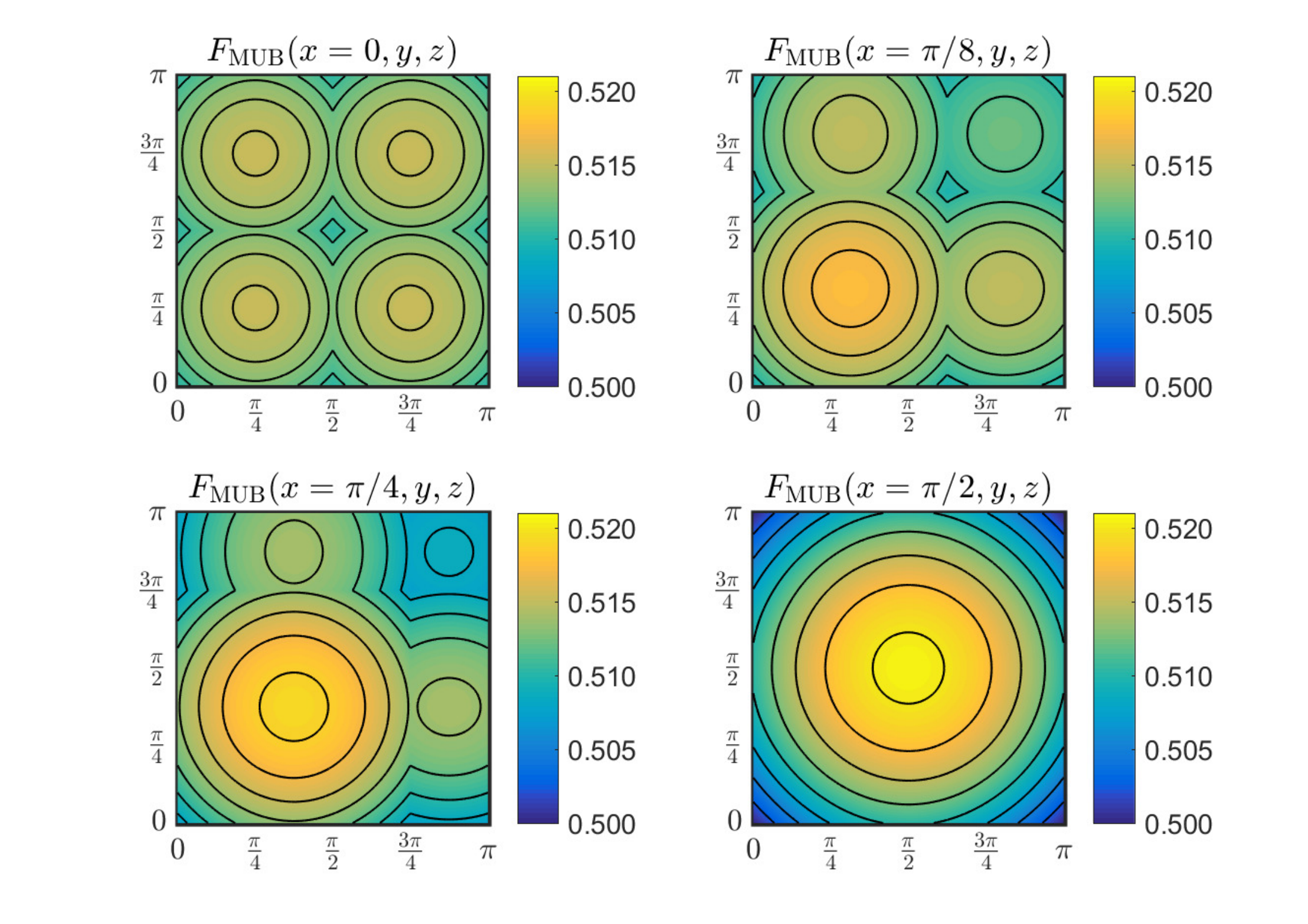}
	\caption{\label{fig:fidMUB2}Contour and color plots of the  three-copy estimation fidelity $F_{\mub}(x,y,z)$ [cf. \eref{eq:Fid3MUB}] for four cross sections that are parallel to the $yz$-plane as specified by  four different values of $x$. 
	}
\end{figure*}

Let $\{|\alpha_j\>\}_j$, $\{|\beta_k\>\}_k$, and $\{|\gamma_l\>\}_l$ be three orthonormal bases in $\caH$; let 
$\scrA=\{|\alpha_j\>\<\alpha_j|\}_j$, $\scrB=\{|\beta_k\>\<\beta_k|\}_k$, and  $\scrC=\{|\gamma_l\>\<\gamma_l|\}_l$ be the corresponding rank-1 projective measurements. In this section we take the convention that $j,k,l\in \{0,1,\ldots, d-1\}$. According to \esref{eq:FidPOVM2}{eq:QABCtrace1}, the three-copy estimation fidelity $F(\scrA\otimes \scrB\otimes \scrC)$ may depend on the triple products
\begin{align}
f_{jkl}:=\tr(|\alpha_j\>\<\alpha_j|\beta_k\>\<\beta_k|\gamma_l\>\<\gamma_l|)
\end{align}
in addition to the pairwise fidelities (transition probabilities) between the basis elements. This dependence can be used to distinguish inequivalent triples of bases that share the same pairwise fidelities.

As an illustration, here we consider triples of MUB in dimension 4 ($d=4$). In this case, according to \rcite{BrieWB10}, there exists a three-parameter family of tripes of MUB. The first basis is chosen to be the computational basis by convention; the second and third bases correspond to  the columns of  the two Hadamard matrices 
\begin{equation}\label{eq:Hadamard4}
\frac{1}{2}\begin{pmatrix}
1 & 1 & 1&1\\
1&1&-1&-1\\
1&-1& \rmi \rme^{\rmi x} &-\rmi \rme^{\rmi x}\\
1&-1& -\rmi \rme^{\rmi x} &\rmi \rme^{\rmi x}
\end{pmatrix},\quad 
\frac{1}{2}\begin{pmatrix}
1 & 1 & 1&1\\
1&1&-1&-1\\
- \rme^{\rmi y} &\rme^{\rmi y}& \rme^{\rmi z} &- \rme^{\rmi z}\\
 \rme^{\rmi y} &- \rme^{\rmi y}& \rme^{\rmi z} &- \rme^{\rmi z}
\end{pmatrix},
\end{equation}
where $x,y,z\in [0,2\pi)$ are three real parameters. Note that the transformation $x\mapsto x+\pi$ amounts to  the permutation of the last two columns of the first Hadamard matrix; similarly, the transformations $y\mapsto y+\pi$ and  $z\mapsto z+\pi$ amount to the permutations of the columns of the second Hadamard matrix. So it suffices to consider the parameter range $x,y,z\in [0,\pi)$.

Now suppose $\{|\alpha_j\>\}_j$ coincides with the computational basis $\{|j\>\}_j$, while $\{|\beta_k\>\}_k$ and $\{|\gamma_l\>\}_l$ are determined by the columns of the two Hadamard matrices  in \eref{eq:Hadamard4}, respectively; here the dependences on the parameters $x,y,z$ are suppressed  to simplify the notation.
Then the three bases $\{|\alpha_j\>\}_j$, $\{|\beta_k\>\}_k$, and $\{|\gamma_l\>\}_l$ are MU, that is,
\begin{align}
&\tr(|\alpha_j\>\<\alpha_j|\beta_k\>\<\beta_k|)=\tr(|\beta_k\>\<\beta_k|\gamma_l\>\<\gamma_l|)\nonumber\\
&=\tr(|\alpha_j\>\<\alpha_j|\gamma_l\>\<\gamma_l|)=\frac{1}{4}\quad \forall j,k,l=0,1,2,3. 
\end{align}
By  \esref{eq:FidPOVM2}{eq:QABCtrace1}, the three-copy estimation fidelity of $\scrA\otimes \scrB\otimes\scrC$, denoted by   $F_{\mub}(x,y,z):=F(\scrA\otimes \scrB\otimes\scrC)$ henceforth, can be computed as 
\begin{align}\label{eq:Fid3MUB}
F_{\mub}(x,y,z)\!=\!\frac{\sum_{j,k,l}\|\caQ(|j\>\<j|\otimes |\beta_k\>\<\beta_k |\otimes |\gamma_l\>\<\gamma_l|)\|}{840}.
\end{align}
The variation of the estimation fidelity $F_{\mub}(x,y,z)$ with $x,y,z$ is shown in \fsref{fig:fidMUB1} and \ref{fig:fidMUB2}.

In general, $F_{\mub}(x,y,z)$ does not have a simple  expression because  $\|\caQ(|j\>\<j|\otimes |\beta_k\>\<\beta_k |\otimes |\gamma_l\>\<\gamma_l|)\|$ depends on the triple product $f_{jkl}$ and in general does not have a simple analytical expression. Nevertheless, analytical formulas for $\|\caQ(|j\>\<j|\otimes |\beta_k\>\<\beta_k |\otimes |\gamma_l\>\<\gamma_l|)\|$ can be derived in a few special cases of interest,
 \begin{align}
 &\bigl\|\caQ(|j\>\<j|\otimes |\beta_k\>\<\beta_k|\otimes |\gamma_l\>\<\gamma_l|)\bigr\|\nonumber\\
 &=\begin{cases}
 \frac{15}{2} & \mbox{if } f_{jkl}=\frac{1}{8},\\[0.3ex]
 \frac{15}{4} & \mbox{if } f_{jkl}=-\frac{1}{8},\\[0.3ex]
4+\frac{5}{4}\sqrt{3} & \mbox{if } f_{jkl}=\pm \frac{\rmi}{8}.
 \end{cases}
 \end{align}
In the case $x=\pi/2$ and $y=z=0$, calculation shows that 48 of the triple products  $f_{jkl}$ are equal to $1/8$, while the remaining 16 triple products are equal to $-1/8$. Therefore,
\begin{align}
F_{\mub}(x,y,z)&=\frac{1}{840}\Bigl(\frac{15}{2}\times 48+\frac{15}{4}\times 16\Bigr)=\frac{1}{2}.
\end{align}
In the case $x=y=z=\pi/2$, calculation shows that 
32 of the triple products  $f_{jkl}$ are equal to $1/8$, while the remaining 32 triple products are equal to  $\rmi/8$ or $-\rmi/8$. Therefore,
\begin{align}
F_{\mub}(x,y,z)&=\frac{1}{840}\Bigl(\frac{15}{2}\times 32+\frac{16+5\sqrt{3}}{4}\times 32\Bigr)\nonumber\\
&=\frac{46+5\sqrt{3}}{105}\approx 0.5206.
\end{align}

Numerical calculation indicates that $1/2$ is the minimum of $F_{\mub}(x,y,z)$, while $(46+5\sqrt{3})/105$ is the maximum of $F_{\mub}(x,y,z)$ (cf. \fsref{fig:fidMUB1} and \ref{fig:fidMUB2}). The difference is about 4.1\%, which is quite significant and is amenable to experimental demonstration.  Here Haar random pure states involved in the estimation problem can be replaced by any ensemble of pure states that forms a 4-design, which can be constructed from a suitable Clifford orbit as described in \rcite{ZhuKGG16}.

\subsection{\label{sec:IneqSIC}Operational distinction between inequivalent SICs}

The estimation fidelity can also be used to distinguish inequivalent SICs. As an illustration, here we consider SICs in dimension 3. It is known that all SICs in dimension 3 are covariant with respect to the 
Heisenberg-Weyl group with respect to a suitable basis \cite{Zaun11,ReneBSC04,ScotG10,FuchHS17,Appl05,Zhu10,Szol14,HughS16}. The standard Heisenberg-Weyl group
 is generated by  the cyclic-shift operator $X$ and  phase operator $Z$ shown below
\begin{align}
X:=\begin{pmatrix}
0&0& 1\\
1&0&0\\
0&1&0
\end{pmatrix},\quad Z:=\begin{pmatrix}
1 & 0&0\\
&\rme^{2\pi\rmi/3}&0\\
0&0&\rme^{4\pi\rmi/3} 
\end{pmatrix}.
\end{align}
Let
\begin{equation}\label{eq:SICfiducial}
|\psi(\phi)\rangle:=\frac{1}{\sqrt{2}}(0,1,-\rme^{\rmi \phi})^\rmT,\quad 0\leq \phi< 2\pi;
\end{equation}
for each choice of the phase $\phi$, a SIC can be constructed as follows \cite{Zaun11,Appl05},
\begin{align}\label{eq:SIC3}
\scrA_{\sic}(\phi):=\Bigl\{\frac{1}{3 }X^j Z^k|\psi(\phi)\>\<\psi(\phi)|\bigl(X^j Z^k\bigr)^\dag\Bigr\}_{j,k=0,1,2}.
\end{align}
Note that $\scrA_{\sic}(\phi+(2\pi/3))$ and $\scrA_{\sic}(\phi)$ are identical up to relabeling. 

Moreover, any SIC in dimension 3 is unitarily equivalent to $\scrA_{\sic}(\phi)$ for $\phi\in [0,\pi/9]$; given $0\leq \phi_1\leq \phi_2\leq \pi/9$, then $\scrA_{\sic}(\phi_1)$ and $\scrA_{\sic}(\phi_2)$ are unitarily  equivalent iff $\phi_1=\phi_2$ \cite{Zhu10,Zhu12the}. The two SICs $\scrA_{\sic}(\phi=0)$ and $\scrA_{\sic}(\phi=\pi/9)$
are exceptional in the sense that they have larger symmetry groups compared with a generic SIC $\scrA_{\sic}(\phi)$ with $0<\phi<\pi/9$. In particular, the SIC $\scrA_{\sic}(\phi=0)$ has the largest symmetry group and can be regarded as the most symmetric SIC \cite{Appl05,Zhu10,Zhu15S}. However, it is not clear if inequivalent SICs have different operational implications before the current study. 

\begin{figure}[t]
	\includegraphics[width=7.5cm]{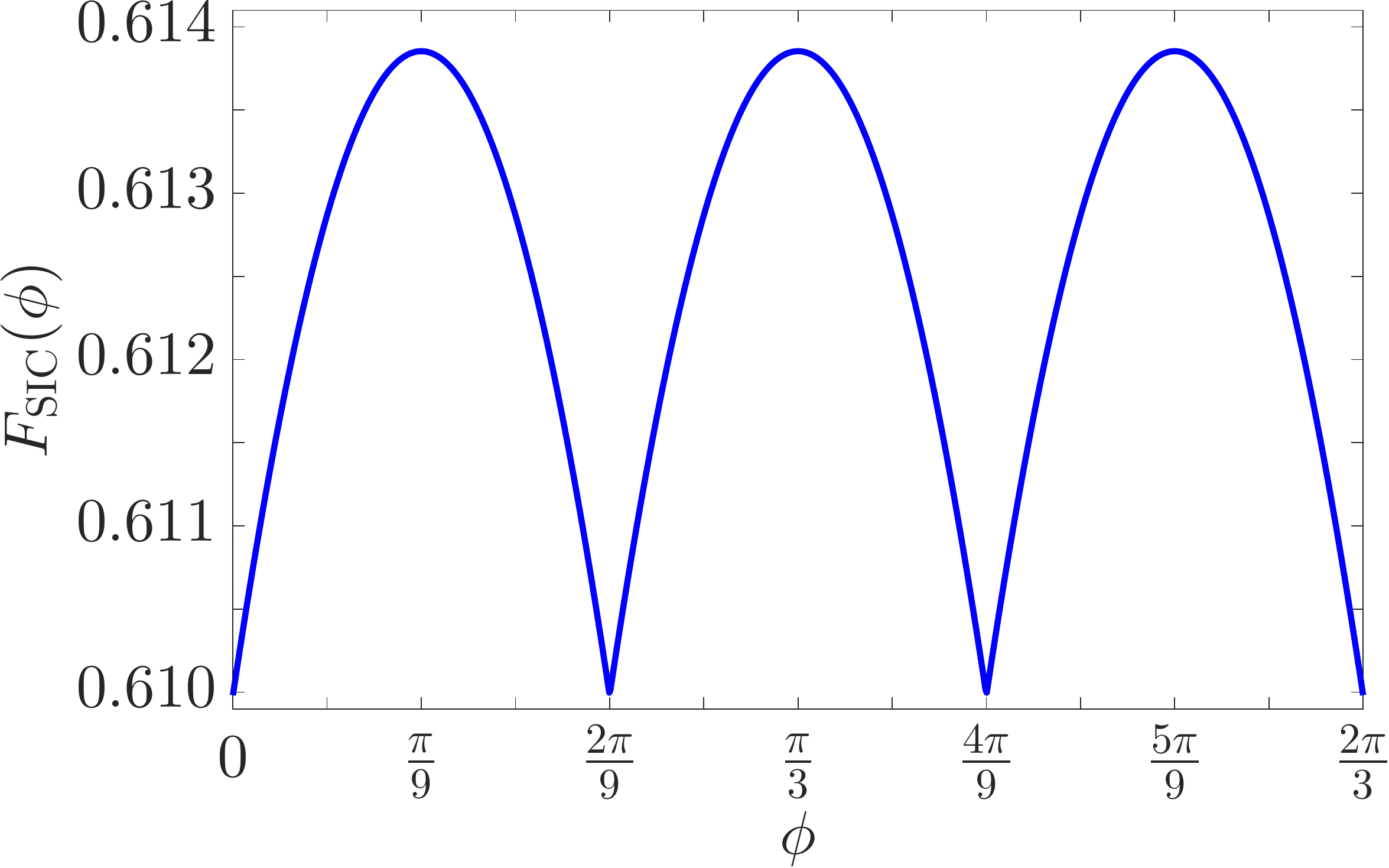}
	\caption{\label{fig:fidSIC}Three-copy estimation fidelity $F_\sic(\phi)$ achieved by the POVM $\scrA_{\sic}(\phi)^{\otimes 3}$, where $\scrA_{\sic}(\phi)$ is the SIC defined in \eref{eq:SIC3}.  
	}
\end{figure}

Here we are interested in the estimation fidelity of $\scrA_{\sic}(\phi)^{\otimes 3}$,  denoted by   $F_\sic(\phi):=F(\scrA_{\sic}(\phi)^{\otimes 3})$ henceforth. The analytical expression for $F_\sic(\phi)$ is too complicated to be informative, but it is easy to compute its value numerically by virtue of \esref{eq:FidPOVM2}{eq:QABCtrace1}. 
The dependence of  $F_\sic(\phi)$ on $\phi$ is illustrated in 
\fref{fig:fidSIC}, which indicates that $F_\sic(\phi)$ is periodic in $\phi$ with period $2\pi/9$. In addition,
$F_\sic(\phi)$ increases monotonically with $\phi$  for $\phi\in [0,\pi/9]$, but decreases monotonically for  $\phi\in [\pi/9,2\pi/9]$.  Notably, $F_\sic(\phi)$ attains its maximum when $\phi=\pi/9$, which corresponds to the exceptional SIC with intermediate symmetry; by contrast, $F_\sic(\phi)$ attains its minimum when $\phi=0$, which corresponds to the most symmetric  SIC. This conclusion seems quite unexpected, and  a simple explanation is yet to be found. The estimation fidelities achieved by generic SICs interpolate between the two extreme cases. In conjunction with known results on the equivalent classes of SICs under unitary transformations \cite{Appl05,Zhu10,Zhu15S}, \fref{fig:fidSIC} provides strong evidence for the following conjecture. 
\begin{conjecture}
Two SICs in dimension 3 can achieve the same three-copy estimation fidelity iff they are unitarily equivalent.
\end{conjecture}

\section{\label{sec:summary}Summary}     We proposed a simple but powerful approach for decoding the characteristics of quantum measurements by virtue of a simple problem in  quantum state estimation. Based on this approach we provided surprisingly simple characterizations of various typical and important quantum measurements, including  rank-1 projective measurements, MUMs, and SICs. Notably, we do not need any assumption on the rank, purity, or the number of POVM elements, and we do not need bases to start with, which seems impossible with all previous approaches. 
Our work demonstrates that all these elementary quantum measurements are uniquely determined by their information-extraction capabilities. In other words, all these elementary quantum measurements  can be defined in purely information theoretic terms, in sharp contrast with traditional algebraic definitions, which lack clear operational meanings.   
In this way, our work  offers a fresh perspective for understanding and exploring quantum measurements from their information-extraction capabilities.

The two-copy estimation fidelity we introduced also offers a new perspective for understanding 
quantum incompatibility as a resource. In addition, this estimation fidelity can be used to construct a universal criterion for detecting incompatibility of two arbitrary POVMs.
Moreover, it has an intimate connection with entropic uncertainty relations. Furthermore, we showed that the three-copy estimation fidelity can be used to distinguish inequivalent MUB and SICs, which cannot be distinguished by pairwise fidelities. Such operational figures of merit are quite rare in the literature and are expected to play an important role in understanding various discrete symmetric structures tied to the quantum state space. In the course of  study, we derived a number of results on quantum measurements and (weighted complex projective) $t$-designs, which are of independent interest.
Our work offers valuable insights not only on  quantum measurements and quantum estimation theory, but also on various related research areas, including geometry of quantum states, $t$-designs and random quantum states, quantum incompatibility, and foundational studies. The implications of these results deserve further explorations in the future.

\bigskip

\acknowledgments
This work is  supported by   the National Natural Science Foundation of China (Grants No.~11875110 and No.~92165109) and  Shanghai Municipal Science and Technology Major Project (Grant No.~2019SHZDZX01).

%%%%%%%%%%%%%%%%%%%%%%%%%%%%%%%%%%%%%%%%%%%%%%%%%%%%%%%%%%%%%%%%%%%%%%%%%%%%%%%%%%%%%%%%%%%%%%%%%%%%%%%%%%%%%%%%%%%%%%%%%%%%%%%%%%%%%%%%%%%%%%%%%%%%%%%%%%%%%%%%%%%%%%%%%%%%%%%%%%%%%%%%%%%%%%%%%%%%%%%%%%%%%%%%%%%%%%%%%%%%%%%%%%%%%%%%%%%%%%%%%%%%%%%%%%%%%%%%%%%%%%%%%%%%%%%%%%%%%%%%%%%%%%%%%%%%%%%%%%%

%\tableofcontents

\appendix
\section{\label{asec:POVMorder}Proofs of \lsref{lem:EquivalentPOVM} and \ref{lem:SimplePOVM}}
\begin{proof}[Proof of \lref{lem:EquivalentPOVM}]
	Without loss of generality we can assume that no POVM element in $\scrA$ or $\scrB$ is equal to the zero operator. The inequality $\wp(\scrA)\leq \wp(\scrB)$ can be proved as follows,
	\begin{align}
	& d\wp(\scrA)=\sum_j \frac{\tr(A_j^2)}{\tr(A_j)}=\sum_j \frac{\sum_{k,l}\Lambda_{jk}\Lambda_{jl}\tr(B_k B_l)}{\tr(A_j)}\nonumber\\
	&\leq  \sum_j \frac{\sum_{k,l}\Lambda_{jk}\Lambda_{jl}\sqrt{\tr(B_k^2)} \sqrt{\tr(B_l^2)}}{\tr(A_j)}\nonumber\\
	&=\sum_j \frac{\bigl[\sum_{k}\Lambda_{jk}\sqrt{\tr(B_k^2)}\,\bigr]^2}{\sum_{k}\Lambda_{jk}\tr(B_k)}\leq \sum_k \frac{\tr(B_k^2)}{\tr(B_k)}=d\wp(\scrB), \label{lem:PurityCGproof}
	\end{align}
	where the first inequality follows from the Cauchy-Schwarz inequality, and the second inequality follows from Lemma~S1 in \rcite{ZhuHC16}. If $\scrA$ is equivalent to $\scrB$, then the opposite inequality $\wp(\scrB)\leq \wp(\scrA)$ holds by the same token, so we have $\wp(\scrA)= \wp(\scrB)$, which confirms the implication $1\imply 2$. 
	
	If $\wp(\scrA)= \wp(\scrB)$, then the two inequalities in \eref{lem:PurityCGproof} are saturated. Note  that $\tr(B_k B_l )\leq \sqrt{\tr(B_k^2)} \sqrt{\tr(B_l^2)}$, and the inequality is saturated iff $B_k$ and $B_l$ are proportional to each other. The saturation of the  first inequality in \eref{lem:PurityCGproof} then
	implies  that $\Lambda_{jk}\Lambda_{jl}=0$ whenever $B_k$ and $B_l$ are  linearly independent, which confirms the implication $2\imply 3$. 
	
	If statement 3 holds, then the product $\Lambda_{jk}\Lambda_{jl}$ can take on a nonzero value only if $B_k$ and $B_l$ are proportional to each other. In this case, $\scrB$ can be realized by data  processing after performing $\scrA$; in other words, $\scrB$ is a coarse graining of $\scrA$. Since $\scrA$ is a coarse graining of $\scrB$ by assumption, it follows that  $\scrA$ is equivalent to $\scrB$, which confirms the implication $3\imply 1$ and completes the proof of \lref{lem:EquivalentPOVM}. 
\end{proof}

\begin{proof}[Proof of \lref{lem:SimplePOVM}]
	Suppose $\scrA=\{A_j\}_j$ and $\scrB=\{B_k\}_k$ are two simple POVMs. Obviously, $\scrA$ and $\scrB$ are equivalent if they are identical up to relabeling.

	Conversely, suppose $\scrA$ and $\scrB$ are equivalent; then $A_j$ can be expressed as $A_j=\sum_k \Lambda_{jk}B_k$, where $\Lambda$ is a stochastic matrix. According to \lref{lem:EquivalentPOVM}, each row of $\Lambda$ has only one nonzero entry, given that the POVM elements in $\scrB$ are pairwise linearly independent, and so are the POVM elements in $\scrA$.
	It follows that $A_j$ is proportional $B_k$ whenever $\Lambda_{jk}>0$. Now the simplicity of $\scrA$ further implies  that each column of $\Lambda$ has only one nonzero entry, which is necessarily equal to 1. Therefore, $\Lambda$ is a permutation matrix, which means $\scrA$ and $\scrB$  are identical up to relabeling, confirming the first statement in \lref{lem:SimplePOVM}.

	To prove the second statement in \lref{lem:SimplePOVM}, let $\scrC=\{C_j\}_j$ be an arbitrary POVM; then an equivalent simple POVM can be constructed by deleting POVM elements that are equal to the zero operator and combining POVM elements that are proportional to each other. According to the first statement in \lref{lem:SimplePOVM}, such a simple POVM is unique up to relabeling. 
\end{proof}

\section{\label{asec:ProjectiveMU}Proofs of \lsref{lem:rank1Projective}-\ref{lem:commutePOVMrank1} and \thref{thm:MUPOVM}}

\begin{proof}[Proof of \lref{lem:rank1Projective}]
	Suppose $\scrA=\{A_j\}_{j=1}^m$, where $A_j$ are rank 1 by assumption. Then we have
	\begin{gather}
\sum_j A_j=1,\quad 	\sum_j \tr(A_j)=d,\\
	d=\sum_{j,k}\tr(A_j A_k)\geq \sum_{j,k}\tr(A_j^2)=\sum_j (\tr A_j)^2\geq \frac{d^2}{m},\label{eq:rank1ProjectiveProof}
	\end{gather}
	which implies that $m\geq d$, so $\scrA$ has at least $d$ POVM elements. Obviously, the lower bound is saturated if $\scrA$ is  a rank-1 projective measurement. 
	
	Conversely, if $\scrA$ has $d$ POVM elements, that is, $m=d$, then the two inequalities in \eref{eq:rank1ProjectiveProof} are saturated, which implies that 
	\begin{align}
	\tr(A_j A_k)=\delta_{jk}.
	\end{align}
	Therefore, all the  POVM elements $A_j$ are mutually orthogonal rank-1 projectors, which means  $\scrA$ is a rank-1 projective measurement, confirming the
	first statement in \lref{lem:rank1Projective}. In the above reasoning it is not necessary to assume that   $\scrA$ is simple.

Next, we turn to the
second statement in \lref{lem:rank1Projective}. 
Let us consider the span of POVM elements in $\scrA$ and its dimension, assuming that $\scrA$ is simple. By assumption each POVM element of $\scrA$ has the form  $A_j=w_j|\psi_j\>\<\psi_j|$ with $0< w_j\leq 1$; in addition, the corresponding set of kets $\{|\psi_j\>\}_j$ spans $\caH$. So we can find $d$ kets, say $|\psi_1\>, |\psi_2\>, \ldots, |\psi_d\>$, that form a basis (not necessarily orthogonal) for $\caH$. Then  the corresponding set of projectors $\{|\psi_j\>\<\psi_j|\}_{j=1}^d$ is necessarily linearly independent, which implies that
\begin{align}\label{eq:SpanDimLB}
\dim(\spa(\scrA))\geq \dim\bigl(\spa \bigl(\{|\psi_j\>\<\psi_j|\}_{j=1}^d\bigr)\bigr)=d.
\end{align}
If $\scrA$ is a rank-1 projective measurement, then the lower bound is saturated. 

Conversely, if the lower bound in \eref{eq:SpanDimLB} is saturated, then each POVM element of $\scrA$ is a linear combination of $|\psi_j\>\<\psi_j|$ for $j=1,2,\ldots, d$. Note that the rank of such a linear combination is equal to the number of nonzero coefficients. Since $\scrA$ is rank 1 by assumption, it follows that each POVM element of $\scrA$ is proportional to $|\psi_j\>\<\psi_j|$ for some $j=1,2,\ldots, d$, which implies that $\scrA$ has $d$ POVM elements given that $\scrA$ is simple. Therefore, $\scrA$ is a rank-1 projective measurement according to the
first statement in \lref{lem:rank1Projective} as proved above.	
\end{proof}

\begin{proof}[Proof of \lref{lem:commutePOVMrank1Proj}]
	Let $\scrA\!=\{A_j\}_{j=1}^m$ and $\scrB=\{B_k\}_{k=1}^n$.  If $\scrB\subseteq\scrA$ and the projectors in $\scrB$ are mutually orthogonal, then $\scrA\setminus\scrB$ is orthogonal to $\scrB$ given that
	\begin{align}\label{eq:APOVM}
	\sum_{j=1}^m A_j =1, \quad A_j\geq 0 \quad \forall j=1,2,\ldots, m,
	\end{align} 
	so $\scrA$ and $\scrB$ commute.

	Conversely, suppose $\scrA$ and $\scrB$ commute. Then each $A_j\in \scrA$ commutes with each $B_k\in \scrB$, which means  $A_j$ is either orthogonal to $B_k$ or proportional to $B_k$. Since $\scrA$ is a simple POVM, it follows that for any given projector $B_k\in \scrB$  there exists a unique POVM element in $\scrA$ that is proportional to $B_k$, and all other POVM elements are orthogonal to $B_k$. By a suitable relabeling if necessary, we can assume that 
	\begin{align}
	A_j\propto B_j \quad \forall j=1,2,\ldots, n.
	\end{align}
	Then the above analysis means
	\begin{gather}
	B_j B_k=B_j \delta_{jk} \quad  \forall  j,k=1,2, \ldots,n;     \label{eq:BjBkorthogonal}\\
	 A_j A_k=A_j^2 \delta_{jk} \quad  \forall j=1,2,\ldots,m;\, k=1,2, \ldots,n. \label{eq:AjAkorthogonal}
	\end{gather} 
	In particular, the rank-1 projectors in $\scrB$
are mutually orthogonal, and so are the POVM elements $A_1, A_2,\ldots, A_n$. Now \ecref{eq:APOVM}{eq:AjAkorthogonal} together further imply that the POVM elements $A_1, A_2,\ldots, A_n$ are mutually orthogonal rank-1 projectors and $A_j=B_j$ for $j=1,2,\ldots, n$, which in turn imply that $\scrB\subseteq \scrA$.	
\end{proof}

\begin{proof}[Proof of \lref{lem:commutePOVMrank1}]
\Lref{lem:commutePOVMrank1}	is a simple corollary of 
	\lref{lem:commutePOVMrank1Proj} and can also be proved directly as follows. Suppose $\scrA=\{A_j\}_j$ and $\scrB=\{B_k\}_k$ are two simple rank-1 POVMs. Then  $\scrA$ and $\scrB$ commute with each other if they are identical projective measurements up to relabeling. 
	
	Conversely, if  $\scrA$ and $\scrB$ commute, then $A_j B_k=B_kA_j$ for any pair $j,k$. Note that two rank-1 positive operators commute with each other iff they are orthogonal  or proportional to each other. In conjunction with the assumption that $\scrA$ and $\scrB$ are simple rank-1 POVMs, 
	we conclude that  each POVM element in $\scrA$ ($\scrB$) is proportional to a unique POVM element in $\scrB$ ($\scrA$) and is orthogonal to all other POVM elements in $\scrB$ ($\scrA$). In this way, POVM elements in $\scrA$ have one-to-one correspondence with POVM elements in $\scrB$; in particular, $\scrA$ and $\scrB$ have the same number of POVM elements. By a suitable relabeling if necessary, we can assume that $A_j$ is proportional to $B_j$ and is orthogonal to $B_k$ with $k\neq j$. So  all POVM elements in $\scrA$ are mutually orthogonal, and so are POVM elements in $\scrB$, which means  both $\scrA$ and $\scrB$ are rank-1 projective measurements.  Moreover, $\scrA$ and $\scrB$ are identical  up to relabeling given the above correspondence. 
\end{proof}

\begin{proof}[Proof of \thref{thm:MUPOVM}]
	Let $\{\scrA_r\}_{r=1}^g$ be an arbitrary set of $g$ MU simple rank-1 POVMs on $\caH$, where $\scrA_r=\{A_{r j}\}_j$. Let $A_{r j}'=A_{r j}-(\tr A_{r j}/d)$ and  $\scrA_r'=\{A_{r j}'\}_j$.  Then 
 $A_{r j}'$ are  traceless and 
	\begin{align}
 \dim (\spa( \scrA_r))= \dim (\spa(\scrA_r'))+1.
	\end{align}
By assumption  $\scrA_r$ and  $\scrA_s$ with $r\neq s$ are MU, which implies that	
\begin{align}
\tr(A_{r j}'A_{s k}')=0 \quad \forall j,k,  
\end{align}	 
so  $A_{r j}'$ and $A_{s k}'$ are orthogonal with respect to the Hilbert-Schmidt inner product.	As a consequence,
 \begin{align}
 d^2&\geq \dim\bigl(\spa \bigl(\cup_{r=1}^g  \scrA_r\bigr)\bigr)=\dim\bigl(\spa \bigl(\cup_{r=1}^g  \scrA_r'\bigr)\bigr)+1\nonumber\\
 &=\sum_{r=1}^g [\dim (\spa( \scrA_r))-1]+1\nonumber\\
 & =\sum_{r=1}^g \dim (\spa( \scrA_r))-g+1\nonumber\\
 &\geq dg-g+1, \label{eq:POVMspanLB}
 \end{align}
where the second inequality follows from \lref{lem:rank1Projective} and is saturated iff each $\scrA_r$ is a rank-1 projective measurement. 
	
	\Eref{eq:POVMspanLB} implies that $g\leq d+1$. If the upper bound is saturated, then the two inequalities in \eref{eq:POVMspanLB} are saturated, which means $\dim (\spa( \scrA_r))=d$ for each $\scrA_r$. 
So  all the POVMs $\scrA_r$ are rank-1 projective measurements by \lref{lem:rank1Projective}, which implies that $\{\scrA_r\}_{r=1}^g$ is a CMUMs. 
\end{proof}

%%%%%%%%%%%%%%%%%%%%%%%%%%%%%%%%%%%%%%%%%%%%%%%%%%%%%%%%%%%%%%%%%%%%%%%%%%%%%%%%%%%%%%%%%%%%%%%%%%%%%%%%%%%%%%%%%%%%%%%%%%%%%%%%%%%%%%%%%%%%%%%%%%%%%%%%%	

\section{\label{asec:FP}Proofs of \lsref{lem:FPhalfLBUB}-\ref{lem:crossFP2design}}

\subsection{Main proofs}

\begin{proof}[Proof of \lref{lem:FPhalfLBUB}]
By assumption $\caS$ can be expressed as $\caS=\{|\psi_j\>, w_j\}_{j=1}^m$, where
\begin{equation}\label{eq:FPhalfnormalizationProof}
w_j>0,\quad \sum_j w_j=d,\quad
\quad \sum_j w_j |\psi_j\>\<\psi_j|=1.
\end{equation}	
The lower bound in \eref{eq:FPhalfLBUB} can be proved as follows,
	\begin{align}
	\Phi_{1/2}(\caS)&=\sum_{j,k}w_jw_k |\<\psi_j|\psi_k\>|\geq \sum_{j,k}w_jw_k |\<\psi_j|\psi_k\>|^2\nonumber\\
	&=\tr\Biggl(\sum_j w_j |\psi_j\>\<\psi_j|\Biggr)^2=d. \label{eq:FPhalfLB}
	\end{align}	
	If $\caS$ is an orthonormal basis (with uniform weights), which means $m=d$, $w_j=1$, and $\<\psi_j|\psi_k\>=\delta_{jk}$, then  it is straightforward to verify that the lower bound is saturated.

Conversely, if the lower bound in \eref{eq:FPhalfLB}  is saturated and $\caS$ is simple, then $|\<\psi_j|\psi_k\>|$ can take on only two distinct values, namely, 0 and 1, so  we have 
	\begin{align}
	\<\psi_j|\psi_k\>=\delta_{jk}\quad \forall j, k=1,2,\ldots m. 
	\end{align}
	This equation can hold only if  $m\leq d$. On the other hand, \eref{eq:FPhalfnormalizationProof} implies that 
	$m\geq d$ [cf. \lref{lem:rank1Projective} and \eref{eq:designEleNumLB} in the main text]. So   $m=d$ and $\{|\psi_j\>\}_{j=1}^d$ forms  an orthonormal basis. Now \eref{eq:FPhalfnormalizationProof} further implies  that $w_j=1$ for all $j$. Therefore, $\caS$ is an orthonormal basis (with uniform weights).

	Next, to prove the upper bound in \eref{eq:FPhalfLBUB}, 	
	define $p_{jk}:=w_jw_k/d^2$ and $x_{jk}:=|\<\psi_j|\psi_k\>|^2$. Then 	
 \eref{eq:FPhalfnormalizationProof}  implies that
	\begin{align}
  \sum_{j,k} p_{jk}= 1, \quad \sum_{j,k}p_{jk} x_{jk}=\frac{1}{d}. 
	\end{align}
	In addition,  from \eref{eq:FPLB} we can deduce that
	\begin{align}
	&\sum_{j,k}p_{jk} x_{jk}^2=\frac{1}{d^2}\sum_{j,k}w_jw_k |\<\psi_j|\psi_k\>|^4\nonumber\\	
	&=\frac{1}{d^2}\tr\Biggl[\sum_j w_j (|\psi_j\>\<\psi_j|)^{\otimes 2}\Biggr]^2\geq\frac{2}{d(d+1)}, \label{eq:2designFPLB}
	\end{align}
	and  the lower bound is saturated iff $\{|\psi_j\>,w_j\}_{j=1}^m$ forms a  2-design.
	By virtue of \lsref{lem:zetaProperty} and \ref{lem:HalfMomBound} in \aref{asec:moment}, we can now deduce that
	\begin{align}
	&\Phi_{1/2}(\caS)=\sum_{j,k}w_jw_k |\<\psi_j|\psi_k\>|=d^2\sum_{j,k}p_{jk}\sqrt{x_{jk}}\nonumber\\
	&\leq d^2\zeta\left(\frac{1}{d},\frac{2}{d(d+1)}\right)=
	1+(d-1)\sqrt{d+1}, \label{eq:halfdesignProof1}
	\end{align}
	which confirms the upper bound in \eref{eq:FPhalfLBUB}. Here the function $\zeta$ is defined in \eref{eq:zetaab} in \aref{asec:moment}.	
	If $\caS$ is a SIC (with uniform weights), which means  $m=d^2$, $w_j=1/d$, and $|\<\psi_j|\psi_k\>|^2=(d\delta_{jk}+1)/(d+1)$, then it is straightforward to verify that the upper bound is saturated [cf. \eref{eq:FPhalfSIC}]. 
	
	Conversely, if  the upper bound in \eref{eq:FPhalfLBUB} is saturated, then the inequality in \eref{eq:halfdesignProof1} is saturated. According to \lsref{lem:zetaProperty} and \ref{lem:HalfMomBound}, the lower bound in \eref{eq:2designFPLB} must saturate, so  $\{|\psi_j\>,w_j\}_{j=1}^m$ forms a  2-design; 
meanwhile, we have
\begin{align}
\sum_{j,k:|\<\psi_j|\psi_k\>|^2=1/(d+1)} \label{eq:wjwk1} w_jw_k&=d^2-1,\\ \sum_{j,k:|\<\psi_j|\psi_k\>|^2=1} w_jw_k&=1,\label{eq:wjwk2}
\end{align}
given that $w_jw_k=d^2p_{jk}$, 
so $|\<\psi_j|\psi_k\>|^2$ can take on only two distinct values, namely, 1 and $1/(d+1)$. 

If in addition $\caS$ is simple, then $|\<\psi_j|\psi_k\>|^2<1$ whenever $j\neq k$, so \esref{eq:wjwk1}{eq:wjwk2} imply that
\begin{align}
|\<\psi_j|\psi_k\>|^2=\frac{d\delta_{jk}+1}{d+1}\quad \forall j, k=1,2,\ldots m. 
\end{align}
This equation can hold only if  $m\leq d^2$. On the other hand, the opposite inequality  $m\geq d^2$ has to hold given that $\{|\psi_j\>,w_j\}_{j=1}^m$ forms a   2-design [cf. \eref{eq:designEleNumLB} in the main text]. So   $m=d^2$ and the set $\{|\psi_j\>\}_{j=1}^{d^2}$ forms a SIC. Furthermore, from \eref{eq:wjwk2} we can deduce that
\begin{align}
\sum_j w_j^2=\sum_{j,k:|\<\psi_j|\psi_k\>|^2=1} w_jw_k=1.\label{eq:wjsquare2}
\end{align}
\Esref{eq:FPhalfnormalizationProof}{eq:wjsquare2} together   imply that $w_j=1/d$ for all $j$. Therefore, $\caS$ is a SIC (with uniform weights), which completes the proof of \lref{lem:FPhalfLBUB}.	
\end{proof}

\begin{proof}[Proof of \lref{lem:EAL}]
By assumption the weighted set $\caS$ has the form $\caS=\{|\psi_j\>, w_j\}_{j=1}^m$ and satisfies the condition
	\begin{align}\label{eq:normalizationEALproof}
0\leq  w_j\leq 1,\quad 	\sum_j w_j= d,\quad \sum_{j,k}w_jw_k |\<\psi_j|\psi_k\>|^2=d. 
	\end{align}
	Let $h=\sum_j w_j^2$; then the above equation implies that
	\begin{align}\label{eq:hLB}
 \frac{d^2}{m}\leq h\leq d,
	\end{align}
	and the lower bound is saturated iff $w_j=d/m$ for all $j$. Therefore,
	\begin{align}
	\Phi_{1/2}(\caS)&=\sum_{j,k}w_jw_k |\<\psi_j|\psi_k\>|\nonumber\\
	&=\sum_{j}w_j^2 +\sum_{j\neq k}w_jw_k |\<\psi_j|\psi_k\>|\nonumber\\
	&
	\leq \sum_{j}w_j^2 +\sqrt{\sum_{j\neq k}w_jw_k}\sqrt{\sum_{j\neq k}w_jw_k|\<\psi_j|\psi_k\>|^2}\nonumber\\
	&=h+\sqrt{d^2-h}\sqrt{d-h}\nonumber\\
	&\leq  \frac{d^2}{m}+\frac{d}{m}\sqrt{d(m-1)(m-d)}, \label{eq:EALfpProof}
	\end{align}
	which confirms the upper bound in \eref{eq:EALfp}. Here the first inequality follows from the Cauchy inequality. The second inequality  follows from  \eref{eq:hLB} and the fact that the function $h+\sqrt{d^2-h}\sqrt{d-h}$ is strictly decreasing in $h$ for $0\leq h\leq d$ and $d\geq 2$; it is saturated iff $h=d^2/m$.  
	
If $\caS$ is composed of   $m$ equiangular states (with uniform weights), then \eref{eq:EAL} in the main text holds and we have $w_j=d/m$ for $j=1,2,\ldots, m$. 
So both inequalities in \eref{eq:EALfpProof} are saturated, which means the upper bound in
\eref{eq:EALfp} is saturated. 

Conversely, if the upper bound in
\eref{eq:EALfp} is saturated, then both inequalities in \eref{eq:EALfpProof} are saturated. The saturation of the second inequality implies that $h=d^2/m$, which in turn implies that  $w_j=d/m$ for $j=1,2,\ldots, m$. Then the saturation of the first inequality implies \eref{eq:EAL} given \eref{eq:normalizationEALproof}.
Therefore,  $\caS$ is composed of   $m$ equiangular states (with uniform weights).
 Note that any equiangular set  in $\caH$ can contain at most $d^2$ states \cite{LemmS73,Zaun11,ApplFZ15G}, so \eref{eq:EALfp} cannot be saturated when $m>d^2$. This observation completes the proof of \lref{lem:EAL}. 
\end{proof}

\begin{proof}[Proof of \lref{lem:crossFP}]
By assumption $\caS$ and $\caT$ can be expressed as  $\caS=\{|\psi_j\>, w_j\}_{j=1}^m$ and $\caT=\{|\varphi_k\>, w_k'\}_{k=1}^n$, which satisfy 
\begin{equation}
\begin{gathered}\label{eq:crossFPSTnormalization}
w_j, w_k'>0,\quad  \sum_j w_j=\sum_k w_k'=d,\\
\quad \sum_j w_j |\psi_j\>\<\psi_j|=\sum_k w_k' |\varphi_k\>\<\varphi_k|=1.
\end{gathered}
\end{equation}
The upper  bound in \eref{eq:crossFPLBUB} can be proved as follows,
\begin{align}
\Phi_{1/2}(\caS,\caT)&=\sum_{j,k}w_jw_k' |\<\psi_j|\varphi_k\>|\nonumber\\
&\leq \sqrt{\Biggl(\sum_{j,k}w_jw_k'\Biggr) \Biggl(\sum_{j,k}w_jw_k'|\<\psi_j|\varphi_k\>|^2\Biggr)}\nonumber\\
&=\sqrt{d^2\times d}=d^{3/2}.
\end{align}
Here the inequality follows from the Cauchy inequality and is saturated iff $|\<\psi_j|\varphi_k\>|^2=1/d$ for each pair $j,k$. Therefore, the upper  bound in \eref{eq:crossFPLBUB} is saturated iff $\caS$ and $\caT$ are MU.

The lower bound in \eref{eq:crossFPLBUB} can be proved following a similar approach used to prove the lower bound in \eref{eq:FPhalfLBUB},
\begin{align}
&\Phi_{1/2}(\caS,\caT)\!=\!\sum_{j,k}w_jw_k' |\<\psi_j|\varphi_k\>|\geq \sum_{j,k}w_jw_k' |\<\psi_j|\varphi_k\>|^2\nonumber\\
&=\!\tr\Biggl[\Biggl(\sum_j w_j |\psi_j\>\<\psi_j|\Biggr)\Biggl(\sum_k w_k' |\varphi_k\>\<\varphi_k|\Biggr)\Biggr]=d.\label{eq:crossFPhalfLBUBproof} 
\end{align}	
If $\caS$ and $\caT$ are identical orthonormal bases (with uniform weights) up to relabeling, then $n=m=d$, $w_j=w_k'=1$, 
 and $|\<\psi_j|\varphi_k\>|=\delta_{jk}$ after a suitable relabeling if necessary, so the lower bound is saturated.

Conversely, if the lower bound in \eref{eq:crossFPhalfLBUBproof}  is saturated, then $|\<\psi_j|\varphi_k\>|$ can take on only two distinct values, namely, 0 and 1. If, in addition,  $\caS$ and $\caT$ are simple 1-designs, which satisfy the normalization condition in \eref{eq:crossFPSTnormalization}, then for each $j$ there exists a unique $k$ such that $|\<\psi_j|\varphi_k\>|=1$; similarly, for each $k$ there exists a unique $j$ such that $|\<\psi_j|\varphi_k\>|=1$. 
Therefore, $n=m=d$ and 
\begin{align}
\<\psi_j|\psi_k\>=\<\varphi_j|\varphi_k\>=\delta_{jk};
\end{align}
 in addition, we have $|\<\psi_j|\varphi_k\>|=\delta_{jk}$ after a suitable relabeling if necessary. Now \eref{eq:crossFPSTnormalization} further implies that $w_j=w_j'=1$ for $j=1,2,\ldots, d$. Therefore, $\caS$ and $\caT$ are identical orthonormal bases (with uniform weights) up to relabeling. Note that 
 we identify weighted sets that differ only by overall phase factors as mentioned in the main text.
\end{proof}

\begin{proof}[Proof of \lref{lem:crossFP2design}]
	To prove the upper bound in \eref{eq:crossFP2design}  we can apply a similar reasoning 	used to prove the upper bound in \eref{eq:FPhalfLBUB}. Without loss of generality, we can assume that $\caS$ is a 1-design, while $\caT$ is a 2-design. Then  $\caS$ and $\caT$ can be expressed as $\caS=\{|\psi_j\>, w_j\}_{j=1}^m$ and $\caT=\{|\varphi_k\>, w_k'\}_{k=1}^n$, which satisfy
	\begin{equation}\label{eq:crossFP2designNormal}
	\begin{gathered}
w_j, w_k'>0,\quad 	\sum_j w_j=\sum_k w_k'=d, \\ \sum_j w_j |\psi_j\>\<\psi_j|=1, \quad
	\sum_k w_k' (|\varphi_k\>\<\varphi_k|)^{\otimes 2}=\frac{2P_2}{d+1}, 
	\end{gathered}
	\end{equation} 	
	where $P_2$ is the projector onto the symmetric subspace in $\caH^{\otimes 2}$. 
	
	Let $p_{jk}=w_jw_k'/d^2$ and $x_{jk}=|\<\psi_j|\varphi_k\>|^2$. Then the conditions in 	
	\eref{eq:crossFP2designNormal} imply that 
	\begin{align}
	&\sum_{j,k} p_{jk}= 1,\quad \sum_{j,k}p_{jk} x_{jk}=\frac{1}{d}; \label{eq:crossFPpjkxjkProof}
	\end{align}
	in addition,
\begin{align}
&\sum_{j,k}p_{jk}x_{jk}^2=\frac{1}{d^2}\sum_{j,k}w_jw_k' |\<\psi_j|\varphi_k\>|^4\nonumber\\	
&=\frac{1}{d^2}\tr\Biggl\{\Biggl[\sum_j w_j (|\psi_j\>\<\psi_j|)^{\otimes 2}\Biggr]\Biggl[\sum_k w_k' (|\varphi_k\>\<\varphi_k|)^{\otimes 2}\Biggr]\Biggr\}\nonumber\\	
&=\frac{2}{d^2(d+1)}\tr\Biggl[\sum_j w_jP_2 (|\psi_j\>\<\psi_j|)^{\otimes 2}\Biggr]\nonumber\\
&=\frac{2}{d^2(d+1)}\sum_j w_j
=\frac{2}{d(d+1)}. \label{eq:crossFP2designProof}
\end{align}	
	
	According to  \lsref{lem:zetaProperty} and \ref{lem:HalfMomBound} in \aref{asec:moment}, \esref{eq:crossFPpjkxjkProof}{eq:crossFP2designProof} imply that 
	\begin{align}
	&\Phi_{1/2}(\caS,\caT)=\sum_{j,k}w_jw_k' |\<\psi_j|\varphi_k\>|=d^2\sum_{j,k}p_{jk}\sqrt{x_{jk}}\nonumber\\
	&\leq d^2\zeta\left(\frac{1}{d},\frac{2}{d(d+1)}\right)=
	1+(d-1)\sqrt{d+1}, \label{eq:crossFP2designProof1}
	\end{align}
	which confirms the upper bound in \eref{eq:crossFP2design}. If $\caS$ and $\caT$ are identical SICs (with uniform weights) up to relabeling, which means  $n=m=d^2$, $w_k'=w_j=1/d$, and $|\<\psi_j|\varphi_k\>|^2=(d\delta_{jk}+1)/(d+1)$ after a suitable relabeling if necessary, then it is straightforward to verify that the upper bound is saturated [cf. \eref{eq:FPhalfSIC}]. 
	
	Conversely, if  the upper bound in \eref{eq:crossFP2design} is saturated, then the inequality in \eref{eq:crossFP2designProof1} is saturated. According to \lsref{lem:zetaProperty} and \ref{lem:HalfMomBound}, we have
	\begin{align}
	\sum_{j,k:|\<\psi_j|\varphi_k\>|^2=1/(d+1)} w_jw_k'&=d^2-1,\\ \sum_{j,k:|\<\psi_j|\varphi_k\>|^2=1} w_jw_k'&=1,
	\end{align}
	so $|\<\psi_j|\varphi_k\>|^2$ can take on only two distinct values, namely, 1 and $1/(d+1)$. If in addition $\caS$ and $\caT$ are simple 1-designs, which satisfy \eref{eq:crossFP2designNormal}, then for each $|\psi_j\>$ in $\caS$ there exists a unique $|\varphi_k\>$ in $\caT$ such that $|\<\psi_j|\varphi_k\>|^2=1$, and vice versa. It follows that $n=m$ and 
	\begin{align}\label{eq:crossFPSICproof}
	|\<\varphi_j|\varphi_k\>|^2=|\<\psi_j|\psi_k\>|^2=|\<\psi_j|\varphi_k\>|^2=\frac{d\delta_{jk}+1}{d+1}
	\end{align}
	for $ j, k=1,2,\ldots, m$ after a suitable relabeling if necessary. This equation can hold only if  $m\leq d^2$. 
	
\Esref{eq:crossFP2designNormal}{eq:crossFPSICproof} together imply that 
\begin{align}\label{eq:FP1proof}
d=\Phi_1(\caS)&=\frac{d^2}{d+1}+\frac{d}{d+1}\sum_{j=1}^{m} w_j^2\geq \frac{d^2}{d+1}+\frac{d^3}{m(d+1)},
\end{align}
which in turn implies that  $\sum_{j=1}^m w_j^2=1$ and $m\geq d^2$. So we have $m=d^2$  given  the opposite inequality $m\leq d^2$ derived above. In conjunction with the normalization conditions in \eref{eq:crossFP2designNormal}, we can deduce that
$w_j=1/d$ for $j=1,2,\ldots, d^2$. A similar reasoning yields $w_k'=1/d$ for $k=1,2,\ldots, d^2$.
  Therefore, $\caS$ and $\caT$ are identical SICs (with uniform weights) up to relabeling, which completes the proof of \lref{lem:crossFP2design}. 	
\end{proof}

\subsection{\label{asec:moment}Auxiliary results on the $1/2$-moment}
Here we derive several results on the $1/2$-moment of a bounded random variable given the first and second moments. For $0<b\leq a<1$,  define
\begin{align}
\zeta(a,b):=&\frac{b-a^2+(1-a)\sqrt{(1-a)(a-b)}}{1-2a+b}\nonumber\\
=&\frac{2a-a^2-b+(1+a)\sqrt{(1-a)(a-b)}}{(\sqrt{1-a}+\sqrt{a-b}\lsp)^2}. \label{eq:zetaab}
\end{align}
In two special cases, \eref{eq:zetaab} reduces to 
\begin{align}
\zeta(a,a)=a,\quad \zeta(a,a^2)=\sqrt{a}.
\end{align}
When $a=1-(1-b)r$ with $0<r\leq 1$, \eref{eq:zetaab} yields
\begin{align}\label{eq:zetarb}
\!\zeta(1-(1-b)r,b)=1-\frac{r(\sqrt{r(1-r)}-r)}{1-2r}(1-b).
\end{align}

\begin{lem}\label{lem:zetaProperty}
	Suppose  $0< b\leq a< 1$. Then $\zeta(a,b)$ is strictly increasing in $a$ and strictly decreasing in $b$. In addition, $\zeta(a,b)$ is jointly concave in $a$ and $b$. 
\end{lem}
\begin{proof}[Proof of \lref{lem:zetaProperty}]
	First we assume $0< b< a< 1$. According to the following equations,
	\begin{widetext}
		\begin{align}
		\frac{\partial \zeta(a,b)}{\partial a}
		&=\frac{(1-a)\bigl[(1+4a-5b)\sqrt{1-a}+3(1-b)\sqrt{a-b}\,\bigr]}{2\sqrt{(1-a)(a-b)}(\sqrt{1-a}+\sqrt{a-b}\,)^3}> 0,\\
		\frac{\partial \zeta(a,b)}{\partial b}&=-\frac{(1-a)^{3/2}\bigl[a-b+\sqrt{(1-a)(a-b)}\,\bigr]}{2(a-b)(\sqrt{1-a}+\sqrt{a-b}\,)^3}< 0,
		\end{align}
		$\zeta(a,b)$ is strictly increasing in $a$ and strictly decreasing in $b$. According to the following equations,
		
		\begin{align}
		\frac{\partial ^2\zeta(a,b)}{\partial a^2}&=-\frac{(1-b)^2\bigl[(1+2a-3b)+4\sqrt{(1-a)(a-b)}\,\bigr]}{4(1-a)^{1/2}(a-b)^{3/2}(\sqrt{1-a}+\sqrt{a-b}\,)^4}< 0,\\
		\frac{\partial ^2\zeta(a,b)}{\partial b^2}&=-\frac{(1-a)^{3/2}\bigl[(1+2a-3b)+4\sqrt{(1-a)(a-b)}\,\bigr]}{4(a-b)^{3/2}(\sqrt{1-a}+\sqrt{a-b}\,)^4}< 0,\\
		&\frac{\partial ^2\zeta(a,b)}{\partial a^2}\frac{\partial ^2\zeta(a,b)}{\partial b^2}-\biggl(\frac{\partial ^2\zeta(a,b)}{\partial a\partial b}\biggr)^2=0,
		\end{align}		
	\end{widetext}
	$\zeta(a,b)$ is jointly concave in $a$ and $b$. Incidentally, the equality in the last equation is tied to \eref{eq:zetarb}. 
	
	In the general situation $0< b\leq a< 1$, the conclusions in \lref{lem:zetaProperty} follow from the above analysis and the fact that $\zeta(a,b)$ is continuous in the limit $b\rightarrow a$, 
	\begin{align}
	\lim_{b\rightarrow a}\zeta(a,b)=\zeta(a,a)=a. 
	\end{align}
\end{proof}

The significance of the function $\zeta(a,b)$ is manifested in the following lemma.
\begin{lem}\label{lem:HalfMomBoundX}
	Suppose $X$ is a random variable that satisfies the conditions $0\leq X\leq 1$, $\bbE[X]=a$, and $\bbE[X^2]=b$ with $0<b\leq a<1$; 
	then 
	\begin{equation}\label{eq:halfmomentLBUBX}
	a\sqrt{\frac{a}{b}}\leq \bbE[\sqrt{X}] \leq \zeta(a,b).
	\end{equation}
	The lower bound  is saturated iff 
	\begin{equation}\label{eq:LBconditionX}
	\rmP(X=0)=1-\frac{a^2}{b},\quad  \rmP\Bigl(X=\frac{b}{a}\Bigl)=\frac{a^2}{b},
	\end{equation}
	while the upper bound  is saturated iff
	\begin{equation}\label{eq:UBconditionX}
	\rmP\Bigl(X=\frac{a-b}{1-a}\Bigr)=\frac{(1-a)^2}{1-2a+b}, \quad \rmP(X=1)=\frac{b-a^2}{1-2a+b}.
	\end{equation}
\end{lem}
The assumptions in \lref{lem:HalfMomBoundX} imply the inequalities   $0< a^2\leq b\leq a< 1$.
\Eref{eq:LBconditionX} means $X$ can only take on the values of 0 and $b/a$, while \eref{eq:UBconditionX}  means $X$ can only take on the values of 1 and $(a-b)/(1-a)$. \Lref{lem:HalfMomBoundX} is a corollary of \lref{lem:HalfMomBound} below.

\begin{lem}\label{lem:HalfMomBound}
	Suppose $a$ and $b$ are positive constants that satisfy $0< a^2\leq b\leq a< 1$.	Suppose $x_1, x_2, \ldots, x_m$ and  $p_1,p_2,\ldots, p_m$ are nonnegative numbers that satisfy
	\begin{equation}\label{eq:MomentConstraints}
	\begin{gathered}
	0\leq x_j\leq 1, \quad 0\leq p_j\leq 1\quad \forall j,\\		
	\sum_j p_j=1,\quad 	\sum_j p_j x_j=a,\quad \sum_j p_j x_j^2=b;
	\end{gathered}
	\end{equation}
	then 
	\begin{equation}\label{eq:halfmomentLBUB}
	a\sqrt{\frac{a}{b}}\leq \sum_j p_j \sqrt{x_j} \leq \zeta(a,b).
	\end{equation}	
	The lower bound in \eref{eq:halfmomentLBUB} is saturated iff 
	\begin{equation}\label{eq:LBcondition}
	\sum_{j|x_j=0} p_j=1-\frac{a^2}{b}, \quad \sum_{j|x_j=b/a} p_j=\frac{a^2}{b},
	\end{equation}
	while the upper bound  is saturated iff
	\begin{equation}\label{eq:UBcondition}
	\sum_{j|x_j=\frac{a-b}{1-a}} p_j=\frac{(1-a)^2}{1-2a+b}, \quad \sum_{j|x_j=1} p_j=\frac{b-a^2}{1-2a+b}.
	\end{equation}
\end{lem}

\begin{proof}[Proof of \lref{lem:HalfMomBound}]
	When $m=2$, \lref{lem:HalfMomBound} follows from \lref{lem:HalfMomBound2v} below. 	When $b=a^2$, we have
	\begin{align}
	\zeta(a,b)=\sqrt{a},\quad \sum_{j|x_j=a}p_j=1, \quad \sum_{j|x_j\neq a}p_j=0,
	\end{align}
	so \lref{lem:HalfMomBound} also holds.

	It remains to consider the case with $m>2$ and $b>a^2$, which means not all $x_j$ with $p_j>0$ are equal to each other given the constraints in \eref{eq:MomentConstraints}.  Suppose the minimum of $\sum_j p_j \sqrt{x_j}$  is attained when $p_j=p_j^*$ and $x_j=x_j^*$. Without loss of generality, we can assume that 
	\begin{equation}\label{eq:OptSolAssump}
	\begin{gathered}
	p_j^*>0 \quad \forall j\leq l, \quad p_j^*=0\quad \forall j\geq l+1,  \\
	x_j^*< x_k^* \quad \forall 1\leq j<k\leq l,
	\end{gathered}	
	\end{equation}
	where $2\leq l\leq m$. Then we have 
	\begin{align}
	p_j \sqrt{x_j}+p_k \sqrt{x_k} \geq 	p_j^* \sqrt{x_j^*}+p_k^* \sqrt{x_k^*} \quad \forall 1\leq j<k\leq l
	\end{align}
	as long as 
	\begin{equation}
	\begin{gathered}
	p_j+p_k=p_j^*+p_k^*,\;\;  p_j x_j+p_k x_k=p_j^* x_j^*+p_k^* x_k^*,\\
	p_j x_j^2+p_k x_k^2=p_j^* {x_j^*}^2+p_k^* {x_k^*}^2.
\end{gathered}
	\end{equation}
By virtue of \lref{lem:HalfMomBound2v} below applied to the set of parameters $x_j,x_k, p_j/(p_j+p_k), p_k/(p_j+p_k)$, we can now deduce that  $x_j^*=0$, which in turn implies that $l=2$ given the assumptions in \eref{eq:OptSolAssump}. According to \lref{lem:HalfMomBound2v}  again, $\sum_j p_j \sqrt{x_j}$ is bounded from below by $a\sqrt{a/b}$,  which confirms the lower bound in \eref{eq:halfmomentLBUB}, and the bound is saturated iff \eref{eq:LBcondition} holds. The upper bound in \eref{eq:halfmomentLBUB} and the saturation condition can be established by a similar reasoning.
\end{proof}

%%%%%%%%%%%%%%%%%%%%%%%%%%%%%%%%%%%%%%%%%%%%%%%%%%%%%%%%%%%%%%%%%%%%%%%%%%%%%%%%%%%%%%%%%%%%%%%%%%%%%%%%%%%%%%%%%%%%%%%%%%%%%%%%%%%%%%%%%%%%%%%%%%%%%%%%%%%%%%%%%%%%%%%%%%%%%%%%%%%%%%%%%%%%%%%%%%%

\begin{lem}\label{lem:HalfMomBound2v}
	Suppose $a$ and $b$ are positive constants that satisfy $0< a^2\leq b\leq a< 1$. 	
	Suppose $p_1,p_2,x_1,x_2$ are nonnegative numbers that satisfy
	\begin{equation}\label{eq:MomentConstraints2v}
	\begin{gathered}
	0\leq x_1\leq x_2\leq 1,\;\;   0\leq p_1,p_2\leq 1,\\
	\!\!	p_1+p_2=1,\;\;      
	p_1x_1+p_2x_2=a,\;\; p_1x_1^2+p_2x_2^2=b;
	\end{gathered}
	\end{equation}
	then 
	\begin{equation}
	a\sqrt{\frac{a}{b}}\leq p_1 \sqrt{x_1}+ p_2 \sqrt{x_2} \leq \zeta(a,b).\label{eq:halfmomentLBUB2v}
	\end{equation}
	If  $b=a^2$, then $p_1 \sqrt{x_1}+ p_2 \sqrt{x_2}=\sqrt{a}=\zeta(a,b)$ and $x_j=a$ whenever $p_j\neq 0$ for $j=1,2$. If  $b>a^2$, then the lower bound in \eref{eq:halfmomentLBUB2v} is saturated iff
	\begin{equation}\label{eq:LBcondition2v}
	x_1=0,\quad p_1=1-\frac{a^2}{b},\quad x_2=\frac{b}{a},\quad  p_2=\frac{a^2}{b},
	\end{equation}	
	while the upper bound  is saturated iff
	\begin{equation}\label{eq:UBcondition2v}
	\begin{aligned}
	\quad x_1&=\frac{a-b}{1-a},&\quad 
	p_1&=\frac{(1-a)^2}{1-2a+b}, \\
	 x_2&=1,&\quad 
	p_2&=\frac{b-a^2}{1-2a+b}.
	\end{aligned}
	\end{equation}	
\end{lem}

\begin{proof}[Proof of \lref{lem:HalfMomBound2v}]
	If $b=a^2$, then $\zeta(a,b)=\sqrt{a}$. In addition,  \eref{eq:MomentConstraints2v} implies that
	\begin{align}
	p_1p_2(x_1-x_2)^2=0,
	\end{align}
	which means $p_1p_2=0$ or $x_1=x_2$. So  $x_j=a$ whenever $p_j\neq 0$, and we have $p_1 \sqrt{x_1}+ p_2 \sqrt{x_2}=\sqrt{a}=\zeta(a,b)$, in which case \eref{eq:halfmomentLBUB2v} holds automatically.

	If $b>a^2$, then the assumptions in \lref{lem:HalfMomBound2v} imply that $0<p_1,p_2<1$ and $0\leq x_1<a<x_2\leq 1$. In addition, $x_2,p_1,p_2$ are determined by $x_1$ as follows,
	\begin{align}\label{eq:x2p1p2}
	\begin{gathered}
	x_2=\frac{b-ax_1}{a-x_1},\quad p_1=\frac{b-a^2}{x_1^2-2ax_1+b},\\ p_2=\frac{(a-x_1)^2}{x_1^2-2ax_1+b}. 
	\end{gathered}
	\end{align}
	Note that $x_2$  increases monotonically with $x_1$. In addition, the assumption $b>a^2$ implies that
\begin{align}\label{eq:hdenominator}
x_1^2-2ax_1+b>0,
\end{align}	 
so the requirement $0< p_1,p_2< 1$ is automatically guaranteed given that 
	$0\leq x_1<a$.  Together with the requirement $0\leq x_1<x_2\leq 1$, \eref{eq:x2p1p2} also implies that
	\begin{align}\label{eq:x1x2order}
	0\leq x_1\leq \frac{a-b}{1-a}<a<\frac{b}{a}\leq x_2\leq 1.
	\end{align}
	
	By virtue of \eref{eq:x2p1p2} we can further deduce that
	\begin{equation}
	p_1 \sqrt{x_1}+ p_2 \sqrt{x_2}=h(x_1,a,b),
	\end{equation}
	where
	\begin{equation}
	h(x_1,a,b)\!:=\!\frac{(b-a^2)\sqrt{x_1}+(a-x_1)\sqrt{(a-x_1)(b-ax_1)}}{x_1^2-2ax_1+b},
	\end{equation}
which is continuous in $x_1$ for $0\leq  x_1\leq (a-b)/(1-a)$ given the assumption $0< a^2< b\leq a< 1$ together with \esref{eq:hdenominator}{eq:x1x2order}.

When $0< x_1\leq (a-b)/(1-a)$, the derivative of $h(x_1,a,b)$ over $x_1$ reads	
	\begin{align}
&\frac{\partial h(x_1,a,b)}{\partial x_1}=\frac{(b-a^2)(y_2\sqrt{x_2}-y_1\sqrt{x_1}\,)}{2\sqrt{x_1x_2}(x_1^2-2ax_1+b)^2}\nonumber\\
&=\frac{(b-a^2)(x_1^2-2ax_1+b)}{2(a-x_1)\sqrt{x_1x_2}(y_2\sqrt{x_2}+y_1\sqrt{x_1}\,)}>0,\label{eq:hderivative}
\end{align}	
where $x_2$ is given in \eref{eq:x2p1p2} and
\begin{align}
y_1=3b-2ax_1-x_1^2,\quad y_2=b+2ax_1-3x_1^2.
\end{align}
The second equality in \eref{eq:hderivative} follows from the facts that $y_1, y_2>0$ and 
\begin{align}
&(y_2\sqrt{x_2}-y_1\sqrt{x_1}\,)(y_2\sqrt{x_2}+y_1\sqrt{x_1}\,)=(y_2^2x_2-y_1^2x_1)\nonumber\\
&=\frac{(x_1^2-2ax_1+b)^3}{a-x_1}.
\end{align}
\Eref{eq:hderivative} implies that $h(x_1,a,b)$ is strictly monotonically increasing in $x_1$ for $0\leq  x_1\leq (a-b)/(1-a)$ given that $h(x_1,a,b)$ is continuous in $x_1$ in this range. In conjunction with \eref{eq:x1x2order} we can now deduce that
	\begin{align}\label{eq:halfmomentLBUBproof}
	a\sqrt{\frac{a}{b}}&=h(0,a,b)\leq  p_1 \sqrt{x_1}+ p_2 \sqrt{x_2} \nonumber\\
	&\leq h\Bigl(\frac{a-b}{1-a},a,b\Bigr)= \zeta(a,b),
	\end{align}
	which confirms \eref{eq:halfmomentLBUB2v}.

	If \eref{eq:LBcondition2v} holds, then it is easy to verify that the lower bound in \eref{eq:halfmomentLBUB2v} [identical to the lower bound in \eref{eq:halfmomentLBUBproof}] is saturated. Conversely, if the lower bound in \eref{eq:halfmomentLBUB2v}  is saturated, then $x_1=0$ given that $h(x_1,a,b)$ is strictly increasing in $x_1$ for $0\leq x_1\leq (a-b)/(1-a)$, so \eref{eq:LBcondition2v} holds according to \eref{eq:x2p1p2}. 
	By a similar reasoning, the upper bound in \eref{eq:halfmomentLBUB2v} is saturated iff \eref{eq:UBcondition2v} holds. 
\end{proof}

\section{Proof of \lref{lem:FidPOVMbasic}}

\begin{proof}[Proof of \lref{lem:FidPOVMbasic}]
	\Eref{eq:fidUnitaryInvariance} in the lemma follows from \eref{eq:FidPOVM2}  and the following equality,
	\begin{equation}
	\Bigl\|\caQ\Bigl(U^{\otimes N}A {U^\dag}^{\otimes N}\Bigr)\Bigr\|=\|\caQ(A)\|,
	\end{equation}
	which holds for any positive operator $A$ on $\caH^{\otimes N}$. 
	
	\Eref{eq:fidTensorTrivialPOVM} follows from \eref{eq:FidPOVM} and the following equation
	\begin{align}
	&\bigl\|\tilde{\caQ}\bigl(A\otimes 1^{\otimes k}\bigr)\bigr\|=\max_{\rho}\tr\bigl[P_{N+k+1} \bigl(A\otimes 1^{\otimes k}\otimes \rho\bigr)\bigr]\nonumber\\
	&=\frac{D_{N+k+1}}{D_{N+1}}\max_{\rho}\tr[P_{N+1} (A\otimes  \rho)]=\frac{D_{N+k+1}}{D_{N+1}}\|\tilde{\caQ}(A)\|,
	\end{align}
	which holds for any positive operator  $A$ on $\caH^{\otimes N}$. 	
	
	To prove \eref{eq:fidCoarseGraining}, let  $\scrA=\{A_j\}_j$ and $\scrB=\{B_k\}_k$.
	By assumption $\scrA$ is a coarse graining of $\scrB$, which means $A_j=\sum_k\Lambda_{jk}B_k$ for some stochastic matrix $\Lambda$. Therefore,
	\begin{align}
	&\sum_j \|\caQ(A_j)\|=\sum_j \Biggl\|\caQ\Biggl(\sum_k \Lambda_{jk}B_k\Biggr)\Biggr\|\nonumber\\
	&=\sum_j \Biggl\|\sum_k \Lambda_{jk} \caQ(B_k)\Biggr\|\leq \sum_j \sum_k \Lambda_{jk}\| \caQ(B_k)\| \nonumber\\
	&=\sum_k \| \caQ(B_k)\|,
	\end{align}
	which implies \eref{eq:fidCoarseGraining} in view of \eref{eq:FidPOVM2}. 
	
The equality $F(\scrA\otimes \scrC)=F(\scrC\otimes \scrA)$ in 	\eref{eq:fidProductComponent} follows from \eref{eq:FidPOVM2} and the following equality
\begin{align}
\|\caQ(A\otimes C)\|=\|\caQ(C\otimes A)\|,
\end{align}
which holds for any positive operator $A$ on $\caH^{\otimes N}$  and any positive operator  $C$ on $\caH^{\otimes k}$. The inequality  in 	\eref{eq:fidProductComponent} follows from 
\esref{eq:fidTensorTrivialPOVM}{eq:fidCoarseGraining} together with the following facts,
	\begin{equation}
	\scrA\otimes \scrI^{\otimes k}\preceq \scrA\otimes \scrC,\quad \scrI^{\otimes N}\otimes \scrC\preceq \scrA\otimes \scrC.
	\end{equation}
\end{proof}

\section{Proofs of  \lsref{lem:OneCopyFidLBUB}-\ref{lem:FidprodPOVM}}

\begin{proof}[Proof of \lref{lem:OneCopyFidLBUB}]
	According to \eref{eq:FidPOVM2} with $N=1$ and \eref{eq:QAnorm} we have
	\begin{align}
	F(\scrA)&=\sum_j \frac{\|\caQ(A_j)\|}{2D_{2}}= \frac{1}{d(d+1)}\sum_j [\lsp \tr(A_j)+\|A_j\|\lsp]\nonumber\\
	&= \frac{1}{d+1}+\frac{1}{d(d+1)}\sum_j \|A_j\|,
	\end{align}
	which confirms the equality in \eref{eq:OneCopyFidUB} in \lref{lem:OneCopyFidLBUB}. 
	Here the third equality follows from the normalization condition $\sum_j A_j=1$. The lower bound in \eref{eq:OneCopyFidUB} follows from the inequality $\|A_j\|\geq \tr(A_j)/d$, which is saturated iff  $A_j$ is proportional to the identity; so the lower bound is saturated iff $\scrA$ is trivial. The upper bound in \eref{eq:OneCopyFidUB}	 follows from the inequality $\|A_j\|\leq \tr(A_j)$, which is saturated iff $A_j$ is rank 1;
	so the upper bound is saturated iff $\scrA$ is rank 1. 	
\end{proof}

\begin{proof}[Proof of \lref{lem:Sym2ProjNorm}]
	Since both sides in \eref{eq:SymProjNorm} are homogeneous in $A$ and $B$, to prove this equation we can assume that $\tr(A)=\tr(B)=1$ without loss of generality, which means $f=\tr(AB)$. 
	
	If $A$ and $B$ are rank 1, then we can further assume that $\caH$ has dimension 2 without loss of generality. In this case $A$ and $B$ have the form 
	\begin{align}
	A=\frac{1+\vec{a}\cdot \vec{\sigma}}{2},\quad B=\frac{1+\vec{b}\cdot \vec{\sigma}}{2},
	\end{align}
	where  
	$\vec{a}$ and $\vec{b}$ are two real unit vectors in dimension 3, and $\vec{\sigma}=(\sigma_x,\sigma_y,\sigma_z)$ is the vector composed of the three Pauli operators.  So   we have $f=\tr(AB)=(1+\vec{a}\cdot \vec{b})/2$, and  \eref{eq:QABtrace1} in the main text implies that
	\begin{align}
	\caQ(A\otimes B)&=3+\vec{a}\cdot \vec{b}+(\vec{a}+\vec{b})\cdot \vec{\sigma}, \label{eq:QABqubit}\\
	\|\caQ(A\otimes B)\|&=3+\vec{a}\cdot \vec{b}+|\vec{a}+\vec{b}|\nonumber\\
	&=3+\vec{a}\cdot \vec{b}+\sqrt{2+2\vec{a}\cdot \vec{b}}\nonumber\\
	&=2\bigl(1+f+\sqrt{f}\lsp\bigr),\label{eq:QABNormRank1}
	\end{align}
	which confirms \eref{eq:SymProjNorm} with equality.

	In general, suppose $A$ and $B$ have convex decompositions $A=\sum_j \lambda_j |\psi_j\>\<\psi_j|$ and $B=\sum_k\mu_k |\varphi_k\>\<\varphi_k|$, respectively, where $\lambda_j,\mu_k>0$; let  $f_{jk}=|\<\psi_j|\varphi_k\>|^2$. Then  we have 
	\begin{align}\label{eq:MeanOverlap}
	\!\!\sum_j\lambda_j=\sum_k\mu_k=1,\quad \sum_{j,k}\lambda_j\mu_kf_{jk}=\tr(AB)=f. 
	\end{align}
	Therefore,
	\begin{align}
	&\|\caQ(A\otimes B)\|\leq \sum_{j,k}\lambda_j\mu_k
	\|\caQ(|\psi_j\>\<\psi_j|\otimes |\varphi_k\>\<\varphi_k|)\|\nonumber\\
	&= \sum_{j,k}2\lambda_j\mu_k \bigl(1+f_{jk}+\sqrt{f_{jk}}\lsp\bigr)
	\leq 2\bigl(1+f+\sqrt{f}\lsp\bigr),\label{eq:TensorNormUB}
	\end{align}
	which confirms \eref{eq:SymProjNorm}. Here the first inequality follows from the triangle inequality for the operator norm; the second inequality follows from \eref{eq:MeanOverlap} and the (strict) concavity of the square-root function. 
	
	If  both $A$ and $B$ are rank 1, then the upper bound in \eref{eq:SymProjNorm} is saturated according to \eref{eq:QABNormRank1}. If $A$ and $B$ have orthogonal supports and one of them is rank 1, then the upper bound is also  saturated, which can be verified by virtue of \eref{eq:QAB} or \eqref{eq:QABtrace1}.

Conversely, if  the upper bound in \eref{eq:SymProjNorm} is saturated, then the two inequalities in \eref{eq:TensorNormUB} are saturated.  The saturation of the  second inequality implies that $f_{jk}=f$ for all $j,k$. Since this result holds irrespective of the convex decompositions of $A$ and $B$, it follows that $|\<\psi|\varphi\>|^2=f$ for any  pure state $|\psi\>$ in the support of $A$ and any pure state $|\varphi\>$ in the support of $B$. Therefore, $A$ and $B$ are rank~1, or $A$ and $B$ have orthogonal supports. In the former case, the upper bound in \eref{eq:SymProjNorm} is indeed saturated according to \eref{eq:QABNormRank1}. In the later case, we have $f=0$ and the upper bound in \eref{eq:SymProjNorm} is equal to 2. In addition,  from \eref{eq:QAB} or \eqref{eq:QABtrace1} we can deduce that 
	\begin{align}
	\caQ(A\otimes B)=1+A+B,
	\end{align}
	which implies that
	\begin{align}
	\|\caQ(A\otimes B)\|&=1+\max\{\|A\|, \|B\|\}\leq 2\nonumber\\
	&=2\bigl(1+f+\sqrt{f}\lsp\bigr).
	\end{align}
	The upper bound is saturated iff $\|A\|=1$ or $\|B\|=1$, which means $A$ is rank 1 or $B$ is rank 1 given the assumption $\tr(A)=\tr(B)=1$. This observation completes the proof of \lref{lem:Sym2ProjNorm}. 
\end{proof}

\begin{proof}[Proof of \lref{lem:FidprodPOVM}]
In general $\scrA$ and $\scrB$ can be expressed as  $\scrA=\{A_j\}_j$ and $\scrB=\{B_k\}_k$, where $A_j,B_k\neq 0$ and $\sum_j A_j=\sum_k B_k=1$. Let  $a_j=\tr(A_j)$, $b_k=\tr(B_k)$, and $f_{jk}=\tr(A_j B_k)/(a_jb_k)$. Then we have
\begin{equation}\label{eq:abfNormalization}
\begin{gathered}
\sum_j a_j=\sum_j \tr(A_j)=\sum_k b_k=\sum_k \tr(B_k)=d,\\
\sum_{j,k} a_jb_kf_{jk}=\sum_{j,k}\tr(A_j B_k)=d. 
\end{gathered}
\end{equation}
In addition, \lref{lem:Sym2ProjNorm} yields the inequality
	\begin{equation}\label{eq:QjknormUB}
\|\caQ(A_j\otimes B_k)\|\leq 2a_j b_k\bigl(1+f_{jk}+\sqrt{f_{jk}}\lsp\bigr).
\end{equation}
By virtue of \eref{eq:FidPOVM2} now  we can deduce that
	\begin{align}
	&F(\scrA\otimes \scrB)= \frac{1}{d(d+1)(d+2)}\sum_{j,k}\|\caQ(A_j\otimes B_k)\|\nonumber\\
	&\leq\frac{2}{d(d+1)(d+2)}\sum_{j,k} a_j b_k\bigl(1+f_{jk}+\sqrt{f_{jk}}\lsp\bigr)\nonumber\\
	&=\frac{2d(d+1)+2\Phi_{1/2}(\scrA,\scrB)}{d(d+1)(d+2)},	\label{eq:FidProdUBgenProof}
	\end{align}	
	which confirms the upper bound in \eref{eq:FidProdUBgen}. Here the last equality follows from the definition of $\Phi_{1/2}(\scrA,\scrB)$ in \eref{eq:POVMcrossFP} and the normalization conditions in \eref{eq:abfNormalization}.

	If $\scrA$ and $\scrB$ are rank-1 POVMs, then the upper bound in \eref{eq:QjknormUB} is saturated for each pair $j,k$, so the upper bound in \eref{eq:FidProdUBgenProof} [identical to the upper bound in \eref{eq:FidProdUBgen}] is saturated. 
	
	Conversely, if the upper bound in \eref{eq:FidProdUBgen} is saturated, then the upper bound in \eref{eq:QjknormUB} is saturated for each pair $j,k$. Suppose on the contrary that $\scrA$ is not rank~1; then it contains  a POVM element, say $A_1$, of rank at least~2.   According to \lref{lem:Sym2ProjNorm}, all POVM elements in $\scrB$ are rank 1 and are orthogonal to $A_1$, which is impossible. Therefore, $\scrA$ is rank 1, and so is $\scrB$ by the same token. In a word, the upper bound in \eref{eq:FidProdUBgen} is saturated iff $\scrA$ and $\scrB$ are rank 1. 
\end{proof}

\section{Proofs of \lsref{lem:SymProjNormSum}-\ref{lem:FidCoarsegraining}}

\begin{proof}[Proof of \lref{lem:SymProjNormSum}]The inequality in \eref{eq:SymProjNormSum} follows from the triangle inequality for the operator norm given that $\caQ(A\otimes B)=\caQ(A\otimes B_1 )+\caQ(A\otimes B_2)$. The inequality is saturated iff there exists a ket $|\psi\rangle\in \caH$ such that 
	\begin{equation}\label{eq:SymProjNormSumProof}
	\begin{aligned}
	\caQ(A\otimes B_1 )|\psi\rangle &=\|\caQ(A\otimes B_1 )\||\psi\>,  \\
	\caQ(A\otimes B_2 )|\psi\rangle &=\|\caQ(A\otimes B_2 )\||\psi\>. 
	\end{aligned}
	\end{equation}
	
	If condition~1 in \lref{lem:SymProjNormSum} holds, that is, $B_2$ is proportional to $B_1$, then $\caQ(A\otimes B_2 )$ is proportional to $\caQ(A\otimes B_1 )$, so the inequality in \eref{eq:SymProjNormSum} is saturated. If condition~2 or 3 in \lref{lem:SymProjNormSum} holds, then the inequality is also saturated according to \eref{eq:QAB} [cf. \eref{eq:QABqubit}] and the above observation. Note that the eigenspace associated with the maximum eigenvalue of $\caQ(A\otimes B_j)$ is two-fold degenerate when $A$ and $B_j$ are orthogonal, but nondegenerate otherwise.

	Next, we suppose that none of the three conditions in \lref{lem:SymProjNormSum} holds. Then $B_2$ is not proportional to $B_1$, and $A$ is not orthogonal to one of the two operators $B_1, B_2$. 	
	If $A, B_1, B_2$ are not supported in any common two-dimensional subspace of $\caH$, then there does not exist any ket $|\psi\rangle$ that satisfies \eref{eq:SymProjNormSumProof}, so the inequality in \eref{eq:SymProjNormSum} is not saturated. 

	If  $A, B_1, B_2$ are all supported in a common two-dimensional subspace of $\caH$,  then $A$ is orthogonal to neither $B_1$ nor $B_2$. Therefore, the eigenspace of $\caQ(A\otimes B_j)$ associated with the maximum eigenvalue is nondegenerate for $j=1,2$ according to \eref{eq:QAB} [cf. \eref{eq:QABqubit}], and there does not exist any ket $|\psi\rangle$ that satisfies \eref{eq:SymProjNormSumProof} either, so the inequality in \eref{eq:SymProjNormSum} is not saturated. This observation completes the proof of \lref{lem:SymProjNormSum}.
\end{proof}

\begin{proof}[Proof of \lref{lem:SymProjNormSumPOVM}]
	The inequality 	in \eref{eq:SymProjNormSumPOVM} follows from the triangle inequality for the operator norm.  If $B_1, B_2,\ldots, B_n$ are mutually orthogonal
	and they commute with all POVM elements in $\scrA$, then it is easy to verify that the inequality is saturated according to \eref{eq:QAB}.

	To prove the converse, let $C=B-B_1-B_2$; then 
	\begin{align}
	&\sum_{j=1}^m\|\caQ(A_j \otimes B)\|\nonumber\\
	&\leq 
	\sum_{j=1}^m\|\caQ(A_j \otimes(B_1+B_2))\|+\sum_{j=1}^m\|\caQ(A_j \otimes C\|\nonumber\\
	&\leq
	\sum_{j=1}^m \sum_{k=1}^n\|\caQ(A_j\otimes B_k )\|.
	\end{align}
	If the final upper bound is saturated, then 
	\begin{align}
	\!\!\sum_{j=1}^m\|\caQ(A_j \otimes(B_1+B_2))\|
	=\sum_{j=1}^m \sum_{k=1}^2\| \caQ(A_j\otimes B_k )\|,
	\end{align}
	which implies that
	\begin{align}\label{eq:SymProjNormSumPOVMproof} 
	\|\caQ(A_j \otimes(B_1+B_2))\|=
	\sum_{k=1}^2 \|\caQ(A_j\otimes B_k )\|
	\end{align}
	for $j=1,2,\ldots, m$.

	Let $\caV_{12}$ be the two-dimensional subspace of $\caH$ that contains the supports of $B_1$ and $B_2$, and let $\caV_{12}^\bot$ be its orthogonal complement; let $P_{12}$ be the orthogonal projector onto $\caV_{12}$. Then \eref{eq:SymProjNormSumPOVMproof} implies that each $A_j$ is either  supported in $\caV_{12}$ or supported in $\caV_{12}^\bot$ according to \lref{lem:SymProjNormSum}, given that $B_2$ is not proportional to $B_1$. In the first case, each $A_j$ is orthogonal to either $B_1$ or $B_2$; in the second case, each $A_j$ 
	is orthogonal to and commutes with both $B_1$ and $B_2$.   Denote by $\scrA_1$ ($\scrA_2$) the set of POVM elements in $\scrA$ that belong to the first (second) category. Then $\scrA_1$ is a POVM on $\caV_{12}$, while $\scrA_2$ is a POVM on $\caV_{12}^\bot$, that is,
	\begin{align}
	\sum_{A_j\in \scrA_1 }A_j =P_{12},\quad \sum_{A_j\in \scrA_2}A_j =1-P_{12}.
	\end{align}
	If $B_1$ is not orthogonal to $B_2$, then the first equality cannot hold. This contradiction shows that $B_1$ and $B_2$ are orthogonal. Consequently, each $A_j$ in $\scrA_1$ is  proportional to either $B_2$ or $B_1$ and commutes with both $B_1$ and $B_2$. In conjunction with the above conclusion, we conclude that all POVM elements in $\scrA$ commute with $B_1$ and $B_2$.

	The above reasoning is still applicable if $B_1, B_2$ are replaced by $B_j, B_k$ with $1\leq j<k\leq n$.
	In this way we can deduce that $B_1, B_2,\ldots, B_n$ are mutually orthogonal and they commute with all  POVM elements in $\scrA$. According to \lref{lem:commutePOVMrank1Proj},
	$\scrA$ contains $n$ rank-1 projectors that are proportional to $B_1, B_2, \ldots, B_n$, respectively. 	
\end{proof}

\begin{proof}[Proof of \lref{lem:FidCoarsegraining}]
	The inequality in \eref{eq:ABAC} follows from \lref{lem:FidPOVMbasic}, which also implies that the inequality is saturated if $\scrB$ is equivalent to $\scrC$. 
	
	To prove the converse, we can assume that $\scrA$ and $\scrC$ are simple rank-1 POVMs without loss of generality. Suppose  $\scrB=\{B_j\}_{j=1}^n$ and $\scrC=\{C_k\}_{k=1}^o$. By assumption $B_j$ can be expressed as 
	\begin{align}
	B_j=\sum_{k=1}^o\Lambda_{jk} C_k,
	\end{align} 
	where $\Lambda$ is a stochastic matrix, which means $\Lambda_{jk}\geq0$ and $\sum_{j=1}^n \Lambda_{jk}=1$ for $k=1,2,\ldots,o$. Therefore, 
	\begin{align}\label{eq:ABjCk}
	\sum_{A\in \scrA} \|\caQ(A \otimes B_j)\|\leq 
	\sum_{A\in \scrA}\sum_{k=1}^o \Lambda_{jk} \|\caQ(A\otimes C_k )\|
	\end{align}
	for  $j=1,2,\ldots,n$. This equation implies that
	\begin{align}\label{eq:ABjCkSum}
	\sum_{A\in \scrA} \sum_{j=1}^n\|\caQ(A \otimes B_j)\|\leq 
	\sum_{A\in \scrA}\sum_{k=1}^o \|\caQ(A\otimes C_k )\|,
	\end{align}
	which is equivalent to \eref{eq:ABAC} according to \eref{eq:FidPOVM2}. 
	
	If  $\scrB$ is not equivalent to $\scrC$, then $\scrB$ is not rank 1 according to \lref{lem:EquivalentPOVM} and \pref{pro:CoarseGraining} and thus
	contains at least one POVM element, say $B_1$,  with rank at least 2. Consequently, the corresponding inequality in \eref{eq:ABjCk} is strict according to \lref{lem:SymProjNormSumPOVM} given that  $\scrA$ is irreducible or $\scrC$ contains no two POVM elements that are mutually orthogonal
by assumption. Therefore, the inequality in \eref{eq:ABjCkSum} is also strict and we have
	\begin{align}
	F(\scrA\otimes \scrB)< F(\scrA\otimes \scrC)
	\end{align}
according to \eref{eq:FidPOVM2}.	
In conclusion,   the inequality in \eref{eq:ABAC} is saturated iff $\scrB$ is equivalent to $\scrC$. 
\end{proof}

\section{\label{asec:OptCollProof}Proof of \thref{thm:nCopyFidUB}}
\subsection{Main proof}

\begin{proof}[Proof of \thref{thm:nCopyFidUB}]
By virtue of  \eref{eq:FidPOVM} in the main text and \lref{lem:QnormLBUB} below, the lower bound in  \eref{eq:nCopyFidLBUB} can be proved as follows,
	\begin{align}
F(\scrA)&=\frac{1}{D_{N+1}}\sum_j \|\tilde{\caQ}(A_j)\|\nonumber\\
&\geq \frac{N+d}{d(N+1)D_{N+1}} \sum_j \tr(P_N A_j)\nonumber\\
&= \frac{\tr(P_N)}{dD_{N}}
=\frac{D_N}{dD_N}
= \frac{1}{d},
\end{align}	
where the inequality is saturated iff $\tilde{\caQ}(A_j)$ for each $j$ is proportional to the identity. 	
	
The upper bound in  \eref{eq:nCopyFidLBUB} can
 be derived from \eref{eq:FidPOVM}  and \lref{lem:QnormLBUB} as follows,	
	\begin{align}
	F(\scrA)&=\frac{1}{D_{N+1}}\sum_j \|\tilde{\caQ}(A_j)\|\leq \frac{1}{D_{N+1}}\sum_j \tr(P_N A_j)\nonumber\\
	&= \frac{\tr(P_N)}{D_{N+1}}=\frac{D_N}{D_{N+1}}
	= \frac{N+1}{N+d}.
	\end{align}
Here the inequality is saturated iff $P_N A_j P_N$ for each $j$ is  proportional to the $N$th tensor power of a pure state. This observation completes the proof of \thref{thm:nCopyFidUB}. 
\end{proof}

\subsection{Auxiliary lemmas}
\begin{lem}\label{lem:QnormLBUB}
Any  positive operator $A$ on $\caH^{\otimes N}$ satisfies 
\begin{gather}
\frac{N+d}{d(N+1)} \tr(P_N A )\leq 	\|\tilde{\caQ}(A)\|\leq \tr(P_N A),   \label{eq:QtildeNormLBUB}
\end{gather}
and the lower bound
is saturated iff $\tilde{\caQ}(A)$ is proportional to the identity, while the upper  bound is saturated iff $P_N A P_N$ is proportional to the $N$th tensor power of a pure  state. 
\end{lem}

\begin{proof}
The lower bound in \eref{eq:QtildeNormLBUB} can be proved as follows,
\begin{align}
\|\tilde{\caQ}(A)\|&\geq\frac{1}{d}\tr[\tilde{\caQ}(A)]=\frac{1}{d}\tr[P_{N+1}(A\otimes 1)]\nonumber\\
&=\frac{D_{N+1}}{d D_N}\tr(P_N A)=\frac{N+d}{d(N+1)}\tr(P_N A).
\end{align}
Here the inequality is saturated iff $\tilde{\caQ}(A)$ is proportional to the identity. 

The upper bound in \eref{eq:QtildeNormLBUB} and the saturation condition follow from  \lref{lem:ArhoSymProj} and the  equation below,
\begin{align}
\|\tilde{\caQ}(A)\|&=\max_{\rho}\,[P_{N+1} (A\otimes \rho)],
\end{align}	
where the maximization is taken over all normalized pure states. 
\end{proof}

\begin{lem}\label{lem:ArhoSymProj}
	Suppose $\rho$ is a pure sate on $\caH$ and 	 $A$ is a positive operator acting on $\caH^{\otimes N}$. Then 
	\begin{align}\label{eq:ArhoSymProj}
	\tr[P_{N+1} (A\otimes \rho)]\leq \tr(P_N A)	,	
	\end{align}
	and the upper bound is saturated iff $P_N A P_N$ is proportional to  the $N$th tensor power of $\rho$. 
\end{lem}
\begin{proof}
	The upper bound in \eref{eq:ArhoSymProj} can be derived as follows,
	\begin{align}\label{eq:ArhoSymProjProof}
	&\tr[P_{N+1} (A\otimes \rho)]=\tr[P_{N+1} (A\otimes \rho)P_{N+1}]\nonumber\\
	&=\tr[P_{N+1} (P_N AP_N\otimes \rho)P_{N+1}]\nonumber\\
	&\leq \tr(P_N AP_N)\tr(\rho)=\tr(P_N A).
	\end{align}
	If $P_N A P_N$ is proportional to the $N$th tensor power of $\rho$, then $P_N AP_N\otimes \rho$ is supported in the symmetric subspace in $\caH^{\otimes (N+1)}$, so the inequality in \eref{eq:ArhoSymProjProof} is saturated, which means the upper bound in \eref{eq:ArhoSymProj} is saturated.
	
	Conversely, if the inequality in \eref{eq:ArhoSymProjProof} is saturated, then $P_N AP_N\otimes \rho$ is supported in the symmetric subspace in $\caH^{\otimes (N+1)}$. Let $\tr_{\bar{j}}(\cdot)$ denote the partial trace over all the parties except for party $j$. Then all the operators
	$\tr_{\bar{j}}(P_N AP_N)$ for $j=1,2,\ldots, N$ have the same support as $\rho$, which implies that $P_N A P_N$ is proportional to the $N$th tensor power of $\rho$.
\end{proof}

\section{\label{asec:Fid2copyProof}Proofs of \thref{thm:Fid2copyProjSIC} and  \crsref{cor:12copyProj}-\ref{cor:SICidentical}}

\begin{proof}[Proof of \thref{thm:Fid2copyProjSIC}]
\thref{thm:Fid2copyProjSIC} is a simple corollary of \lsref{lem:POVMFP} and \ref{lem:Fid2copyFP}. 	
\end{proof}

\begin{proof}[Proof of \crref{cor:12copyProj}]
If $\scrA$ is a rank-1 projective measurement, then  $F(\scrA^{\otimes 2})=F(\scrA)=2/(d+1)$ according to \eref{eq:FidProjective}. Conversely, if $F(\scrA^{\otimes 2})=F(\scrA)=2/(d+1)$, then $\scrA$ is rank 1 according to \lref{lem:OneCopyFidLBUB}. Furthermore, \thref{thm:Fid2copyProjSIC} implies that $\scrA$ is a rank-1 projective measurement. 
\end{proof}

\begin{proof}[Proof of \crref{cor:rank1Projidentical}]
	If $\scrA$ and $\scrB$ are identical  rank-1 projective measurements up to relabeling, then  \eref{eq:FidProjective} implies  that
	\begin{align}
F(\scrA\otimes \scrB)=F(\scrB)=F(\scrA)= \frac{2}{d+1}.
	\end{align}
Conversely, if this equation holds, then $\scrA$ and $\scrB$ are rank 1 according to \lref{lem:OneCopyFidLBUB}. By virtue of \lref{lem:FidprodPOVM} we can further deduce that	
\begin{align}
&\frac{2}{d+1}=F(\scrA\otimes \scrB)=\frac{2d(d+1)+2\Phi_{1/2}(\scrA,\scrB)}{d(d+1)(d+2)},
\end{align}	
which implies that $\Phi_{1/2}(\scrA,\scrB)=d$. Therefore, $\scrA$ and $\scrB$  are identical rank-1 projective measurements up to relabeling according to \lref{lem:POVMcrossFP}. 
\end{proof}

\begin{proof}[Proof of \crref{cor:Fid2copyUBmix}]
	If $\scrA$ is a rank-1 POVM, then the conclusions in \crref{cor:Fid2copyUBmix} follow from  \thref{thm:Fid2copyProjSIC}.

	If $\scrA$ is not a rank-1 POVM, let $\scrB$ be a rank-1 POVM that refines $\scrA$. Then $\scrB^{\otimes 2}$ is a refinement of $\scrA^{\otimes 2}$, so \lref{lem:FidPOVMbasic} and \thref{thm:Fid2copyProjSIC} imply that
	\begin{align}
	F(\scrA^{\otimes 2})\leq F(\scrB^{\otimes 2})&\leq F_2^{\iid}.
	\end{align}
	To attain  the upper bound $F_2^{\iid}$, any rank-1 refinement of $\scrA$ must  attain  the  bound $F_2^{\iid}$  and is thus equivalent  to a SIC according to \thref{thm:Fid2copyProjSIC}.  However, this condition  is impossible when $\scrA$ is not rank 1. So the upper bound $F_2^{\iid}$
	cannot be attained except for rank-1 POVMs. This observation completes the proof of \crref{cor:Fid2copyUBmix}. 
\end{proof}

\begin{proof}[Proof of \crref{cor:fidgSIC}]
Let
\begin{align}
\scrA=\bigcup_{r=1}^g \frac{\scrA_r}{g};
\end{align}
then $\scrA$ is a POVM on $\caH$. By virtue of  \eref{eq:FidPOVM2} and \crref{cor:Fid2copyUBmix} we can deduce that 
	\begin{align}\label{eq:fidgSICproof}
\sum_{r,s} F(\scrA_r\otimes \scrA_s)=g^2 F(\scrA^{\otimes 2})\leq g^2 F_2^{\iid},
\end{align}
which confirms \eref{eq:fidgSIC}. If $\scrA_1,\scrA_2,\ldots, \scrA_g$ are  identical SICs up to relabeling, then $\scrA$ is equivalent to a SIC, so the inequality is saturated according to \eref{eq:FidSIC}.

Conversely, if the inequality in \eref{eq:fidgSICproof} is saturated, then $\scrA$ is equivalent to a SIC by \crref{cor:Fid2copyUBmix}. Let $\scrA'$ be  a simple POVM that is equivalent to $\scrA$; then $\scrA'$ is a SIC. Suppose $\scrA'$ is composed of the POVM elements $A_1, A_2, \ldots, A_{d^2}$. Then these POVM elements form a basis in the operator space, and a linear combination of them is equal to the identity operator iff all the coefficients are equal to 1. In addition, each POVM element in $\scrA_r$ for $r=1,2,\ldots, g$ is proportional to a POVM element in $\scrA'$. Since $\scrA_r$ is a simple POVM by assumption, it follows that $\scrA_r$ is identical to $\scrA'$ up to relabeling for $r=1,2,\ldots, g$, so $\scrA_1,\scrA_2,\ldots, \scrA_g$ are  identical SICs up to relabeling.
\end{proof}

\begin{proof}[Proof of \crref{cor:SICidentical}]
\Crref{cor:SICidentical} would follow from 
\crref{cor:fidgSIC} if the condition $F(\scrA\otimes \scrB)=F(\scrA^{\otimes 2})= F_2^{\iid}$ is replaced by
\begin{align}\label{eq:SICidenticalProof}
F(\scrA\otimes \scrB)=F(\scrB^{\otimes 2})=F(\scrA^{\otimes 2})= F_2^{\iid}.
\end{align}
 Without this stronger condition we need to devise a different proof.

If $\scrA$ and $\scrB$ are identical SICs up to relabeling, then \eref{eq:SICidenticalProof} holds according to \eref{eq:FidSIC} and \thref{thm:Fid2copyProjSIC}.

Conversely, if $F(\scrA\otimes \scrB)=F(\scrA^{\otimes 2})= F_2^{\iid}$, then $\scrA$ is a SIC by  \crref{cor:Fid2copyUBmix} given that $\scrA$ and $\scrB$  are simple POVMs. So $\scrA$ is constructed from a 2-design. Let $\scrB'$ be any simple rank-1 POVM that refines
$\scrB$; then by virtue of \lref{lem:FidPOVMbasic} we can deduce that
\begin{align}
F(\scrA\otimes \scrB')\geq F(\scrA\otimes \scrB)= F_2^{\iid},
\end{align}
which implies that
\begin{align}
\Phi_{1/2}(\scrA,\scrB')\geq 	1+(d-1)\sqrt{d+1}
\end{align}
by \lref{lem:FidprodPOVM}. Thanks to
 \lref{lem:POVMcrossFP2design}, this inequality is necessarily saturated;
moreover, $\scrB'$ is identical to the SIC $\scrA$ up to relabeling. The last conclusion holds for any  simple rank-1  POVM $\scrB'$ that refines  $\scrB$, which is impossible if $\scrB$ is not rank 1. Therefore, $\scrB$ is rank 1; moreover,
$\scrA$ and $\scrB$ are identical SICs up to relabeling. 
\end{proof}
 
\section{Proofs of \lref{lem:FidSepUB}, \thref{thm:FidProdUB}, and \crsref{cor:MaxFidMinProdPOVM}-\ref{cor:MaxFMUB}}

 \begin{proof}[Proof of \lref{lem:FidSepUB}]
 Thanks to \lref{lem:FidPOVMbasic}, to prove the inequality $F(\scrA)\leq F_2^{\sep}$ in \lref{lem:FidSepUB}, we can assume that $\scrA$ is a rank-1 POVM. Then each POVM element $A_j$ of $\scrA$  is a tensor product of two rank-1 positive operators. 
Let  $a_j=\tr(A_j)$ and $f_j=\tr(WA_j)/a_j$. Then the normalization condition $\sum_j A_j=1^{\otimes 2}$ implies that
 \begin{align}
  \sum_j a_j&=\sum _j\tr(A_j)=d^2,\\
 \sum_j a_j f_j&=\sum_j \tr[(2P_2-1)A_j]=d,
 \end{align} 
from which we can deduce that
\begin{align}
\sum_j a_j \sqrt{f_j}&\leq d\sqrt{d}. \label{eq:wjpjSum}
 \end{align}	
In addition, \lref{lem:Sym2ProjNorm} implies that
 	\begin{align}
 	\|\caQ(A_j)\|=2a_j\bigl(1+f_j+\sqrt{f_j}\lsp\bigr).
 	\end{align}
 	By virtue of \eref{eq:FidPOVM2}  we can now deduce that
 	\begin{align}
 	F(\scrA)&=\frac{1}{d(d+1)(d+2)}\sum_j \|\caQ(A_j)\|
 	\nonumber\\
 	&=\frac{2}{d(d+1)(d+2)}\sum_j a_j\bigl(1+f_j+\sqrt{f_j}\lsp\bigr)\nonumber\\
 	&\leq 
 	\frac{2(d+1+\sqrt{d}\lsp)}{(d+1)(d+2)}=F_2^{\sep}.
 	\end{align}
 	Here the inequality is saturated iff the inequality in \eref{eq:wjpjSum} is saturated, which is the case iff each $f_j$ is equal to $1/d$,  that is, $d\tr(WA_j)=\tr(A_j)$. 
 \end{proof}

\begin{proof}[Proof of \thref{thm:FidProdUB}]
	The inequality $F(\scrA\otimes \scrB)\leq F_2^{\sep} $ in the theorem follows from 	\lref{lem:FidSepUB} (cf. \lsref{lem:POVMcrossFP} and \ref{lem:FidprodPOVM}). When $\scrA$ and $\scrB$ are rank 1,  \lref{lem:FidSepUB} also implies that 
	the upper bound  is saturated iff  $\scrA$ and $\scrB$ are  MU. To complete the proof, it suffices to prove that the upper bound can never be saturated if $\scrA$ or $\scrB$ is not rank 1. 
	
	According to \lsref{lem:FidPOVMbasic} and \ref{lem:FidSepUB}, to saturate the inequality $F(\scrA\otimes \scrB)\leq F_2^{\sep} $, any rank-1 refinement of $\scrA$ and any rank-1 refinement of $\scrB$ are necessarily MU. So the fidelity between any state vector in the support of each POVM element in $\scrA$ and any state vector in the support of  each POVM element in $\scrB$ is equal to $1/d$. However, this condition can never hold if $\scrA$ or $\scrB$ is not rank-1.  This observation completes the proof of \thref{thm:FidProdUB}.   
\end{proof}

\begin{proof}[Proof of \crref{cor:MaxFidMinProdPOVM}]
	According to \thref{thm:FidProdUB}, the assumption $F(\scrA\otimes \scrB)= F_2^{\sep}$ implies that $\scrA$ and $\scrB$ are rank-1 and MU. 
	Therefore, both $\scrA$ and $\scrB$ have at least $d$ POVM elements according to \lref{lem:rank1Projective}, which implies that $\scrA\otimes \scrB$ has at least $d^2$ POVM elements. Obviously, the lower bound is saturated 
	if $\scrA$ and $\scrB$ are MU rank-1 projective measurements. Conversely, if $\scrA\otimes \scrB$ has $d^2$ POVM elements, then  both $\scrA$ and $\scrB$ have  $d$ POVM elements and are thus  rank-1 projective measurements according to \lref{lem:rank1Projective}. In addition, the two measurements are MU according to \thref{thm:FidProdUB} as mentioned above.
\end{proof}

\begin{proof}[Proof of \crref{cor:FidToMUB}]
If $\scrA$ and $\scrB$ 
are MU rank-1 projective measurements, then we have $F(\scrA\otimes \scrB)= F_2^{\sep} $ according to \eref{eq:FidMUB} and 
$F(\scrA^{\otimes 2})=F(\scrB^{\otimes 2})=2/(d+1)$ according to \eref{eq:FidProjective}.

 Conversely, if $F(\scrA\otimes \scrB)= F_2^{\sep} $, then $\scrA$ and $\scrB$ are MU rank-1  POVMs
according to  \thref{thm:FidProdUB}. If in addition  $F(\scrA^{\otimes 2})=F(\scrB^{\otimes 2})=2/(d+1)$, then  $\scrA$ and $\scrB$ are rank-1 projective measurements according to \thref{thm:Fid2copyProjSIC}. Therefore, $\scrA$ and $\scrB$ are MU rank-1 projective measurements if both conditions hold.
\end{proof}

\begin{proof}[Proof of \crref{cor:MaxFMUB}]
	The upper bound in \eref{eq:FidCMUB} follows from the inequality $F(\scrA\otimes \scrB)\leq  F_2^{\sep} $ in  \thref{thm:FidProdUB}. If $\scrA_1,\scrA_2,\ldots, \scrA_g$ are  MU rank-1 POVMs, then the upper bound is saturated according to \thref{thm:FidProdUB} again.

	Conversely, if the upper bound in \eref{eq:FidCMUB} is saturated, then  \thref{thm:FidProdUB} implies that
	\begin{equation}
	F(\scrA_r \otimes \scrA_s)= F_2^{\sep}\quad \forall r\neq s;
	\end{equation}
	moreover, $\scrA_1,\scrA_2,\ldots, \scrA_g$ are $g$ MU rank-1 POVMs. If in addition  these POVMs are simple and $g=d+1$, then
	 $\scrA_1,\scrA_2,\ldots, \scrA_{d+1}$ are rank-1 projective measurements according to \thref{thm:MUPOVM} and thus form a CMUMs. 
\end{proof}

\section{Proof of \thref{thm:FidCommPOVM}}

\begin{proof}[Proof of \thref{thm:FidCommPOVM}] Let $\scrA=\{A_j\}_j$ and $\scrB=\{B_k\}_k$ be two commuting POVMs. Then the estimation fidelity $F(\scrA\otimes \scrB)$ can be computed by virtue of \eref{eq:FidPOVM2}, with the result
\begin{align}
&F(\scrA\otimes \scrB)=\frac{1}{d(d+1)(d+2)}\sum_{j,k} \|\caQ(A_j\otimes B_k)\|\nonumber\\
&\leq \frac{1}{d(d+1)(d+2)}\sum_{j}2(d+2)\tr(A_j)=\frac{2}{d+1},
\end{align}
which confirms \thref{thm:FidCommPOVM}. 
Here the inequality follows from \lref{lem:QABkSumUB} below. 
\end{proof}

\begin{lem}\label{lem:QABkSumUB}
Suppose $A$ is a positive operator on $\caH$ and $\scrB$ is a POVM that commutes with $A$. Then  
\begin{align}\label{eq:QABsumUB}
\sum_{B\in \scrB} \|\caQ(A\otimes B)\|\leq 2(d+2)\tr(A).
\end{align}
\end{lem}

\begin{proof}
To prove \eref{eq:QABsumUB}, first consider the case in which $A$ is a projector of rank $r$. By assumption each POVM element $B$ in $\scrB$ commutes with $A$ and is thus block diagonal with respect to the eigenspaces of $A$. 

If  $\scrB$ is rank 1, then $\|B\|=\tr(B)$ and  $B$ is supported in the support of $A$ or in its orthogonal complement. So  either the condition  $AB=B$ or the condition $AB=0$ holds. According to \eref{eq:QAB} we have

\begin{align}
&\caQ(A\otimes B)=\tr(A)\tr(B)+\tr(AB)
+\tr(B)A\nonumber\\
&\quad +\tr(A)B
+AB+BA\nonumber\\
&=\begin{cases}
(r+1)\tr(B)+\tr(B)A
+(r+2)B &\mbox{if } AB=B,\\
r\tr(B)+\tr(B)A
+rB & \mbox{if } AB=0,
\end{cases}
\end{align}	
which implies that
	\begin{align}
	\!\!\|\caQ(A\otimes B)\|=\begin{cases}
	2(r+2)\tr(B) &\mbox{if } AB=B,\\
	2r\tr(B) & \mbox{if } AB=0.
	\end{cases}
	\end{align}
	Therefore,
	\begin{align}
	&\sum_{B\in \scrB} \|\caQ(A\otimes B)\|=2r\sum_{B\in \scrB}\tr(B)+4\sum_{B\in \scrB|AB=B}\tr(B)
	\nonumber\\
	&=2rd+4r=2(d+2)\tr(A),
	\end{align}
	which confirms \eref{eq:QABsumUB} with equality. Here the second equality follows from the following facts
\begin{align}
\sum_{B\in \scrB} B=1,\quad \sum_{B\in \scrB|AB=B}B=A.
\end{align}

If $\scrB$ is not rank 1, then we can find a rank-1 refinement $\scrB'$ of $\scrB$ that  commutes with $A$ given that $\scrB$ commutes with $A$. 
Therefore,
\begin{align}
\sum_{B\in \scrB} \|\caQ(A\otimes B)\|\leq \sum_{B\in \scrB'} \|\caQ(A\otimes B)\|= 2(d+2)\tr(A),
\end{align}
which confirms \eref{eq:QABsumUB} again. 

Next, we drop the assumption that $A$ is a projector. Let $A=\sum_j\lambda_j P_j$ be the  spectral decomposition of $A$, where $\lambda_j$ are distinct eigenvalues, and $P_j$ are the corresponding eigenprojectors. By assumption each $B\in \scrB$ commutes with $A$ and thus also commutes with all eigenprojectors $P_j$. Therefore,
	\begin{align}
	&\sum_{B\in \scrB} \|\caQ(A\otimes B)\|= \sum_{B\in \scrB} \Biggl\|\sum_j \lambda_j\caQ(P_j\otimes B)\Biggr\|\nonumber\\
	&
	\leq \sum_{j} \sum_{B\in \scrB}\lambda_j \|\caQ(P_j\otimes B)\|
	\leq \sum_{j}2\lambda_j (d+2)\tr(P_j)
	\nonumber\\
	&= 2(d+2)\tr(A),
	\end{align}
	which confirms \eref{eq:QABsumUB} and completes the proof of \lref{lem:QABkSumUB}. 
\end{proof}

\section{Proof of \thref{thm:FidJMPOVM}}
\begin{proof}[Proof of \thref{thm:FidJMPOVM}]
	Let $\scrC$ be a simple POVM that refines both  $\scrA$ and $\scrB$. Then $\scrC\otimes \scrC$ is a refinement of $\scrC\otimes \scrB$, which is in turn a refinement of $\scrA\otimes \scrB$, so
	\begin{align}\label{eq:FidJMPOVMproof}
	F(\scrA\otimes \scrB)\leq F(\scrC\otimes \scrB)\leq F(\scrC\otimes \scrC)\leq  F_2^{\iid},
	\end{align}
	where the  first two inequalities follow from \lref{lem:FidPOVMbasic} and the 
	third inequality follows from \crref{cor:Fid2copyUBmix}. If $\scrA$ and $\scrB$ are identical SICs up to relabeling, then the upper bound is saturated according to \eref{eq:FidSIC} (cf. \thref{thm:Fid2copyProjSIC} and \crref{cor:Fid2copyUBmix}).

	Conversely, if the final upper bound in \eref{eq:FidJMPOVMproof} is saturated, then all  three inequalities are saturated simultaneously, that is, 
	\begin{align}
	F(\scrA\otimes \scrB)= F(\scrC\otimes \scrB)= F(\scrC\otimes \scrC)=  F_2^{\iid}. 
	\end{align}
	Here the third equality implies that  $\scrC$ is a  SIC according to \crref{cor:Fid2copyUBmix} and is thus irreducible. Then the second equality implies that  $\scrB$ is equivalent to $\scrC$ according to \lref{lem:FidCoarsegraining}  and is thus also irreducible. Finally, the first equality implies that 
	$\scrA$ is equivalent to $\scrC$ according to \lref{lem:FidCoarsegraining}  gain.  If both $\scrA$ and $\scrB$ are simple, then they are both identical to $\scrC$ up to relabeling according to \lref{lem:SimplePOVM} and are thus identical SICs up to relabeling.
\end{proof}

\section{\label{sec:thmProdPOVMproof}Proof of \thref{thm:FidProdPOVMLB}}

\begin{proof}[Proof of \thref{thm:FidProdPOVMLB}]
The inequality $F(\scrA\otimes \scrB)\geq  F(\scrA)$ and  equality $F(\scrA)=2/(d+1)$ follow from  \lsref{lem:FidPOVMbasic} and \ref{lem:OneCopyFidLBUB}, respectively, given that $\scrA$ is rank 1 by assumption. If $\scrA$ and $\scrB$ commute, then $F(\scrA\otimes \scrB)\leq 2/(d+1)$ by \thref{thm:FidCommPOVM}. So  the inequality  $F(\scrA\otimes \scrB)\geq  F(\scrA)$ is saturated. 
	
To further clarify the saturation condition of the inequality 	 $F(\scrA\otimes \scrB)\geq  F(\scrA)$, note that  $\scrA$ can be expressed as $\scrA=\{A_j\}_j$, where $A_j=a_j|\psi_j\>\<\psi_j|$ with $a_j=\tr(A_j)$, and $\scrB$ can be expressed as $\scrB=\{B_k\}_k$. By virtue of \eref{eq:FidPOVM2} we can deduce that
\begin{align}
	F(\scrA)&=\frac{1}{d(d+1)}\sum_{j}\|Q_{j}\|=\frac{2}{d+1}, \label{eq:FA}\\
	F(\scrA\otimes \scrB)&=\frac{1}{d(d+1)(d+2)}\sum_{j,k}\|Q_{jk}\|,\label{eq:FAB}
	\end{align}
	where
	\begin{align}
	Q_j&=\caQ(A_j)=
	a_j(|\psi_j\>\<\psi_j|+1),\\
	Q_{jk}&=\caQ(A_j\otimes B_k)
	=a_j\bigl[\lsp\tr(B_k)+\<\psi_j|B_k|\psi_j\>+B_k\nonumber\\
	&\quad +\tr(B_k)|\psi_j\>\<\psi_j|+|\psi_j\>\<\psi_j|B_k+B_k|\psi_j\>\<\psi_j|\lsp\bigr].  \label{eq:QAjBk}
	\end{align}
Note that $|\psi_j\>$ is the eigenstate of $Q_j$ corresponding to the largest eigenvalue, which is nondegenerate. 	In addition, $Q_j$ and $Q_{jk}$ satisfy the following relations
	\begin{align} 
	\sum_k Q_{jk}=(d+2)Q_j, \quad \sum_k \|Q_{jk}\|\geq(d+2)\|Q_j\|. \label{eq:QjkIneq}
	\end{align}
By \esref{eq:FA}{eq:FAB}, the inequality $F(\scrA\otimes \scrB)\geq  F(\scrA)$ is saturated iff the inequality in \eref{eq:QjkIneq} is saturated for each $j$.

	If $\scrA$  commutes with  $\scrB$, then  $|\psi_j\>\<\psi_j|$ commutes with  $B_k$, and  $|\psi_j\>$ is an eigenstate of $B_k$. Consequently, $|\psi_j\>$ is an eigenstate of $Q_{jk}$ associated with the largest eigenvalue $\|Q_{jk}\|$, so that
	\begin{align}
	\sum_{k}\|Q_{jk}\|&=\sum_{k}\<\psi_j|Q_{jk}|\psi_j\>=\<\psi_j|(d+2)Q_j|\psi_j\>\nonumber\\
	&=(d+2)\|Q_j\|,
	\end{align}
	which implies that $F(\scrA\otimes \scrB)= F(\scrA)$ given \esref{eq:FA}{eq:FAB}. This derivation reproduces the conclusion derived above, which is based on \thref{thm:FidCommPOVM}.

	Conversely, if the inequality $F(\scrA\otimes \scrB)\geq  F(\scrA)$ saturates, then  the inequality in \eref{eq:QjkIneq} is saturated for each $j$, which implies that 
	\begin{equation}
	\<\psi_j|Q_{jk}|\psi_j\>=\|Q_{jk}\|\quad \forall j,k,
	\end{equation}
	so $|\psi_j\>$ is an eigenstate of $Q_{jk}$ with eigenvalue  $\|Q_{jk}\|$. By virtue of  \eref{eq:QAjBk}, we can further deduce that $|\psi_j\>$ is an eigenstate of $B_k$.
	Therefore, each $|\psi_j\>\<\psi_j|$ commutes with each $B_k$, which means  $\scrA$ commutes with  $\scrB$. 
	
Next, suppose both $\scrA$ and $\scrB$ are simple rank-1 POVMs. By \lref{lem:commutePOVMrank1} and the first conclusion in \thref{thm:FidProdPOVMLB} as proved above, the inequality $F(\scrA\otimes \scrB)\geq F(\scrA)$ is saturated iff $\scrA$ and $\scrB$ are identical rank-1 projective measurements up to relabeling. 
Alternatively, this conclusion follows from \lsref{lem:POVMcrossFP}, \ref{lem:OneCopyFidLBUB}, and \ref{lem:FidprodPOVM} (cf. the proof of \crref{cor:rank1Projidentical}). 
\end{proof}

\bibliography{all_references}

\end{document}